\documentclass[journal, onecolumn, draft]{IEEEtran}

\usepackage{tikz}
\usepackage{color}
\usetikzlibrary{arrows, calc, snakes}
\usepackage{cite}
\usepackage{graphicx}
\usepackage[cmex10]{amsmath}
\usepackage{amsfonts}
\usepackage{amssymb}
\usepackage{algorithm}
\usepackage{algorithmic}
\usepackage{array}
\usepackage{mdwmath}
\usepackage{mdwtab}
\usepackage{eqparbox}
\usepackage{fixltx2e}
\usepackage{url}
\usepackage{enumerate}
\usepackage{multicol}

\usepackage{ntheorem}

\usepackage{comment}
\usepackage{footnote}

\usepackage{dblfloatfix}
\usepackage{subfigure}
\usepackage[justification=centering]{caption}

\usepackage{relsize}

\newtheorem{definition}{Definition}
\newtheorem{theorem}{Theorem}
\newtheorem{lemma}{Lemma}

\newtheorem{property}{Property}

\newcommand*{\permcomb}[4][0mu]{{{}^{#3}\mkern#1#2_{#4}}}
\newcommand*{\comb}[1][-1mu]{\permcomb[#1]{C}}
\newcommand\independent{\protect\mathpalette{\protect\independenT}{\perp}}
\def\independenT#1#2{\mathrel{\rlap{$#1#2$}\mkern2mu{#1#2}}}

\begin{document}

\title{Wiretapped Oblivious Transfer}

\author{
Manoj Mishra, ~\IEEEmembership{Member, ~IEEE,}  Bikash Kumar Dey, ~\IEEEmembership{Member, ~IEEE,}\\ Vinod M. Prabhakaran, ~\IEEEmembership{Member, ~IEEE,} Suhas Diggavi, ~\IEEEmembership{Fellow, ~IEEE}
\thanks{The work of M.~Mishra and B.~K.~Dey is supported in part by the Bharti Centre for Communication, IIT Bombay, a grant from the Department of Science and Technology, Government of India and by the Information Technology Research Academy (ITRA), Government of India under ITRA-Mobile grant ITRA/15(64)/Mobile/USEAADWN/01.  V.~M.~Prabhakaran's work was supported in part by a Ramanujan Fellowship from the Department of Science and Technology, Government of India and by the Information Technology Research Academy (ITRA), Government of India under ITRA-Mobile grant ITRA/15(64)/Mobile/USEAADWN/01. The work of S.~Diggavi was supported in part by NSF grant 1321120. This work was presented in part at the 2014 and 2015 IEEE International Symposia on Information Theory and at the IEEE Information Theory Workshop, Hobart, 2014}
\thanks{M.~Mishra and B.~K.~Dey are with the Department of Electrical Engineering, Indian Institute of Technology Bombay (IIT Bombay), Mumbai, India(email: mmishra,bikash@ee.iitb.ac.in). V.~M.~Prabhakaran is with the School of Technology and Computer Science, Tata Institute of Fundamental Research (TIFR), Mumbai, India(email: vinodmp@tifr.res.in).S.~Diggavi is with the Department of Electrical Engineering, University of California at Los Angeles (UCLA), Los Angeles, USA(email: suhasdiggavi@ucla.edu).}
}

\maketitle

\begin{abstract}

In this paper, we study the problem of obtaining $1$-of-$2$ string oblivious transfer (OT) between users Alice and Bob, in the presence of a passive eavesdropper Eve. The resource enabling OT in our setup is a noisy broadcast channel from Alice to Bob and Eve. Apart from the OT requirements between the users, Eve is not allowed to learn anything about the users' inputs. When Alice and Bob are honest-but-curious and the noisy broadcast channel is made up of two independent binary erasure channels (connecting Alice-Bob and Alice-Eve), we derive the $1$-of-$2$ string OT capacity for both $2$-privacy (when Eve can collude with either Alice or Bob) and $1$-privacy (when no such collusion is allowed). We generalize these capacity results to $1$-of-$N$ string OT and study other variants of this problem. When Alice and/or Bob are malicious, we present a different scheme based on interactive hashing. This scheme is shown to be optimal for certain parameter regimes. We present a new formulation of multiple, simultaneous OTs between Alice-Bob and Alice-Cathy. For this new setup, we present schemes and outer bounds that match in all but one regime of parameters. Finally, we consider the setup where the broadcast channel is made up of a cascade of two independent binary erasure channels (connecting Alice-Bob and Bob-Eve) and $1$-of-$2$ string OT is desired between Alice and Bob with $1$-privacy. For this setup, we derive an upper and lower bound on the $1$-of-$2$ string OT capacity which match in one of two possible parameter regimes.

\end{abstract}



\section{Introduction}
\label{sec:intro}

In secure multiparty computation (MPC), mutually distrusting users wish to communicate with each other in such a way that, at the end of the communication, each user can compute a function of the distributed private inputs without learning any more than what the function output and the private input reveal about other users' inputs and outputs. Applications such as voting, auctions and data-mining, amongst several others \cite{CramDamNiel} illustrate the need for secure MPC in real life. It is well known that information-theoretically (unconditionally) secure computation between two users is not possible in general, when the users have only private randomness and noiseless communication as a resource to enable the computation. A combinatorial characterization of functions that \emph{can} be securely computed was derived in \cite{kushilevitz1992}. However, additional stochastic resources, such as a noisy channel \cite{CrepKilian1988} or distributed sources, can be used to enable two users to compute a function unconditionally securely.

Oblivious Transfer (OT) is a secure two-user computation which has been shown to be a primitive for all two-user secure computation \cite{jkilian1988},\cite{jkilian2000}. That is, if the two users can obtain OT using the resources available to them, then they can securely compute any function of their inputs. In particular, OT can be achieved if the two users have access to a noisy channel. A $1$-of-$2$ string OT is a two-party computation where user Alice's private inputs are two equal-length strings and user Bob's private input is a choice bit. Bob obtains exactly one string of his choice from Alice's strings, without Alice finding out the identity of the string chosen by Bob. If a discrete memoryless channel (DMC) is used as a resource to enable such OT, then the OT capacity of the DMC is the largest rate, i.e. string-length per channel use, that can be obliviously transferred to Bob. Nascimento and Winter \cite{NascWinter2008} characterized source distributions and channels from which non-zero $1$-of-$2$ string OT rates can be obtained. When Alice and Bob are \emph{honest-but-curious}, Ahlswede and Csisz\'ar \cite{ot2007} derived upper bounds on the $1$-of-$2$ string OT capacity both for DMCs and distributed sources. Users are honest-but-curious if they do not deviate from the given protocol but, from whatever they learn during the protocol, they will infer all they can about forbidden information. In contrast, \emph{malicious} users may deviate arbitrarily from the given protocol. When the DMC is a binary erasure channel (BEC) and users are honest-but-curious, Ahlswede and Csisz\'ar \cite{ot2007} presented a protocol which they showed was capacity achieving, establishing that $\min \{ \epsilon, 1 - \epsilon \}$ is the $1$-of-$2$ string OT capacity of a BEC($\epsilon$), where $\epsilon$ is the erasure probability of the channel. They extended these results for a generalized erasure channel (GEC). A GEC is a channel $p_{Y|X}$, with input alphabet $\mathcal{X}$ and output alphabet $\mathcal{Y}$, where $\mathcal{Y}$ can be partitioned as $\mathcal{Y}_e \cup \mathcal{Y}_{\overline{e}}$ such that $p_{Y|X}(y|x)$ does not depend on the input $x \in \mathcal{X}$ whenever $y \in \mathcal{Y}_e$. Specifically, Ahlswede and Csisz\'ar \cite{ot2007} derived lower bounds on $1$-of-$2$ string OT capacity for a GEC and showed that the bounds are tight when the erasure probability of the GEC is at least $1/2$. In a surprising result, Pinto et al. \cite{PintoDowsMorozNasc2011} proved that using a GEC with erasure probability at least $1/2$, any $1$-of-$2$ string OT rate achieved when Alice and Bob are honest-but-curious can also be achieved even if Alice and Bob behave maliciously. This result characterized the $1$-of-$2$ string OT capacity of a GEC, with erasure probability is atleast $1/2$, for malicious users. The achievable scheme presented by Pinto et al. \cite{PintoDowsMorozNasc2011} for establishing this result is a generalization of the scheme presented by Savvides \cite{savvides_thesis}, that uses a BEC($1/2$) and uses the cryptographic primitive of \emph{interactive hashing} (see Appendix~\ref{appndx:interactive_hashing} for the properties and a protocol for interactive hashing) to establish checks that detect malicious behavior. More recently, Dowsley and Nascimento proved \cite{DowsNasc2014-arxiv} that even when the GEC's erasure probability is less than $1/2$, the rate that was shown to be achievable in \cite{ot2007} for honest-but-curious Alice and Bob is also achievable when Alice and Bob are malicious. To the best of our knowledge, characterizing the $1$-of-$2$ string OT capacity for other natural channels such as a binary symmetric channel (BSC) remains open in the two-party setting, even with honest-but-curious users.

In this paper we study a natural extension of the OT setup when there is an eavesdropper Eve, who may wiretap the noisy channel between Alice and Bob. In this case, Eve, who receives partial information about the transmissions, can use it to deduce the private data or outputs of Alice and Bob. The noisy wiretapped channel we consider is a binary erasure broadcast channel whose inputs come from Alice and whose outputs are available to Bob and Eve. For the most part, we consider a  binary erasure broadcast channel which provides independent erasure patterns to Bob and Eve. We also consider the physically degraded binary erasure broadcast channel. In our $3$-party setups, we define two privacy regimes. Privacy against individual parties is referred to as \emph{$1$-privacy}, whereas privacy against any set of $2$ colluding parties is referred to as \emph{$2$-privacy}.


\subsection{Contributions and organization of the paper}

\begin{itemize}

\item When the noisy broadcast channel is made up of two independent BECs (see Figure~\ref{fig:ot_hbc_wtap}) and the users are honest-but-curious, we characterize the $1$-of-$2$ string OT capacity both for $2$-privacy and $1$-privacy (Theorem~\ref{thm:result_2p_1p_hbc_wtap}). We extend these capacity results to $1$-of-$N$ string OT (Theorem~\ref{thm:result_N_hbc_wtap}). Our protocols are natural extensions of the two-party protocols of Ahlswede and Csisz\'ar \cite{ot2007} where we use secret keys between Alice and Bob, secret from Eve, to provide rate-optimal schemes for both privacy regimes. Our converse arguments generalize the converse of Ahlswede and Csisz\'ar \cite{ot2007}.

\item We consider the setup of Figure~\ref{fig:ot_hbc_wtap}, where Alice and Bob may act maliciously during the OT protocol. We derive an expression for an achievable rate under $2$-privacy constraints (Theorem~\ref{thm:result_2p_malicious_wtap}) for this setup. The achievable rate is optimal when $\epsilon_1 \leq 1/2$ and is no more than a factor of $\epsilon_1$ away from the optimal rate when $\epsilon_1 > 1/2$. In a departure from previous protocols \cite{PintoDowsMorozNasc2011},\cite{DowsNasc2014-arxiv} which used interactive hashing primarily to detect the malicious behavior of a user, our protocol uses interactive hashing to generate the secret keys used by Alice and Bob, secret from Eve, to achieve $2$-privacy even with malicious users (for $\epsilon_1 > 1/2$). Using interactive hashing only for checks to detect malicious behavior will not work when $\epsilon_1 > 1/2$, since it is possible for Bob, in collusion with Eve, to pass any such check for uncountably many values of $\epsilon_1, \epsilon_2$. 

\item In a generalization of the setup of Figure~\ref{fig:ot_hbc_wtap}, we consider the setup of Figure~\ref{fig:ot_hbc_indep}, where instead of the eavesdropper, we have a legitimate user Cathy. All users are honest-but-curious. Independent $1$-of-$2$ string OTs are required between Alice-Bob and Alice-Cathy, with $2$-privacy. We derive inner and outer bounds on the rate-region (Theorem~\ref{thm:result_2p_hbc_indep}) for this setup. These bounds match except when $\epsilon_1,\epsilon_2 > 1/2$.

\item When the channel is a physically degraded broadcast channel made up of a cascade of two independent BECs (see Figure~\ref{fig:ot_hbc_degraded}), a BEC($\epsilon_1$) connecting Alice-Bob followed by a BEC($\epsilon_2$) connecting Bob-Eve, we derive upper and lower bounds on the $1$-of-$2$ string OT capacity under $1$-privacy (Theorem~\ref{thm:result_1p_hbc_degraded}), for honest-but-curious users. These bounds match when $\epsilon_1 \leq (1/3) \cdot \epsilon_2(1 - \epsilon_1)$. Unlike the secret key agreement problem, which has a simpler optimal scheme when the broadcast channel is degraded, the scheme for OT turns out to be more complicated than when Bob's and Eve's erasure patterns are independent. This happens because Eve knows more about the legitimate channel's noise process when the channel is degraded. Hiding Bob's choice bit from a more informed Eve is the main novelty of this protocol, compared to the independent erasures case.

\end{itemize}

The main system model we consider is for obtaining OT between honest-but-curious Alice and Bob, in the presence of an eavesdropper Eve. This model is introduced in Section~\ref{sec:prob_statement_wtap}. We consider several variants of this model. Section~\ref{sec:prob_statement_malicious_wtap} defines a variant of the main model, where OT is required when Alice and Bob may be malicious. In Section~\ref{sec:prob_statement_indep}, we generalize the main model by introducing the user Cathy instead of the eavesdropper and requiring independent OTs between Alice-Bob and Alice-Cathy. Section~\ref{sec:prob_statement_degraded} is a variant of the main model where a physically degraded broadcast channel is used as the resource for OT, instead of a broadcast channel providing independent erasure patterns to Bob and Eve considered in all previous models. The problem statement for each model is followed by a statement of the result we derive for that model. These results are proved in Sections~\ref{sec:proofs_0},~\ref{sec:proofs_1},~\ref{sec:proofs_2} and~\ref{sec:proofs_3}. In Section~\ref{sec:conclusions}, we summarize the work presented in this paper. Section~\ref{sec:open_problems} contains a discussion of the open problems related to the present work. The Appendices at the end consist of the supporting results referenced in the main proofs.


\section{Problem Statement and Results}
\label{sec:model_and_problem_statement}


\subsection{Notation}

We will use the capital letter $X$ to denote a random variable, whose alphabet will be specified in the context where $X$ is used. The small letter $x$ will denote a specific realization of $X$. The bold, small letter $\boldsymbol{x}$ will denote a $k$-tuple, where $k$ will be clear from the context in which $\boldsymbol{x}$ is used. The small, indexed letter $x_i$, $i=1,2,\ldots,k$ will denote the $i$th element of $\boldsymbol{x}$. The bold, capital $\boldsymbol{X}$ will denote a random $k$-tuple. Furthermore,

\begin{itemize}
\item $\boldsymbol{x}^i := (x_1, x_2, \ldots, x_i)$
\item Suppose $\boldsymbol{a} \in \{1,2,\ldots,k\}^m$. Then,
    \begin{itemize}
      \item $\boldsymbol{x}|_{\boldsymbol{a}} := (x_{a_i} : i=1,2,\ldots,m)$.
    \end{itemize}
\item $\{\boldsymbol{x}\} := \{ x_i : i=1,2,\ldots,k\}$.
\item Let $A \subset \mathbb{N}$. Then, $(A)$ is the tuple formed by arranging the elements of $A$ in increasing order. That is,
          \begin{itemize}
            \item $(A) := (a_i \in A, i=1,2,\ldots,|A| : \forall i > 1, a_{i-1} < a_i)$
          \end{itemize}
          For example, if $A = \{1,7,3,9,5\}$, then $(A) = (1,3,5,7,9)$.
\item Let $A \subset \{1,2,\ldots,k\}$. Then,
          \begin{itemize}
            \item $\boldsymbol{x}|_A := \boldsymbol{x}|_{(A)}$
          \end{itemize}
        For example, if $\boldsymbol{x} = (a,b,c,d,e,f,g)$ and $A = \{7,2,5\}$, then $\boldsymbol{x}|_A = (b,e,g)$.
\item Suppose $\boldsymbol{y} \in \{ 0,1, \bot \}^k$, where $\bot$ represents an erasure. Then,
    \begin{itemize}
      \item $\#_e(\boldsymbol{y}) := | \{i \in \{1,2,\ldots,k\} : y_i = \bot \}|$.
      \item $\#_{\overline{e}}(\boldsymbol{y}) := | \{i \in \{1,2,\ldots,k\} : y_i \neq \bot \}|$. 
    \end{itemize}
\end{itemize}

For $a \in \mathbb{R}$, $b \in \mathbb{R}^{+}$, we define:
\begin{itemize}
\item $<a> := |a| - \lfloor |a| \rfloor$
\item $\mathcal{N}_{b}(a) := \{\alpha \in \mathbb{R} : |a - \alpha| \leq b\}$
\end{itemize}


\subsection{Oblivious Transfer over a Wiretapped Channel: Honest-but-Curious Model}
\label{sec:prob_statement_wtap}

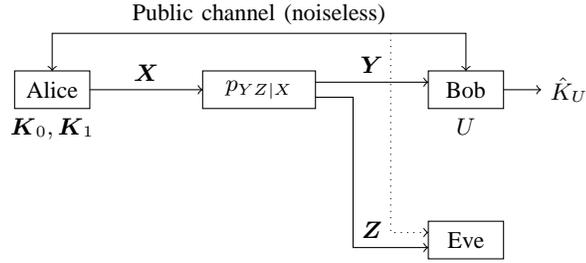
\begin{figure}[h]
\setlength{\unitlength}{1cm}
\centering
\begin{tikzpicture}[trim left, scale=1]

\draw (1,3) rectangle (2,3.5);
\draw (3.5,3) rectangle (5,3.5);
\draw (6.5,1) rectangle (7.5,1.5);
\draw (6.5,3) rectangle (7.5,3.5);

\draw [->] (2,3.25) -- (3.5,3.25);
\draw [->] (5,3.35) -- (6.5,3.35);
\draw [<->] (1.5,3.5) -- (1.5,4) -- (7,4) -- (7,3.5);
\draw [->] (7.5,3.25) -- (8,3.25);
\draw [->] (5,3.15) -- (5.5,3.15) -- (5.5,1.15) -- (6.5,1.15);
\draw [->, dotted] (6,4) |- (6.5,1.35);

\node at (1.5,3.25) {\small{Alice}};
\node at (7,3.25) {\small{Bob}};
\node at (7,1.25) {\small{Eve}};
\node at (4.25, 3.25) {\small{$p_{YZ|X}$}};
\node [below] at (1.5,3) {\small{$\boldsymbol{K}_0,\boldsymbol{K}_1$}};
\node [below] at (7,3) {\small{$U$}};
\node [right] at (8,3.25) {\small{$\hat{K}_U$}};
\node [above] at (2.75,3.25) {\small{$\boldsymbol{X}$}};
\node [above] at (5.75,3.35) {\small{$\boldsymbol{Y}$}};
\node [above] at (5.75,1.15) {\small{$\boldsymbol{Z}$}};
\node [above] at (4.25,4) {\small{Public channel (noiseless)}};

\end{tikzpicture}
\caption{$1$-of-$2$ string OT in presence of an eavesdropper}
\label{fig:ot_hbc_wtap_bcast}
\end{figure}

The setup of Figure~\ref{fig:ot_hbc_wtap_bcast} has two users Alice and Bob and an eavesdropper Eve. Alice and Bob are honest-but-curious. Alice's private data consists of two $m$-bit strings $\boldsymbol{K}_0,\boldsymbol{K}_1$. Bob's private data is his choice bit $U$. The random variables $\boldsymbol{K}_0,\boldsymbol{K}_1,U$ are independent and chosen uniformly at random over their respective alphabets. Alice can communicate with Bob and Eve over a broadcast channel $p_{YZ|X}$, with the output $Y$ available to Bob and the output $Z$ available to Eve. Additionally, Alice and Bob can send messages over a noiseless public channel, with each such message becoming available to Eve as well.

\begin{definition}
\label{defn:protocol_hbc_wtap}
Let $m,n \in \mathbb{N}$. An \emph{$(m,n)$-protocol} is an exchange of messages between Alice and Bob in the setup of Figure~\ref{fig:ot_hbc_wtap_bcast}. Alice's private strings $\boldsymbol{K}_0, \boldsymbol{K}_1$ are $m$-bits each. Alice transmits a bit $X_t$ over the channel at each time instant $t = 1,2,\ldots,n$. Also, before each channel transmission and after the last channel transmission, Alice and Bob take turns to send messages (arbitrarily many but finite number) over the public channel. Any transmission by a user is a function of the user's input, private randomness and all the public messages, channel inputs or channel outputs the user has seen. The rate of the protocol is $r_n := m/n$. Let $\boldsymbol{\Lambda}$ denote the transcript of the public channel at the end of an $(m,n)$-protocol.
\end{definition}

Let the \emph{final views} of Alice, Bob and Eve be, respectively, $V_A$, $V_B$ and $V_E$, where the \emph{final view} of a user is the set of all random variables received and generated by that user over the duration of the protocol. For the present setup:
\begin{align}
V_A & := \{\boldsymbol{K}_0, \boldsymbol{K}_1, \boldsymbol{X}, \boldsymbol{\Lambda}\} \\
V_B & := \{U, \boldsymbol{Y}, \boldsymbol{\Lambda}\} \\
V_E & := \{\boldsymbol{Z}, \boldsymbol{\Lambda}\}
\end{align}
where $\boldsymbol{X} := (X_1,X_2,\ldots,X_n)$, $\boldsymbol{Y} := (Y_1,Y_2,\ldots,Y_n)$ and $\boldsymbol{Z} := (Z_1,Z_2,\ldots,Z_n)$. Bob computes the estimate $\hat{\boldsymbol{K}}_U$ (of the string $\boldsymbol{K}_U$) as a function of its final view $V_B$.

\begin{definition}
\label{defn:ach_rate_2p_hbc_wtap}
$R_{2P}$ is an \emph{achievable $2$-private rate for honest-but-curious users} if there exists a sequence of $(m,n)$-protocols such that $m/n \longrightarrow R_{2P}$ as $n \longrightarrow \infty$ and
\begin{align}
P[\hat{\boldsymbol{K}}_U \neq \boldsymbol{K}_U] & \longrightarrow 0 \label{eqn:ach_2p_wtap_0} \\
I(\boldsymbol{K}_{\overline{U}} ; V_B,V_E) & \longrightarrow 0 \label{eqn:ach_2p_wtap_1} \\
I(U ; V_A,V_E) & \longrightarrow 0 \label{eqn:ach_2p_wtap_2} \\
I(\boldsymbol{K}_0,\boldsymbol{K}_1,U ; V_E) & \longrightarrow 0 \label{eqn:ach_2p_wtap_3}
\end{align} 
as $n \longrightarrow \infty$, where $\overline{U} = U \oplus 1$ and $\oplus$ is the sum modulo-$2$.
\end{definition}

\begin{definition}
\label{defn:ach_rate_1p_hbc_wtap}
$R_{1P}$ is an \emph{achievable $1$-private rate for honest-but-curious users} if there exists a sequence of $(m,n)$-protocols such that $m/n \longrightarrow R_{1P}$ as $n \longrightarrow \infty$ and
\begin{align}
P[\hat{\boldsymbol{K}}_U \neq \boldsymbol{K}_U] & \longrightarrow 0 \label{eqn:ach_1p_wtap_0} \\
I(\boldsymbol{K}_{\overline{U}} ; V_B) & \longrightarrow 0 \label{eqn:ach_1p_wtap_1} \\
I(U ; V_A) & \longrightarrow 0 \label{eqn:ach_1p_wtap_2} \\
I(\boldsymbol{K}_0,\boldsymbol{K}_1,U ; V_E) & \longrightarrow 0 \label{eqn:ach_1p_wtap_3}
\end{align} 
\end{definition}

The \emph{$2$-private capacity} $C_{2P}$ is the supremum of all achievable $2$-private rates for honest-but-curious users and the \emph{$1$-private capacity} $C_{1P}$ is the supremum of all achievable $1$-private rates for honest-but-curious users.


The main result in this section is a characterization of $C_{2P}$ and $C_{1P}$ for the setup of Figure~\ref{fig:ot_hbc_wtap}. The setup of Figure~\ref{fig:ot_hbc_wtap} is a specific case of the setup of Figure~\ref{fig:ot_hbc_wtap_bcast}, where the broadcast channel is made up of two independent binary erasure channels (BECs), namely, BEC($\epsilon_1$) which is a BEC with erasure probability $\epsilon_1$ connecting Alice to Bob and BEC($\epsilon_2$) connecting Alice to Eve.

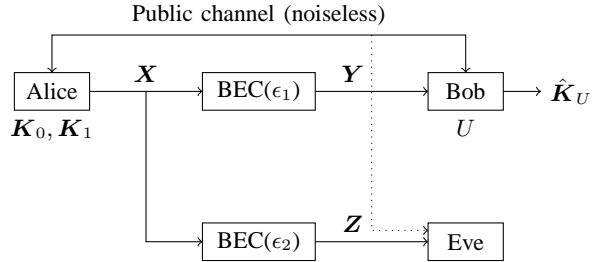
\begin{figure}[h]
\setlength{\unitlength}{1cm}
\centering
\begin{tikzpicture}[trim left, scale=1]

\draw (1,3) rectangle (2,3.5);
\draw (3.5,3) rectangle (5,3.5);
\draw (6.5,3) rectangle (7.5,3.5);
\draw (3.5,1) rectangle (5,1.5);
\draw (6.5,1) rectangle (7.5,1.5);

\draw [->] (5,3.25) -- (6.5,3.25);
\draw [->] (2,3.25) -- (3.5,3.25);
\draw [<->] (1.5,3.5) -- (1.5,4) -- (7,4) -- (7,3.5);
\draw [->] (7.5,3.25) -- (8,3.25);
\draw [->] (2.75,3.25) |- (3.5,1.25);
\draw [->] (5,1.25) -- (6.5,1.25);
\draw [->, dotted] (5.75,4) |- (6.5,1.4);

\node at (1.5,3.25) {\small{Alice}};
\node at (7,3.25) {\small{Bob}};
\node at (7,1.25) {\small{Eve}};
\node at (4.25, 3.25) {\small{BEC($\epsilon_1$)}};
\node at (4.25, 1.25) {\small{BEC($\epsilon_2$)}};
\node [above] at (2.75,3.25) {\small{$\boldsymbol{X}$}};
\node [above] at (5.5,3.25) {\small{$\boldsymbol{Y}$}};
\node [above] at (5.5,1.25) {\small{$\boldsymbol{Z}$}};
\node [below] at (1.5,3) {\small{$\boldsymbol{K}_0,\boldsymbol{K}_1$}};
\node [below] at (7,3) {\small{$U$}};
\node [right] at (8,3.25) {\small{$\hat{\boldsymbol{K}}_U$}};
\node [above] at (4.25,4) {\small{Public channel (noiseless)}};

\end{tikzpicture}
\caption{$1$-of-$2$ string OT using a binary erasure broadcast channel}
\label{fig:ot_hbc_wtap}
\end{figure}

\begin{theorem}[OT capacity for erasure broadcast channel]
\label{thm:result_2p_1p_hbc_wtap}
The $1$-of-$2$ string OT capacity, with $2$-privacy, for honest-but-curious users in the setup of Figure~\ref{fig:ot_hbc_wtap} is
\[ C_{2P} = \epsilon_2 \cdot \min \{\epsilon_1, 1 - \epsilon_1 \}.\]

The $1$-of-$2$ string OT capacity, with $1$-privacy, for honest-but-curious users in the setup of Figure~\ref{fig:ot_hbc_wtap} is
\[C_{1P} = \left\{ \begin{array}{ll} \epsilon_1, & \epsilon_1 < \frac{\epsilon_2}{2} \\ \frac{\epsilon_2}{2}, & \frac{\epsilon_2}{2} \leq \epsilon_1 < \frac{1}{2} \\ \epsilon_2(1 - \epsilon_1), & \frac{1}{2} \leq \epsilon_1 \end{array} \right.\]
\end{theorem}

This result is proved in Section~\ref{sec:proofs_0}.

The above results extend easily to the setup of $1$-of-$N$ ($N \geq 2$) string OT, with honest-but-curious users, in the presence of an eavesdropper (see Figure~\ref{fig:ot_hbc_wtap_N}). The difference with the setup of Figure~\ref{fig:ot_hbc_wtap} is that Alice now has $N$ private strings $\boldsymbol{K}_0,\boldsymbol{K}_1,\ldots,\boldsymbol{K}_{N-1}$ and Bob's choice variable $U$ can take values in $\{ 0,1,\ldots, N-1 \}$. Definition~\ref{defn:protocol_hbc_wtap} still defines a protocol and it is straightforward to extend Definition~\ref{defn:ach_rate_2p_hbc_wtap} and Definition~\ref{defn:ach_rate_1p_hbc_wtap} to define the achievable rates, for the setup of Figure~\ref{fig:ot_hbc_wtap_N}. The following theorem characterizes $C_{2P}$ and $C_{1P}$ for this setup:

\begin{theorem}[1-of-N OT capacity for erasure broadcast channel]
\label{thm:result_N_hbc_wtap}
The $1$-of-$N$ string OT capacity, with $2$-privacy and with $1$-privacy, for honest-but-curious users in the setup of Figure~\ref{fig:ot_hbc_wtap_N} is, respectively,
\begin{align*}
C^N_{2P} & = \epsilon_2 \cdot \min \left\{ \frac{\epsilon_1}{N-1}, 1 - \epsilon_1 \right\} \\
C^N_{1P} & = \left\{ \begin{array}{ll} \frac{\epsilon_1}{N-1}, & \frac{\epsilon_1}{N-1} < \frac{\epsilon_2}{N} \\ \frac{\epsilon_2}{N}, & \frac{\epsilon_2}{N} \leq \frac{\epsilon_1}{N-1} < \frac{1}{N} \\ \epsilon_2(1 - \epsilon_1), & \frac{1}{N} \leq \frac{\epsilon_1}{N-1} \end{array} \right.
\end{align*}
\end{theorem}

\begin{figure}[h]
\setlength{\unitlength}{1cm}
\centering
\begin{tikzpicture}[trim left, scale=1]

\draw (1,3) rectangle (2,3.5);
\draw (3.5,3) rectangle (5,3.5);
\draw (6.5,3) rectangle (7.5,3.5);
\draw (3.5,1) rectangle (5,1.5);
\draw (6.5,1) rectangle (7.5,1.5);

\draw [->] (5,3.25) -- (6.5,3.25);
\draw [->] (2,3.25) -- (3.5,3.25);
\draw [<->] (1.5,3.5) -- (1.5,4) -- (7,4) -- (7,3.5);
\draw [->] (7.5,3.25) -- (8,3.25);
\draw [->] (2.75,3.25) |- (3.5,1.25);
\draw [->] (5,1.25) -- (6.5,1.25);
\draw [->, dotted] (5.75,4) |- (6.5,1.4);

\node at (1.5,3.25) {\small{Alice}};
\node at (7,3.25) {\small{Bob}};
\node at (7,1.25) {\small{Eve}};
\node at (4.25, 3.25) {\small{BEC($\epsilon_1$)}};
\node at (4.25, 1.25) {\small{BEC($\epsilon_2$)}};
\node [above] at (2.75,3.25) {\small{$\boldsymbol{X}$}};
\node [above] at (5.5,3.25) {\small{$\boldsymbol{Y}$}};
\node [above] at (5.5,1.25) {\small{$\boldsymbol{Z}$}};
\node [below] at (1.2,3) {\small{$\boldsymbol{K}_0,\boldsymbol{K}_1,\ldots,\boldsymbol{K}_{N-1}$}};
\node [below] at (7,3) {\small{$U$}};
\node [right] at (8,3.25) {\small{$\hat{\boldsymbol{K}}_U$}};
\node [above] at (4.25,4) {\small{Public channel (noiseless)}};

\end{tikzpicture}
\caption{$1$-of-$N$ string OT using a binary erasure broadcast channel}
\label{fig:ot_hbc_wtap_N}
\end{figure}
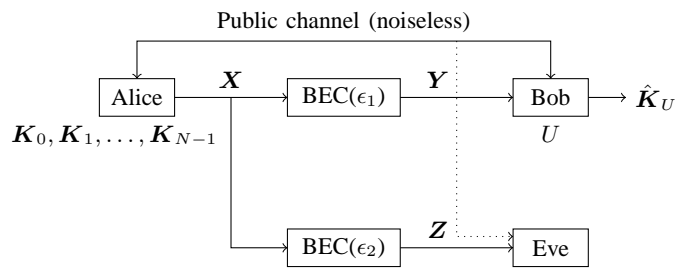

Theorems~\ref{thm:result_2p_1p_hbc_wtap} and~\ref{thm:result_N_hbc_wtap} show that the presence of an eavesdropper reduces the OT capacity by a factor of $\epsilon_2$ for $2$-privacy, compared to the results of Ahlswede and Csisz\'ar \cite{ot2007}. Intuitively, this means that Alice and Bob can get OT with $2$-privacy only over the segment of Alice's transmissions that were erased for Eve. Also, note that for $\epsilon_1 \geq 1/2$, $C_{1P} = C_{2P}$ while for $\epsilon_1 < 1/2$, $C_{1P} > C_{2P}$. By putting $\epsilon_2 = 1$, as one would expect, these capacity results reduce to the $2$-party OT capacity results of Ahlswede and Csisz\'ar \cite{ot2007}.


\subsection{Oblivious Transfer over a Wiretapped Channel: Malicious Model}
\label{sec:prob_statement_malicious_wtap}

The setup is the same as that shown in Figure~\ref{fig:ot_hbc_wtap}. The main difference with the problem definition of Section~\ref{sec:prob_statement_wtap} is that Alice an Bob can be malicious. That is, they can deviate arbitrarily from the protocol. We consider only $2$-privacy in this setup\footnote{A protocol for achieving $1$-privacy in this setup is obtained by only a minor modification (greater privacy amplification) to the two-party protocols presented in \cite{PintoDowsMorozNasc2011}, \cite{DowsNasc2014-arxiv} and is, therefore, being omitted from this work.}. Definition~\ref{defn:protocol_hbc_wtap} defines a protocol and the rate of the protocol for this setup.

\begin{definition}
\label{defn:ach_rate_2p_malicious_wtap}
$R$ is an \emph{achievable $2$-private rate for malicious users} if there exists a sequence of $(m,n)$-protocols such that $m/n \longrightarrow R$ and:

\begin{enumerate}

\item If Alice and Bob are both honest, then the protocol aborts with vanishing probability and (\ref{eqn:ach_2p_wtap_0})-(\ref{eqn:ach_2p_wtap_3}) are satisfied, as $n \longrightarrow \infty$.

\item If Alice is malicious and colludes with Eve and Bob is honest, let $V_n$ be the view of a malicious Alice colluding with Eve at the end of the protocol. Then, $I(U ; V_n) \longrightarrow 0$ as $n \longrightarrow \infty$.

\item If Alice is honest and Bob is malicious and colludes with Eve, let $V_n$ be the view of a malicious Bob colluding with Eve at the end of the protocol. Then, $\min \{ I(\boldsymbol{K}_0 ; V_n),  I(\boldsymbol{K}_1 ; V_n)\} \longrightarrow 0$ as $n \longrightarrow \infty$. 

\end{enumerate}

\end{definition}

\begin{theorem}[An achievable OT rate with malicious users]
\label{thm:result_2p_malicious_wtap}
Any $R < \left\{ \begin{array}{lr} C_{2P}, & \epsilon_1 \leq \frac{1}{2}\\ \epsilon_1 \cdot C_{2P}, & \epsilon_1 > \frac{1}{2} \end{array} \right\}$, where $C_{2P} = \epsilon_2 \cdot \min \{\epsilon_1, 1 - \epsilon_1\}$, is an achievable $2$-private $1$-of-$2$ string OT rate for malicious users in the setup of Figure~\ref{fig:ot_hbc_wtap}.
\end{theorem}

This result is proved in Section~\ref{sec:proofs_1}. Note that $C_{2P}$ is the $2$-private OT capacity when users are honest-but-curious in this setup. Hence, the result shows that the achievable scheme we present is rate-optimal when $\epsilon_1 \leq 1/2$ and no more than a fraction $\epsilon_1$ away from the optimal rate otherwise. The compromise in rate when $\epsilon_1 > 1/2$ happens for the following reason. Our protocol (for the regime where $\epsilon_1 > 1/2$) uses interactive hashing to obtain two subsets of Alice's transmissions over the broadcast channel. Alice converts the non-overlapping parts of these subsets into two secret keys using standard techniques.\footnote{Alice will use these keys to encrypt her strings. For obtaining $2$-privacy, our protocol ensures that one of the keys is secret from Bob and both the keys are secret from Eve.} Losing the overlapping part of both the subsets in this process gives us shorter secret keys, which in turn results in the rate loss by a factor of $\epsilon_1$.


\subsection{Independent Oblivious Transfers over a broadcast channel}
\label{sec:prob_statement_indep}

\begin{figure}[h]
\setlength{\unitlength}{1cm}
\centering
\begin{tikzpicture}[trim left, scale=1]

\draw (1,3) rectangle (2,3.5);
\draw (3.5,3) rectangle (5,3.5);
\draw (6.5,3) rectangle (7.5,3.5);
\draw (3.5,1) rectangle (5,1.5);
\draw (6.5,1) rectangle (7.5,1.5);

\draw [->] (5,3.25) -- (6.5,3.25);
\draw [->] (2,3.25) -- (3.5,3.25);
\draw [<->] (1.5,3.5) -- (1.5,4) -- (7,4) -- (7,3.5);
\draw [->] (7.5,3.25) -- (8,3.25);
\draw [->] (7.5,1.25) -- (8,1.25);
\draw [->] (2.75,3.25) |- (3.5,1.25);
\draw [->] (5,1.25) -- (6.5,1.25);
\draw [<->] (5.75,4) |- (6.5,1.4);

\node at (1.5,3.25) {\small{Alice}};
\node at (7,3.25) {\small{Bob}};
\node at (7,1.25) {\small{Cathy}};
\node at (4.25, 3.25) {\small{BEC($\epsilon_1$)}};
\node at (4.25, 1.25) {\small{BEC($\epsilon_2$)}};
\node [above] at (2.75,3.25) {\small{$\boldsymbol{X}$}};
\node [above] at (5.5,3.25) {\small{$\boldsymbol{Y}$}};
\node [above] at (5.5,1.25) {\small{$\boldsymbol{Z}$}};
\node [below] at (1.5,3) {\small{$\begin{array}{c} \boldsymbol{K}_0,\boldsymbol{K}_1 \\ \boldsymbol{J}_0, \boldsymbol{J}_1 \end{array} $}};
\node [below] at (7,3) {\small{$U$}};
\node [below] at (7,1) {\small{$W$}};
\node [right] at (8,3.25) {\small{$\hat{\boldsymbol{K}}_U$}};
\node [right] at (8,1.25) {\small{$\hat{\boldsymbol{J}}_W$}};
\node [above] at (4.25,4) {\small{Public channel (noiseless)}};

\end{tikzpicture}
\caption{Independent OTs using a binary erasure broadcast channel}
\label{fig:ot_hbc_indep}
\end{figure}
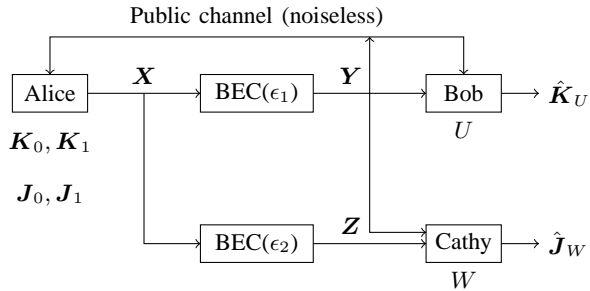

In the setup of Figure~\ref{fig:ot_hbc_indep}, we have three users Alice, Bob and Cathy. Alice is connected to Bob and Cathy by a broadcast channel made up of two independent BECs, a BEC($\epsilon_1$) connecting Alice to Bob and a BEC($\epsilon_2$) connecting Alice to Cathy. In addition, there is a noiseless public channel of unlimited capacity over which the three users can take turns to send messages. Each such public message is received by all the users. Alice's private data consists of two pairs of strings $\boldsymbol{K}_0, \boldsymbol{K}_1$ and $\boldsymbol{J}_0, \boldsymbol{J}_1$. Bob's and Cathy's private data are the choice bits $U$ and $W$ respectively. $\boldsymbol{K}_0, \boldsymbol{K}_1, \boldsymbol{J}_0, \boldsymbol{J}_1, U, W $ are independent and uniform over their respective alphabets. The goal is for Bob to obtain $\boldsymbol{K}_U$ with $2$-privacy and for Cathy to obtain $\boldsymbol{J}_W$ with $2$-privacy\footnote{ 
The BGW algorithm \cite{BGW88} gives a mechanism to achieve $1$-private computations in any $3$-user setting when each user is allowed to use private randomness and there are private links between each pair of users. But achieving a $2$-private computation in a $3$-user setting is, in general, not feasible even with honest-but-curious users.}

\begin{definition}
Let $n,m_B,m_C \in \mathbb{N}$. An $(n,m_B,m_C)$-\emph{protocol} is an exchange of messages between Alice, Bob and Cathy over the setup of Figure~\ref{fig:ot_hbc_indep}. Alice's private data consists of strings $\boldsymbol{K}_0,\boldsymbol{K}_1$ which are $m_B$-bits each and strings $\boldsymbol{J}_0,\boldsymbol{J}_1$ which are $m_C$-bits each. Alice transmits a bit $X_t$ over the broadcast channel at each time instant $t = 1,2,\ldots,n$. In addition, before each such transmission and after the last transmission ($t=n$), the users take turns to send messages on the noiseless public channel over several rounds. The number of rounds maybe random, but finite with probability one. Any transmission by a user is a function of the user's input, private randomness and all the public messages, channel inputs or channel outputs the user has seen. The \emph{rate-pair} $(r_{B,n},r_{C,n})$ of an $(n,m_B,m_C)$-protocol is given by $r_{B,n} := m_B/n$ and $r_{C,n} := m_C/n$. Let $\boldsymbol{\Lambda}$ denote the transcript of the public channel at the end of an $(n,m_B,m_C)$-protocol.
\end{definition}

The \emph{final view} of a user is the collection of all random variables available to the user at the end of the execution of the $(n,m_B,m_C)$-protocol. We denote these for Alice, Bob, and Cathy by $V_A$, $V_B$, and $V_C$, respectively. At the end of an $(n,m_B,m_C)$-protocol, Bob generates an estimate $\hat{\boldsymbol{K}}_U$ of $\boldsymbol{K}_U$ as a function of its final view $V_B$. Similarly, Cathy generates an estimate $\hat{\boldsymbol{J}}_W$ of $\boldsymbol{J}_W$ as a function of its final view $V_C$.

\begin{definition}
\label{defn:ach_rate_2p_indep}
$(R_B,R_C)$ $\in \mathbb{R}^2$ is an \emph{achievable 2-private rate-pair for honest-but-curious users} in the setup of Figure~\ref{fig:ot_hbc_indep} if there exists a sequence of $(n,m_B,m_C)$-protocols with $(r_{B,n},r_{C,n}) \longrightarrow (R_B,R_C)$ as $n \longrightarrow \infty$, such that
\begin{eqnarray}
 P[\hat{\boldsymbol{K}}_U \neq \boldsymbol{K}_U] & \longrightarrow 0 \label{eqn:ach_2p_indep_0}\\
 P[\hat{\boldsymbol{J}}_W \neq \boldsymbol{J}_W] & \longrightarrow 0 \label{eqn:ach_2p_indep_1}\\
 I(\boldsymbol{K}_{\overline{U}}, \boldsymbol{J}_{\overline{W}} ; V_B,V_C) & \longrightarrow 0 \label{eqn:ach_2p_indep_2}\\
 I(U ; V_A,V_C) & \longrightarrow 0 \label{eqn:ach_2p_indep_3} \\
 I(W ; V_A,V_B) & \longrightarrow 0 \label{eqn:ach_2p_indep_4} \\
 I(U,W ; V_A) & \longrightarrow 0 \label{eqn:ach_2p_indep_5} \\
 I(\boldsymbol{K}_0,\boldsymbol{K}_1,U, \boldsymbol{J}_{\overline{W}} ; V_C) & \longrightarrow 0 \label{eqn:ach_2p_indep_6}\\
 I(\boldsymbol{K}_{\overline{U}}, \boldsymbol{J}_0,\boldsymbol{J}_1,W ; V_B) & \longrightarrow 0 \label{eqn:ach_2p_indep_7}
\end{eqnarray}
as $n \longrightarrow \infty$.
\end{definition}

\begin{definition}
The 2-private rate-region $\mathcal{R} \subset \mathbb{R}^2$ for the setup of Figure~\ref{fig:ot_hbc_indep} is
the closure of the set of all achievable $2$-private rate pairs for honest-but-curious users.
\end{definition}


The main results in this section are inner and outer bounds for the $2$-private rate region $\mathcal{R}$, for the setup of Figure~\ref{fig:ot_hbc_indep}, when the users are honest-but-curious\footnote{See Section~\ref{sec:open_problems} for a discussion on considering malicious users in this setup.}. 

\begin{theorem}[OT rate-region for erasure broadcast channel]
\label{thm:result_2p_hbc_indep}
The rate-region $\mathcal{R}$ of independent pairs of $1$-of-$2$ string OTs, with $2$-privacy, for honest-but-curious users in the setup of Figure~\ref{fig:ot_hbc_indep} is such that
\[ \mathcal{R}_{\text{inner}} \subseteq \mathcal{R} \subseteq \mathcal{R}_{\text{outer}}  \]
\end{theorem}

where 
\begin{align*}
\mathcal{R}_{\text{inner}} = \Big\{ (R_B , R_C ) \in \mathbb{R}_+^2 : R_B & \leq \epsilon_2  \min \{ \epsilon_1, 1 - \epsilon_1\},\\
 R_C &\leq \epsilon_1  \min \{ \epsilon_2, 1 - \epsilon_2\}, \\
R_B + R_C  &\leq  \epsilon_2\cdot \min\{\epsilon_1, 1 - \epsilon_1\}  + \epsilon_1\cdot \min\{\epsilon_2, 1 - \epsilon_2\} \\
& \quad - \min\{\epsilon_1, 1 - \epsilon_1\} \cdot \min\{\epsilon_2, 1 - \epsilon_2\} \Big\}
\end{align*}

and 
\begin{align*}
\mathcal{R}_{\text{outer}} = \Big\{
(R_B,R_C) \in \mathbb{R}_+^2 : R_B &\leq  \epsilon_2 \cdot \min \{ \epsilon_1, 1 - \epsilon_1\}, \\
                                 R_C &\leq  \epsilon_1 \cdot \min \{ \epsilon_2, 1 - \epsilon_2\}, \\
                           R_B + R_C &\leq \epsilon_1 \cdot \epsilon_2 \Big\}.
\end{align*}

Theorem~\ref{thm:result_2p_hbc_indep} is proved in Section~\ref{sec:proofs_2}. The regions $\mathcal{R}, \mathcal{R}_{\text{inner}}, \mathcal{R}_{\text{outer}}$ are illustrated for different regimes of $\epsilon_1, \epsilon_2$ in Figure~\ref{fig:rate_region_indep_lt_half}, Figure~\ref{fig:rate_region_indep_bet_half_1} and Figure~\ref{fig:rate_region_indep_gt_half}. The inner and outer bounds match except when $\epsilon_1,\epsilon_2 > 1/2$. The upper bounds on $R_B$ and $R_C$ in the expressions above are the $2$-private OT capacities for Bob and Cathy, respectively, obtained as a consequence of Theorem~\ref{thm:result_2p_1p_hbc_wtap}. The upper bound on the sum-rate is the fraction of Alice's transmissions that are erased for both Bob and Cathy.

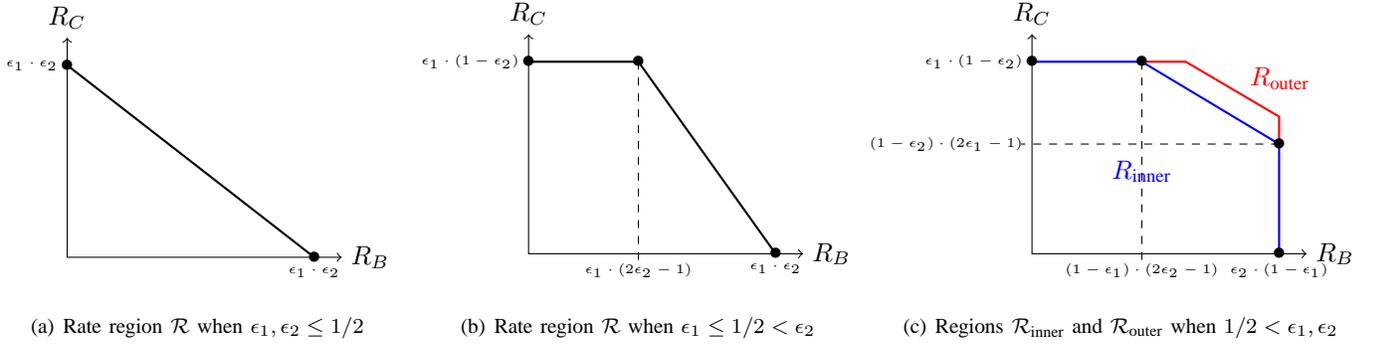
\begin{figure*}
\setlength{\unitlength}{1cm}
\centering

\subfigure[Rate region $\mathcal{R}$ when $\epsilon_1,\epsilon_2 \leq 1/2$]{
\centering
\begin{tikzpicture}[scale=0.73]

\draw [<->] (1,5) |- (6,1);
\node [right] at (6,1) {$R_B$};
\node [above] at (1,5) {$R_C$};

\draw [thick] (1,4.5) -- (5.5,1);
\node at (1,4.5) {\textbullet};
\node at (5.5,1) {\textbullet};
\node [left] at (1,4.5) { \tiny{$\epsilon_1 \cdot \epsilon_2$} };
\node [below] at (5.5,1) { \tiny{$\epsilon_1 \cdot \epsilon_2$} };

\end{tikzpicture}
\label{fig:rate_region_indep_lt_half}
}%
\subfigure[Rate region $\mathcal{R}$ when $\epsilon_1 \leq 1/2 < \epsilon_2$]{
\centering
\hspace{-5pt}
\begin{tikzpicture}[scale=0.73]

\draw [<->] (1,5) |- (6,1);
\node [right] at (6,1) {$R_B$};
\node [above] at (1,5) {$R_C$};

\draw [thick] (1,4.5) -- (3,4.5) -- (5.5,1);
\node at (1,4.5) {\textbullet};
\node at (3,4.5) {\textbullet};
\draw [thin, dashed] (3,4.5) -- (3,0.8);
\node at (5.5,1) {\textbullet};
\node [left] at (1,4.5) { \tiny{$\epsilon_1 \cdot (1 - \epsilon_2)$} };
\node [below] at (3,1) { \tiny{$\epsilon_1 \cdot (2 \epsilon_2 - 1)$} };
\node [below] at (5.5,1) { \tiny{$\epsilon_1 \cdot \epsilon_2$} };

\end{tikzpicture}
\label{fig:rate_region_indep_bet_half_1}
}%
\subfigure[Regions $\mathcal{R}_{\text{inner}}$ and $\mathcal{R}_{\text{outer}}$ when $1/2 < \epsilon_1,\epsilon_2$]{
\centering
\hspace{-10pt}
\begin{tikzpicture}[scale=0.73]

\draw [<->] (1,5) |- (6,1);
\node [right] at (6,1) {$R_B$};
\node [above] at (1,5) {$R_C$};

\draw [thick, color=blue] (1,4.5) -- (3,4.5) -- (5.5,3) -- (5.5,1);
\draw [thick, color=red] (3,4.5) -- (3.8,4.5) -- (5.5,3.5) -- (5.5,3);
\node at (1,4.5) {\textbullet};
\node at (3,4.5) {\textbullet};
\draw [thin, dashed] (3,4.5) -- (3,0.8);
\node at (5.5,3) {\textbullet};
\draw [thin, dashed] (5.5,3) -- (0.8,3);
\node at (5.5,1) {\textbullet};
\node [left] at (1,4.5) { \tiny{$\epsilon_1 \cdot (1 - \epsilon_2)$} };
\node [left] at (1,3) { \tiny{$(1 - \epsilon_2) \cdot (2\epsilon_1 - 1)$} };
\node [below] at (3,1) { \tiny{$(1 - \epsilon_1) \cdot (2 \epsilon_2 - 1)$} };
\node [below] at (5.5,1) { \tiny{$\epsilon_2 \cdot (1 - \epsilon_1)$} };

\node [right] at (4.8,4.2) {\color{red}{$R_{\text{outer}}$}};
\node at (3,2.5) {\color{blue}{$R_{\text{inner}}$}};

\end{tikzpicture}
\label{fig:rate_region_indep_gt_half}
}

\caption{ $\mathcal{R}$, $\mathcal{R}_{\text{inner}}$, $\mathcal{R}_{\text{outer}}$ for all regimes of $\epsilon_1, \epsilon_2$}
\end{figure*}


\subsection{Oblivious Transfer Over a Degraded Wiretapped Channel}
\label{sec:prob_statement_degraded}

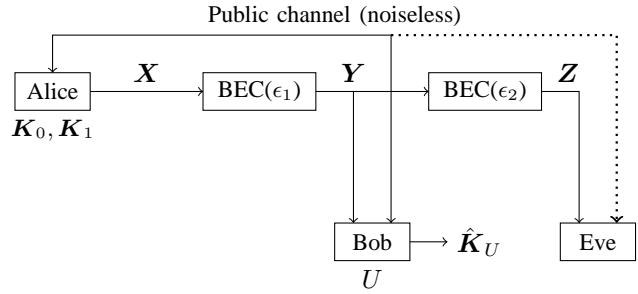
\begin{figure}[h]
\setlength{\unitlength}{1cm}
\centering
\begin{tikzpicture}[scale=1]

\draw (1,3) rectangle (2,3.5);
\draw (3.5,3) rectangle (5,3.5);
\draw (6.5,3) rectangle (8,3.5);
\draw (5.25,1) rectangle (6.25,1.5);
\draw (8.25,1) rectangle (9.25,1.5);

\draw [->] (2,3.25) -- (3.5,3.25);
\draw [->] (5,3.25) -- (6.5,3.25);
\draw [->] (5.5,3.25) -- (5.5,1.5);
\draw [->] (8,3.25) -| (8.5,1.5);

\draw [<-](1.5,3.5) |- (6,4);
\draw [->] (6,4) -- (6,1.5);
\draw [->, dotted, line width=0.3mm] (6,4) -| (9,1.5);
\draw [->] (6.25,1.25) -- (6.75,1.25);

\node at (1.5,3.25) {\small{Alice}};
\node at (5.75,1.25) {\small{Bob}};
\node at (8.75,1.25) {\small{Eve}};
\node at (4.25, 3.25) {\small{BEC($\epsilon_1$)}};
\node at (7.25, 3.25) {\small{BEC($\epsilon_2$)}};
\node [above] at (5.25,4) {\small{Public channel (noiseless)}};
\node [below] at (1.5,3) {\small{$\boldsymbol{K}_0,\boldsymbol{K}_1$}};
\node [below] at (5.75,1) {$U$};
\node [right] at (6.75,1.25) {$\hat{\boldsymbol{K}}_U$};

\node [above] at (2.75,3.25) {$\boldsymbol{X}$};
\node [above] at (5.5,3.25) {$\boldsymbol{Y}$};
\node [above] at (8.35,3.25) {$\boldsymbol{Z}$};

\end{tikzpicture}
\caption{$1$-of-$2$ string OT over a degraded binary erasure broadcast channel}
\label{fig:ot_hbc_degraded}
\end{figure}

In the setup of Figure~\ref{fig:ot_hbc_degraded}, Alice is connected to Bob and Eve by a broadcast channel made up of a cascade of two independent BECs, a BEC($\epsilon_1$) followed by a BEC($\epsilon_2$). Alice and Bob are honest-but-curious. A $1$-of-$2$ string OT is desired between Alice and Bob, with $1$-privacy\footnote{We suspect that no positive $2$-private OT rate can be achieved in this setup, though our brief attempt to prove this has not been successful. The problem of obtaining OT when users can behave maliciously in this setup appears to require newer techniques and has been deferred to a future study.}. Definition~\ref{defn:protocol_hbc_wtap} and Definition~\ref{defn:ach_rate_1p_hbc_wtap} define a protocol and an achievable rate, respectively, for this setup.


\begin{theorem}[OT capacity bounds for degraded erasure broadcast channel]
\label{thm:result_1p_hbc_degraded}
The $1$-of-$2$ string OT capacity with $1$-privacy, $C_{1P}$, for honest-but-curious users in the setup of Figure~\ref{fig:ot_hbc_degraded}, is such that
\begin{equation*}
\min \left\{\frac{1}{3}\epsilon_2(1 - \epsilon_1), \epsilon_1 \right\} \leq C_{1P} \leq \min\{\epsilon_2(1 - \epsilon_1), \epsilon_1 \}.
\end{equation*}
\end{theorem}

This result is proved in Section~\ref{sec:proofs_3}. The upper and lower bounds in this result match when $\epsilon_1 \leq \frac{1}{3}\epsilon_2(1 - \epsilon_1)$. Unlike the previous setups where Bob and Eve/Cathy receive independent erasure patterns, Eve here has more knowledge of the noise process in the channel connecting Alice and Bob. Specifically, Eve knows that Bob's erasure pattern is a subset of the erasure pattern she observes. This makes it harder to guarantee privacy for Bob against Eve. 


\section{Oblivious transfer over a wiretapped channel in the honest-but-curious model : Proof of Theorem~\ref{thm:result_2p_1p_hbc_wtap}}
\label{sec:proofs_0}


\subsection{2-privacy : Achievability}

For the achievability part of our proof, we describe a protocol (Protocol~\ref{protocol:C2P}) which is a natural extension of the two-party protocol of Ahlswede and Csisz\'ar \cite{ot2007} for achieving OT between Alice and Bob using a BEC($\epsilon_1$). Our extension is designed to achieve OT in the presence of Eve (see Figure~\ref{fig:ot_hbc_wtap}), with $2$-privacy. For a sequence of Protocol~\ref{protocol:C2P} instances of rate $r < C_{2P}$, we show that (\ref{eqn:ach_2p_wtap_0})-(\ref{eqn:ach_2p_wtap_3}) hold. This establishes that any $r < C_{2P}$ is an achievable $2$-private rate in the setup of Figure~\ref{fig:ot_hbc_wtap}. We begin by introducing the two-party OT protocol of Ahlswede and Csisz\'ar \cite{ot2007}.


\subsubsection{Two-party OT protocol \cite{ot2007}}
\label{sec:two_party_ot_protocol}

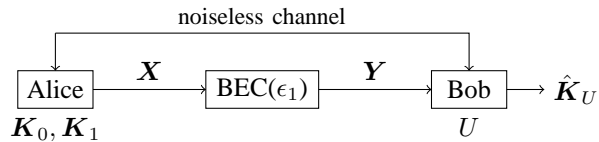
\begin{figure}[h]
\setlength{\unitlength}{1cm}
\centering
\begin{tikzpicture}[trim left, scale=1]

\draw (1,3) rectangle (2,3.5);
\draw (6.5,3) rectangle (7.5,3.5);

\draw [->] (1.5,4) -- (1.5,3.5);
\draw [->] (1.5,4) -| (7,3.5);
\draw [->] (7.5,3.25) -- (8,3.25);

\draw (3.5,3) rectangle (5,3.5);
\draw [->] (2,3.25) -- (3.5,3.25 );
\draw [->] (5,3.25) -- (6.5, 3.25);
\node at (4.25, 3.25) {BEC($\epsilon_1$)};

\node at (1.5,3.25) {Alice};
\node at (7,3.25) {Bob};
\node [below] at (1.5,3) {$\boldsymbol{K}_0,\boldsymbol{K}_1$};
\node [below] at (7,3) {$U$};
\node [above] at (4.25,4) {\small{noiseless channel}};
\node [above] at (2.75, 3.25) {$\boldsymbol{X}$};
\node [above] at (5.75, 3.25) {$\boldsymbol{Y}$};
\node [right] at (8,3.25) {$\hat{\boldsymbol{K}}_U$};

\end{tikzpicture}
\caption{Setup for two-party OT using a BEC($\epsilon_1$)}
\label{fig:ot_twoparty_dmc}
\end{figure}

Consider the two-party setup of Figure~\ref{fig:ot_twoparty_dmc}. The OT capacity in this setup is $\min\{\epsilon_1, 1 - \epsilon_1\}$ \cite{ot2007}. Let $r < \min\{\epsilon_1, 1 - \epsilon_1\}$. Ahlswede and Csisz\'ar's \cite{ot2007} protocol begins with Alice transmitting a sequence $\boldsymbol{X}$ of $n$ i.i.d., Bernoulli($1/2$) bits over the channel. Bob receives the channel output $\boldsymbol{Y}$. Let $E$ be the set of all indices at which $\boldsymbol{Y}$ is erased and $\overline{E}$ that of all indices at which $\boldsymbol{Y}$ is unerased. If $|E| < nr$ or $|\overline{E}| < nr$, Bob aborts the protocol since he does not have sufficient erasures or non-erasures to run the protocol. From $\overline{E}$, Bob picks a subset $L_U$ of cardinality $nr$, uniformly at random. From $E$, Bob picks a subset $L_{\overline{U}}$ of the same cardinality as $L_U$, also uniformly at random and then sends the sets $L_0,L_1$ over the public channel. Alice cannot infer which of the sets $L_0,L_1$ consists of indices at which $\boldsymbol{Y}$ was erased, since the channel acts independently on each input bit. As a result, Alice does not learn $U$ when it receives $L_0,L_1$ from Bob. Alice uses $\boldsymbol{X}|_{L_0}, \boldsymbol{X}|_{L_1}$ as the keys to encrypt its strings and send these encrypted strings to Bob. That is, Alice sends $\boldsymbol{K}_0 \oplus \boldsymbol{X}|_{L_0}, \boldsymbol{K}_1 \oplus \boldsymbol{X}|_{L_1}$ to Bob over the public channel. Bob knows only the key $\boldsymbol{X}|_{L_U}$ (since $\boldsymbol{Y}|_{L_U}$ is unerased) and knows nothing about the key $\boldsymbol{X}|_{L_{\overline{U}}}$ (since $\boldsymbol{Y}|_{L_{\overline{U}}}$ is erased). As a result, Bob learns $\boldsymbol{K}_U$ but learns nothing about $\boldsymbol{K}_{\overline{U}}$.


\subsubsection{Protocol for $2$-privacy in the wiretapped model}
\label{sec:2p_ot_protocol}

The above two-party protocol, as is, is insufficient for privacy against Eve in our wiretapped channel model (see Figure~\ref{fig:ot_hbc_wtap}). This is because the keys $\boldsymbol{X}|_{L_0}, \boldsymbol{X}|_{L_1}$ that Alice uses in the two-party protocol are both partially known to Eve, since Eve knows an independently erased version $\boldsymbol{Z}$ of $\boldsymbol{X}$. As a result, when Alice sends $\boldsymbol{K}_0 \oplus \boldsymbol{X}|_{L_0}, \boldsymbol{K}_1 \oplus \boldsymbol{X}|_{L_1}$ to Bob over the public channel, Eve learns approximately a fraction $(1 - \epsilon_2)$ of both of Alice's strings, violating (\ref{eqn:ach_2p_wtap_3}).

The key idea in our protocol (Protocol~\ref{protocol:C2P}) is that Alice converts the strings $\boldsymbol{X}|_{L_0}, \boldsymbol{X}|_{L_1}$ into independent secret keys $\boldsymbol{S}_0,\boldsymbol{S}_1$ respectively. Only one of these keys will be known to Bob and none of these keys will be known to Eve. Alice now sends $\boldsymbol{K}_0 \oplus \boldsymbol{S}_0, \boldsymbol{K}_1 \oplus \boldsymbol{S}_1$ to Bob over the public channel. In order to have $\boldsymbol{S}_0$ and $\boldsymbol{S}_1$ of length about $nr$ bits each, our protocol requires $|L_0| = |L_1| = nr/\epsilon_2$ approximately. Clearly, Bob knows $\boldsymbol{S}_U$ since he knows $\boldsymbol{X}|_{L_U}$ and can, thus, recover $\boldsymbol{K}_U$ from Alice's public message. As we prove later, $\boldsymbol{S}_{\overline{U}}$ remains unknown to a colluding Bob and Eve and so these colluding parties do not learn $\boldsymbol{K}_{\overline{U}}$, a key requirement for $2$-privacy. Since $\boldsymbol{S}_0, \boldsymbol{S}_1$ were independent and secret from Eve, clearly, Eve does not learn anything about Alice's strings from $\boldsymbol{K}_0 \oplus \boldsymbol{S}_0, \boldsymbol{K}_1 \oplus \boldsymbol{S}_1$. In order to convert $\boldsymbol{X}|_{L_0}, \boldsymbol{X}|_{L_1}$ into independent secret keys $\boldsymbol{S}_0, \boldsymbol{S}_1$ respectively, Alice selects two functions $F_0,F_1$ randomly and independently from a class $\mathcal{F}$ of \emph{universal}$_2$ hash functions \cite{carter_wegman_1979, carter_wegman_1981} (see Appendix~\ref{appndx:privacy_amplification} for details). The inputs of $F_0,F_1$ are about $nr/\epsilon_2$ bits long and their outputs are about $nr$ bits long. The required keys are $\boldsymbol{S}_0 = F_0(\boldsymbol{X}|_{L_0})$ and $\boldsymbol{S}_1 = F_1(\boldsymbol{X}|_{L_1})$. The main property of universal$_2$ hash functions used here is \emph{privacy amplification} \cite[Corollary 4]{generalized_privacy_ampl_1995}. In the present case, privacy amplification by the chosen universal$_2$ hash functions guarantees that the function output appears nearly random to any eavesdropper (e.g. colluding Bob and Eve) who does not know approximately a fraction $\epsilon_2$ (or more) of the function input. Alice sends $F_0,F_1$ to Bob alongwith $\boldsymbol{K}_0 \oplus \boldsymbol{S}_0, \boldsymbol{K}_1 \oplus \boldsymbol{S}_1$ over the public channel.

\begin{algorithm*}
\floatname{algorithm}{Protocol}
\caption{Protocol for achieving any $r < C_{2P}$}
\label{protocol:C2P}

Parameters : \begin{minipage}[t]{0.8\linewidth}
 \begin{itemize}
  \item $\delta \in (0,1)$ such that $r < (\epsilon_2 - \delta)(\min \{\epsilon_1, 1 - \epsilon_1\} - \delta)$ and $(\epsilon_2 - \delta) \in \mathbb{Q}$
  \item $0 < \tilde{\delta} < r$, $\tilde{\delta} \in \mathbb{Q}$
  \item $\beta = \frac{r}{\epsilon_2 - \delta}$
  \item $\beta n, n(r - \tilde{\delta}) \in \mathbb{N}$
  \item The rate\footnotemark of the protocol is $(r - \tilde{\delta})$
 \end{itemize}
\end{minipage}

\begin{multicols}{2}
\begin{algorithmic}[1]

\STATE Alice transmits an $n$-tuple $\boldsymbol{X}$ of i.i.d. Bernoulli($1/2$) bits over the channel.

\STATE \label{step:2p_wtap_abort0} Bob receives the $n$-tuple $\boldsymbol{Y}$ from BEC($\epsilon_1$). Bob forms the sets
\begin{align*}
\overline{E} & := \{ i \in \{1,2,\ldots,n\}: Y_i \neq \bot\} \\
E & := \{ i \in \{1,2,\ldots,n\}: Y_i = \bot\}
\end{align*}

If $|\overline{E}| < \beta n$ or $|E| < \beta n$, Bob aborts the protocol.

\STATE Bob creates the following sets:
\begin{align*}
L_U & \thicksim \text{Unif}\{A \subset \overline{E} : |A| = \beta n\} \\
L_{\overline{U}} & \thicksim \text{Unif}\{A \subset E : |A| = \beta n\} \\
\end{align*}

Bob reveals $L_0,L_1$ to Alice over the public channel.

\STATE Alice randomly and independently chooses functions $F_0,F_1$ from a family $\mathcal{F}$ of universal$_2$ hash functions:
\[ F_0,F_1 : \{0,1\}^{\beta n} \longrightarrow \{0,1\}^{n(r - \tilde{\delta})} \]

Alice finally sends the following information to Bob on the public channel:
\[ F_0, \; F_1, \; \boldsymbol{K}_0 \oplus F_0(\boldsymbol{X}|_{L_0}), \; \boldsymbol{K}_1 \oplus F_1(\boldsymbol{X}|_{L_1}) \]

\STATE Bob knows $F_U,\boldsymbol{X}|_{L_U}$ and can, therefore, recover $\boldsymbol{K}_U$.

\end{algorithmic}
\end{multicols}
\end{algorithm*}

\footnotetext{The parameters $\delta, \tilde{\delta}$ can be chosen to be arbitrarily small so that this rate takes any desired value less than $C_{2P}$.}

\begin{lemma}
\label{lem:c2p_ach_wtap}
Any $r < C_{2P}$ is an achievable $2$-private rate in the setup of Figure~\ref{fig:ot_hbc_wtap} when users are honest-but-curious.
\end{lemma}

A formal proof of this lemma is deferred to Appendix~\ref{appndx:proof_ach_2p_hbc_wtap}. A sketch of this proof is as follows. It suffices to prove this lemma only for rational values of $r < C_{2P}$ due to the denseness of $\mathbb{Q}$ in $\mathbb{R}$. Let $(\mathcal{P}_n)_{\{n \in \mathbb{N}\}}$ be a sequence of Protocol~\ref{protocol:C2P} instances, of rate $r - \tilde{\delta}$. With high probability, $\mathcal{P}_n$ does not abort. In that case, Bob knows the key $\boldsymbol{S}_U = F_U(\boldsymbol{X}|_{L_U})$ and can recover $\boldsymbol{K}_U$ from $\boldsymbol{K}_0 \oplus \boldsymbol{S}_0, \boldsymbol{K}_1 \oplus \boldsymbol{S}_1$ that Alice sends. As a result, (\ref{eqn:ach_2p_wtap_0}) holds for $(\mathcal{P}_n)_{\{n \in \mathbb{N}\}}$. For the key $\boldsymbol{S}_{\overline{U}} = F_{\overline{U}}(\boldsymbol{X}|_{L_{\overline{U}}})$, the privacy amplification by $F_{\overline{U}}$ on its input ensures that the amount of information that colluding Bob and Eve learn about $\boldsymbol{S}_{\overline{U}}$ falls exponentially in $n$. As a result, colluding Bob and Eve learn only a vanishingly small amount of information about $\boldsymbol{K}_{\overline{U}}$ and, thus, (\ref{eqn:ach_2p_wtap_1}) holds for $(\mathcal{P}_n)_{\{n \in \mathbb{N}\}}$. The only way that colluding Alice and Eve can learn $U$ is when Bob sends $L_0,L_1$. But since the channel acts independently on each input bit, the composition of $L_0,L_1$ does not reveal $U$. Thus, (\ref{eqn:ach_2p_wtap_2}) holds for the protocol sequence. Finally, conditioned on knowing $U$, Eve still does not learn anything about Alice's strings. This is because in the keys $\boldsymbol{S}_0 = F_0(\boldsymbol{X}|_{L_0}), \boldsymbol{S}_1 =  F_1(\boldsymbol{X}|_{L_1})$, the privacy amplification by $F_0,F_1$ on their respective inputs ensures that the amount of information Eve learns about $\boldsymbol{S}_0, \boldsymbol{S}_1$ falls exponentially in $n$. As a result, Eve gains only a vanishingly small amount of information about $\boldsymbol{K}_0, \boldsymbol{K}_1$ from Alice's public message. This guarantees that (\ref{eqn:ach_2p_wtap_3}) holds for $(\mathcal{P}_n)_{\{n \in \mathbb{N}\}}$.


\subsection{1-privacy : Achievability}
\label{sec:1p_ot_protocol}

 Our protocol (Protocol~\ref{protocol:C1P}) for achieving OT in the presence of Eve, with $1$-privacy in our setup (see Figure~\ref{fig:ot_hbc_wtap}), is an extension of Ahlswede and Csisz\'ar's two-party OT protocol \cite{ot2007}. For a sequence of Protocol~\ref{protocol:C1P} instances of rate $r < C_{1P}$, we show that (\ref{eqn:ach_1p_wtap_0})-(\ref{eqn:ach_1p_wtap_3}) hold. This establishes that any $r < C_{1P}$ is an achievable $1$-private rate in the setup of Figure~\ref{fig:ot_hbc_wtap}. 

For achieving $1$-privacy, recall that privacy for Alice's strings is required only individually against Bob and against Eve, not against colluding Bob and Eve. As a result, the main change in Protocol~\ref{protocol:C1P}, compared to Protocol~\ref{protocol:C2P}, is that the requirement of $L_{\overline{U}}$ coming entirely from $E$ is relaxed. Protocol~\ref{protocol:C1P} requires that $nr$ indices in $L_{\overline{U}}$ have to come from $E$.  The remaining about $(nr/\epsilon_2) - nr$ indices in $L_{\overline{U}}$ can come from an arbitrary combination of leftover indices of $E$ and $\overline{E}$. Since the key $\boldsymbol{S}_{\overline{U}} = F_{\overline{U}}(\boldsymbol{X}|_{L_{\overline{U}}})$ is about $nr$ bits long, privacy amplification by $F_{\overline{U}}$ on its input $\boldsymbol{X}|_{L_{\overline{U}}}$ ensures that $\boldsymbol{S}_{\overline{U}}$ is unknown to Bob. Since $|L_{\overline{U}}|$ is about $nr/\epsilon_2$, the privacy amplification also guarantees that $\boldsymbol{S}_{\overline{U}}$ is unknown to Eve as well. Thus, the key $\boldsymbol{S}_{\overline{U}}$ remains hidden individually from Bob and from Eve and that suffices to achieve $1$-privacy in the setup. Furthermore, note that when $\epsilon_1 < 1/2$, Protocol~\ref{protocol:C2P} had unused indices from $\overline{E}$ which Protocol~\ref{protocol:C1P} can use in constructing a larger $L_{\overline{U}}$. This results in higher achievable $1$-private rates compared to achievable $2$-private rates when $\epsilon_1 < 1/2$.

\begin{algorithm*}
\floatname{algorithm}{Protocol}
\caption{Protocol for achieving any $r < C_{1P}$}
\label{protocol:C1P}

Parameters : \begin{minipage}[t]{0.8\linewidth}
 \begin{itemize}
  \item $\delta \in (0,1)$ such that $r < \min\{ (\epsilon_1 - \delta), \frac{1}{2}(\epsilon_2 - \delta), (\epsilon_2 - \delta)(1 - \epsilon_1 - \delta) \}$ and $(\epsilon_2 - \delta) \in \mathbb{Q}$
  \item $0 < \tilde{\delta} < r$, $\tilde{\delta} \in \mathbb{Q}$
  \item $\beta = \frac{r}{\epsilon_2 - \delta}$
  \item $\beta n, nr, n(r - \tilde{\delta}) \in \mathbb{N}$
  \item The rate\footnotemark of the protocol is $(r - \tilde{\delta})$
 \end{itemize}
\end{minipage}

\begin{multicols}{2}
\begin{algorithmic}[1]

\STATE Alice transmits an $n$-tuple $\boldsymbol{X}$ of i.i.d. Bernoulli($1/2$) bits over the channel.

\STATE \label{step:1p_wtap_abort0} Bob receives the $n$-tuple $\boldsymbol{Y}$ from BEC($\epsilon_1$). Bob forms the sets
\begin{align*}
\overline{E} & := \{ i \in \{1,2,\ldots,n\}: Y_i \neq \bot\} \\
E & := \{ i \in \{1,2,\ldots,n\}: Y_i = \bot\}
\end{align*}

If $|\overline{E}| < \beta n$ or $|E| < nr$, Bob aborts the protocol.

\STATE Bob creates the following sets:
\begin{align*}
L & \thicksim \text{Unif}\{ A \subset E : |A| = nr \} \\
L_U & \thicksim \text{Unif}\{A \subset \overline{E} : |A| = \beta n\} \\
L_{\overline{U}} & \thicksim L \cup \text{Unif}\{A \subset \overline{E} \backslash L_U \; \cup E \backslash L : |A| = (\beta - r)n\}
\end{align*}

Bob reveals $L_0,L_1$ to Alice over the public channel.

\STATE Alice randomly and independently chooses functions $F_0,F_1$ from a family $\mathcal{F}$ of universal$_2$ hash functions:
\[ F_0,F_1 : \{0,1\}^{\beta n} \longrightarrow \{0,1\}^{n(r - \tilde{\delta})} \]

Alice finally sends the following information on the public channel:
\[ F_0, \; F_1, \; \boldsymbol{K}_0 \oplus F_0(\boldsymbol{X}|_{L_0}), \; \boldsymbol{K}_1 \oplus F_1(\boldsymbol{X}|_{L_1}) \]

\STATE Bob knows $F_U,\boldsymbol{X}|_{L_U}$ and can, therefore, recover $\boldsymbol{K}_U$.

\end{algorithmic}
\end{multicols}
\end{algorithm*}

\begin{lemma}
\label{lem:c1p_ach_wtap}
Any $r < C_{1P}$ is an achievable $1$-private rate in the setup of Figure~\ref{fig:ot_hbc_wtap} when users are honest-but-curious.
\end{lemma}

This lemma is formally proved in Appendix~\ref{appndx:proof_ach_1p_hbc_wtap}. A sketch of its proof now follows. Let $(\mathcal{P}_n)_{\{n \in \mathbb{N}\}}$ be a sequence of Protocol~\ref{protocol:C1P} instances, of rate $r - \tilde{\delta}$. If the protocol does not abort, then (\ref{eqn:ach_1p_wtap_0}), (\ref{eqn:ach_1p_wtap_2}) and (\ref{eqn:ach_1p_wtap_3}) hold for $(\mathcal{P}_n)_{\{n \in \mathbb{N}\}}$ for the same reasons that (\ref{eqn:ach_2p_wtap_0}), (\ref{eqn:ach_2p_wtap_2}) and (\ref{eqn:ach_2p_wtap_3}) respectively hold for a sequence of Protocol~\ref{protocol:C2P} instances. To see that (\ref{eqn:ach_1p_wtap_1}) holds for $(\mathcal{P}_n)_{\{n \in \mathbb{N}\}}$, note that $L_{\overline{U}}$ consists of at least $nr$ indices at which $\boldsymbol{Y}$ is erased. Also, the key $\boldsymbol{S}_{\overline{U}} = F_{\overline{U}}(\boldsymbol{X}|_{L_{\overline{U}}})$ is about $nr$ bits long and privacy amplification by $F_{\overline{U}}$ on its input $\boldsymbol{X}|_{L_{\overline{U}}}$ ensures that the amount of information Bob learns about $\boldsymbol{S}_{\overline{U}}$ falls exponentially with $n$. Hence, Bob learns only a vanishingly small amount of information about the string $\boldsymbol{K}_{\overline{U}}$ from $\boldsymbol{K}_0 \oplus \boldsymbol{S}_0, \boldsymbol{K}_1 \oplus \boldsymbol{S}_1$.

\footnotetext{The parameters $\delta, \tilde{\delta}$ can be chosen to be arbitrarily small so that this rate takes any desired value less than $C_{1P}$.}


\subsection{$2$-privacy : Converse}

We only require a weaker secrecy condition to prove our converse. Specifically, we only need (\ref{eqn:ach_2p_wtap_1}) and (\ref{eqn:ach_2p_wtap_3}) to hold with a $1/n$ multiplied to their left-hand-sides.

\begin{lemma}
\label{lem:upper_bound_2p_wtap}
If $r_{2P}$ is an achievable $2$-private rate in the setup of Figure~\ref{fig:ot_hbc_wtap} when users are honest-but-curious, then
\[ r_{2P} \leq C_{2P}. \]
\end{lemma}

\begin{proof}
We first show a general upper bound on $r_{2P}$. For the setup in Figure~\ref{fig:ot_hbc_wtap_bcast},
\[  r_{2P} \leq  \min\left \{ \max_{p_X} I(X ; Y | Z), \max_{p_X} H(X | Y, Z) \right \}. \]

It is straightforward to verify that any OT protocol for the setup in Figure~\ref{fig:ot_hbc_wtap_bcast} is a two-party OT protocol between Alice and Bob-Eve combined. Using an outerbound for OT capacity in \cite{ot2007}, we have
\[ r_{2P} \leq \max_{p_X} H(X | Y, Z).  \]

To see that $\max_{p_X} I(X ; Y | Z)$ is an upper bound on $r_{2P}$, we argue that using an OT protocol, Alice and Bob can agree on a secret key, secret from Eve, at the same rate as the OT.  Suppose we modify the OT protocol so that at the end of it, Bob reveals $U$ over the public channel. As a result, Alice learns $\boldsymbol{K}_U$. We show that this string $\boldsymbol{K}_U$ is a secret key between Alice and Bob, which Eve knows nothing about. Since Alice learns $\boldsymbol{K}_U$ and (\ref{eqn:ach_2p_wtap_0}) holds, both Alice and Bob learn $\boldsymbol{K}_U$. Further, (\ref{eqn:ach_2p_wtap_3}) implies that $(1/n) \cdot I(\boldsymbol{K}_0,\boldsymbol{K}_1,U,\boldsymbol{K}_U  ; V_E)  \longrightarrow 0$. This, in turn, implies that $(1/n) \cdot I(\boldsymbol{K}_U  ; V_E | U)  \longrightarrow 0$. Now:
\begin{align*}
\frac{1}{n}I(\boldsymbol{K}_U  ; V_E | U) & = \frac{1}{n}(I(\boldsymbol{K}_U  ; V_E, U) - I(\boldsymbol{K}_U ; U)) \\
                                          & \geq \frac{1}{n}I(\boldsymbol{K}_U ; V_E, U) - \frac{1}{n}.
\end{align*}

Hence, $(1/n) \cdot I(\boldsymbol{K}_U ; V_E, U) \longrightarrow 0$ as $n \longrightarrow \infty$. This shows that in the modified protocol, after Bob reveals $U$ at the end, Alice and Bob learn $\boldsymbol{K}_U$ and Eve learns only a vanishingly small amount of information about $\boldsymbol{K}_U$. Hence, $\boldsymbol{K}_U$ becomes a secret key between Alice and Bob, against Eve. Since $\max_{p_X} I(X ; Y | Z)$ is an upperbound on secret key capacity for the setup of Figure~\ref{fig:ot_hbc_wtap_bcast} \cite{sec-key1993}, we get:
\[ r_{2P} \leq \max_{p_X} I(X ; Y | Z). \]

Evaluating the upper bound for the setup of Figure~\ref{fig:ot_hbc_wtap}, we get:
\begin{align*}
r_{2P} & \leq \max_{p_X} I(X ; Y | Z)  = \epsilon_2  (1 - \epsilon_1) \\
r_{2P} & \leq \max_{p_X} H(X | (Y, Z)) =  \epsilon_2 \epsilon_1.
\end{align*}

As a result, $r_{2P} \leq C_{2P}$.

\end{proof}


\subsection{$1$-privacy : Converse}

As before, we show that the upper bounds hold even under weakened security conditions, where (\ref{eqn:ach_1p_wtap_1}) and (\ref{eqn:ach_1p_wtap_3}) hold with a $1/n$ multiplied to their left-hand-sides.

\begin{lemma}
If $r_{1P}$ is an achievable $1$-private rate, with honest-but-curious users, for the setup of Figure~\ref{fig:ot_hbc_wtap}, then
\[ r_{1P} \leq C_{1P}. \]
\end{lemma}

\begin{proof}

We first show that $r_{1P} \leq   \min \left\{ \epsilon_1, \epsilon_2  (1 - \epsilon_1)  \right\}$ by means of the following more general statement: For the setup of Figure~\ref{fig:ot_hbc_wtap_bcast},
\begin{equation}
\label{eqn:c1p_general_outer}
 r_{1P} \leq \left \{ \max_{p_X} I(X ; Y | Z), \max_{p_X} H(X | Y) \right \}. 
\end{equation}

Proof for $r_{1P} \leq \max_{p_X} I(X ; Y | Z)$ is identical to the proof for $r_{2P} \leq \max_{p_X} I(X ; Y | Z)$ (in the proof of Lemma~\ref{lem:upper_bound_2p_wtap}).

$r_{1P} \leq \max_{p_X} H(X | Y)$ follows from the observation that any OT protocol achieving $1$-privacy in the setup of Figure~\ref{fig:ot_hbc_wtap_bcast}, is also a two-party OT protocol between two users connected by the channel $p_{Y|X}$. As a result, $\max_{p_X} H(X | Y)$, which is an upper bound on two-party OT rate \cite{ot2007}, is also an upper bound on $r_{1P}$. Evaluated for the setup of Figure~\ref{fig:ot_hbc_wtap}, these upperbounds are:
\begin{align*}
r_{1P} & \leq \max_{p_X} H(X | Y)  = \epsilon_1 \\
r_{1P} & \leq \max_{p_X} I(X ; Y | Z) = \epsilon_2(1 - \epsilon_1).
\end{align*}

We now prove that $r_{1P} \leq \epsilon_2/2$ for setup of Figure~\ref{fig:ot_hbc_wtap}, which will complete the proof of the converse. We use the following lemma (proved in Appendix \ref{appndx:proof_small_quant_1p_hbc_wtap}) which shows that both $\boldsymbol{K_0},\boldsymbol{K}_1$ can be decoded from $\boldsymbol{X}, \boldsymbol{\Lambda}$.

\begin{lemma} 
\label{lem:small_quant_1p_hbc_wtap}
\[ \frac{1}{n}H(\boldsymbol{K_0},\boldsymbol{K}_1 | \boldsymbol{X}, \boldsymbol{\Lambda} ) \longrightarrow 0 \text{ as } n \longrightarrow \infty \]
\end{lemma}

Intuitively, this lemma holds for the following reason: Alice cannot learn which of its two strings Bob requires (cf.(\ref{eqn:ach_1p_wtap_2})). Thus, anyone observing the cut separating Alice from the system, i.e. $\boldsymbol{X},\boldsymbol{\Lambda}$, must be able to decode both $\boldsymbol{K_0},\boldsymbol{K}_1$. If this were not the case and, say, $\boldsymbol{K}_0$ could not be decoded from $\boldsymbol{X},\boldsymbol{\Lambda}$, then Alice can infer that Bob wanted $\boldsymbol{K}_1$ which violates (\ref{eqn:ach_1p_wtap_2}).

We can convert this lemma into an upperbound as follows: Knowing $\boldsymbol{X},\boldsymbol{\Lambda}$, one can decode $\boldsymbol{K_0},\boldsymbol{K}_1$. Eve has access to almost all of $\boldsymbol{X},\boldsymbol{\Lambda}$, except about an $\epsilon_2$ fraction of  $\boldsymbol{X}$ which was erased by the channel. It is required that Eve learns no information about both $\boldsymbol{K_0},\boldsymbol{K}_1$. As a result, each string's length cannot exceed $n \epsilon_2/2$. More formally, let $\tilde{E} := \{ i \in \{1,2,\ldots,n\} : Z_i = \bot \}$. Let $\tilde{e}$ denote a realization of $\tilde{E}$ and let $\overline{\tilde{e}} = \{1,2,\ldots,n\} \backslash \tilde{e}$ be the complement of $\tilde{e}$. Then,
\begin{align*}
2m & =  H(\boldsymbol{K}_0,\boldsymbol{K}_1) \\
   & =  I(\boldsymbol{K}_0,\boldsymbol{K}_1 ; \boldsymbol{X},\boldsymbol{\Lambda}) + H(\boldsymbol{K}_0,\boldsymbol{K}_1 | \boldsymbol{X},\boldsymbol{\Lambda}) \\
   & \stackrel{\text{(a)}}{=}  I(\boldsymbol{K}_0,\boldsymbol{K}_1 ; \boldsymbol{X},\boldsymbol{\Lambda}) + o(n) \\
   & \stackrel{\text{(b)}}{=}  I(\boldsymbol{K}_0,\boldsymbol{K}_1 ; \boldsymbol{X},\boldsymbol{\Lambda} | \tilde{E}) + o(n) \\
   & =  \sum_{\tilde{e} \subseteq \{1,2,\ldots,n\} } p_{\tilde{E}}(\tilde{e}) I(\boldsymbol{K}_0,\boldsymbol{K}_1 ; \boldsymbol{X},\boldsymbol{\Lambda} | \tilde{E} = \tilde{e}) + o(n) \\
   & =  \sum_{\tilde{e} \subseteq \{1,2,\ldots,n\} } p_{\tilde{E}}(\tilde{e}) I(\boldsymbol{K}_0,\boldsymbol{K}_1 ; \boldsymbol{X}|_{\overline{\tilde{e}}}, \boldsymbol{\Lambda}  \;  \mathlarger{\mid}  \;    \tilde{E} = \tilde{e}) +\sum_{\tilde{e} \subseteq \{1,2,\ldots,n\} } p_{\tilde{E}}(\tilde{e}) I(\boldsymbol{K}_0,\boldsymbol{K}_1 ; \boldsymbol{X}|_{\tilde{e}} \; \mathlarger{\mid} \; \boldsymbol{X}|_{\overline{\tilde{e}}}, \boldsymbol{\Lambda}, \tilde{E} = \tilde{e}) + o(n)  \\
   & \leq   \sum_{\tilde{e} \subseteq \{1,2,\ldots,n\} } p_{\tilde{E}}(\tilde{e}) I(\boldsymbol{K}_0,\boldsymbol{K}_1 ; \boldsymbol{X}|_{\overline{\tilde{e}}}, \boldsymbol{\Lambda} \;  \mathlarger{\mid}  \; \tilde{E} = \tilde{e}) + \sum_{\tilde{e} \subseteq \{1,2,\ldots,n\} } p_{\tilde{E}}(\tilde{e}) H(\boldsymbol{X}|_{\tilde{e}} \; \mathlarger{\mid}  \; \tilde{E} = \tilde{e}) + o(n) \\
   & \leq    I(\boldsymbol{K}_0,\boldsymbol{K}_1 ; \boldsymbol{Z}, \boldsymbol{\Lambda}) +   \sum_{\tilde{e} \subseteq \{1,2,\ldots,n\} } p_{\tilde{E}}(\tilde{e}) |\tilde{e}|   + o(n) \\
 & = I(\boldsymbol{K}_0,\boldsymbol{K}_1 ; \boldsymbol{Z}, \boldsymbol{\Lambda}) +      n\epsilon_2 + o(n) \\
  & \stackrel{\text{(c)}}{=}   n\epsilon_2 + o(n)
\end{align*}
where (a) follows from Lemma~\ref{lem:small_quant_1p_hbc_wtap}, (b) from the independence of Eve's channel, and (c) from (\ref{eqn:ach_1p_wtap_3}). Therefore,
\begin{align*} 
r_{1P} & = \frac{m}{n} \\
          & \leq \frac{\epsilon_2}{2} + \frac{o(n)}{n}
\end{align*}

\end{proof}


\section{Oblivious transfer over a wiretapped channel in the malicious model : Proof of Theorem~\ref{thm:result_2p_malicious_wtap}}
\label{sec:proofs_1}

In this setup (see Figure~\ref{fig:ot_hbc_wtap}), Alice and Bob are required to achieve OT, with $2$-privacy, in the presence of Eve when both Alice and Bob can be malicious. We show that for this problem, any
\[ R < \left\{ \begin{array}{lr} \epsilon_1 \epsilon_2, & \epsilon_1 \leq \frac{1}{2}\\ \epsilon_1 \epsilon_2 (1 - \epsilon_1), & \epsilon_1 > \frac{1}{2} \end{array} \right. \]
is an achievable $2$-private rate. The protocol we give for proving the achievability of $R$ is described separately for the regimes $\epsilon_1 \leq 1/2$ and $\epsilon_1 > 1/2$, since the protocol differs substantially in these two regimes.


\subsection{Protocol when $\epsilon_1 \leq 1/2$}
\label{sec:protocol_malicious_lt}

In this regime, our protocol (Protocol~\ref{protocol:malicious_lt}) is an adaptation of the protocol described for two-party OT with malicious users in \cite{savvides_thesis}, \cite{PintoDowsMorozNasc2011} and \cite{DowsNasc2014-arxiv}. Bob forms the \emph{tuples} of indices $\boldsymbol{L}_0, \boldsymbol{L}_1$ and communicates these tuples to Alice. In contrast, recall that in the honest-but-curious case Bob communicated \emph{sets} $L_0,L_1$ to Alice. Also, unlike the honest-but-curious case, a small fraction of both tuples is now allocated for use in checks that Alice performs to detect possible malicious behavior by Bob. These checks are based on interactive hashing \cite{savvides_thesis}, which also provides a mechanism for Bob to detect possible malicious behavior by Alice. Interactive hashing is an interactive protocol between two users over a noiseless channel, one acting as a sender and the other acting as a receiver. The input to the protocol is a bit-string held by the sender. The output of the protocol are two bit-strings of the same lengths as the input bit-string, both available to the sender as well as the receiver. Of the two output strings, one is the same as the input string but the receiver cannot make out which one of the two it is. The sender, of course, knows which of the output strings was the input for the protocol but it cannot influence the choice of the other string being output by the protocol. Appendix~\ref{appndx:interactive_hashing} states the properties and describes a protocol for interactive hashing, as given in \cite{savvides_thesis}. While using interactive hashing in our protocol, Bob acts as the sender and Alice acts as the receiver. 
The following explains our protocol in more detail.

Alice initiates the protocol by transmitting a sequence $\boldsymbol{X}$ of $n$ i.i.d. Bernoulli($1/2$) bits over the channel. Bob and Eve receive the channel outputs $\boldsymbol{Y}$ and $\boldsymbol{Z}$ respectively. Let $E$ be the set of all the indices at which $\boldsymbol{Y}$ is erased and $\overline{E}$ that of all the indices at which $\boldsymbol{Y}$ is unerased. If $|E|$ or $|\overline{E}|$ are not sufficiently close to their respective expected values, Bob aborts the protocol since he does not have enough of either erasures or non-erasures to run the protocol. Bob now has to form two equal-sized, disjoint tuples of indices, $\boldsymbol{L}_U$ and $\boldsymbol{L}_{\overline{U}}$, each tuple consisting of about $n/2$ indices. The \emph{good tuple} $\boldsymbol{L}_U$ is picked uniformly at random from $\overline{E}$. To form the \emph{bad tuple} $\boldsymbol{L}_{\overline{U}}$, Bob first uniformly at random selects a subset $J \subset \{1,2,\ldots,|\{\boldsymbol{L}_{\overline{U}}\}|\}$, with $|J|$ being about $(1/2 - \epsilon_1)n$. The elements $\boldsymbol{L}_{\overline{U}}|_J$ are chosen uniformly at random from elements of $\overline{E} \backslash \{\boldsymbol{L}_U\}$ while the elements $\boldsymbol{L}_{\overline{U}}|_{J^c}$ are chosen uniformly at random from elements of $E$. Here, $J^c$ is the set $\{1,2,\ldots,|\{\boldsymbol{L}_{\overline{U}}\}|\} \backslash J$. Note that $|\overline{E}|$ will be nearly equal to $|\{\boldsymbol{L}_U\}| + |J|$. Bob reveals $\boldsymbol{L}_0,\boldsymbol{L}_1$ to Alice. Conditioned on Alice's view, $\boldsymbol{L}_0,\boldsymbol{L}_1$ are uniformly distributed tuples of indices. This uniform distribution prevents leakage of any information about $J$, and thereby $U$, to Alice, when Alice sees $\boldsymbol{L}_0,\boldsymbol{L}_1$. Thereafter, Bob and Alice initiate interactive hashing, with a bit-string $\boldsymbol{S}$ representing $J$ being the input to interactive hashing. As the output of interactive hashing, both Alice and Bob receive some pair of strings $\boldsymbol{S}_0, \boldsymbol{S}_1$ which represent some subsets $J_0,J_1$ respectively. Suppose $J_{\Phi} = J$ where the random variable $\Phi \in \{0,1\}$. In a crucial step designed to catch a malicious Bob, Bob is now required to reveal the bits of $\boldsymbol{X}$ at indices $\boldsymbol{L}_U|_{J_{\overline{\Phi}}}$ and at indices $\boldsymbol{L}_{\overline{U}}|_{J_{\Phi}}$. An honest Bob knows these bits fully. And we  prove later that if Bob cheats by populating $\boldsymbol{L}_{\overline{U}}$ with more than the designated number of elements from $\overline{E}$, Bob will fail to reveal all the bits of $\boldsymbol{X}$ at indices $\boldsymbol{L}_U|_{J_{\overline{\Phi}}}$ with high probability. As in the honest-but-curious setup, Alice forms two keys to encrypt her strings, where both the keys are secret from Eve. Towards this goal, Alice randomly and independently selects two functions $F_0,F_1$ from a universal$_2$ class of functions $\mathcal{F}$ (see Appendix~\ref{appndx:privacy_amplification} for its definition). The required keys are $F_0(\boldsymbol{X}|_{\boldsymbol{L}_0})$ and $F_1(\boldsymbol{X}|_{\boldsymbol{L}_1})$. Alice now sends $F_0, F_1, \boldsymbol{K}_0 \oplus F_0(\boldsymbol{X}|_{\boldsymbol{L}_0}), \boldsymbol{K}_1 \oplus F_1(\boldsymbol{X}|_{\boldsymbol{L}_1})$ to Bob over the public channel. An honest Bob knows $\boldsymbol{X}|_{\boldsymbol{L}_U}$ and can obtain $\boldsymbol{K}_U$. As we will show, a malicious Bob colluding with Eve, if not caught already, learns a vanishingly small amount of information about at least one of the two keys and, as a result, can learn only a vanishing amount of information about the corresponding string.

\begin{algorithm*}
\floatname{algorithm}{Protocol}
\caption{Malicious Users, $\epsilon_1 \leq 1/2$}
\label{protocol:malicious_lt}

Parameters : \begin{minipage}[t]{0.8\linewidth}
 \begin{itemize}
   \item $\gamma = (\frac{1}{2} - \epsilon_1 - \tilde{\delta})$, $\tilde{\delta} \in (-1,1)$ such that $\gamma > 0, \gamma \in \mathbb{Q}$
   \item $\beta = (\frac{1}{2} - \delta - \tilde{\delta})$, $\delta \in (0,1)$ such that $\beta > 0, \beta \in \mathbb{Q}$
   \item $\delta' \in (0,1)$ such that $(\epsilon_1 \epsilon_2 - 5\delta - 2\tilde{\delta} - \delta') > 0, (\epsilon_1 \epsilon_2 - 5\delta - 2\tilde{\delta} - \delta') \in \mathbb{Q}$
   \item $\beta n, \gamma n, (\epsilon_1 \epsilon_2 - 5\delta - 2\tilde{\delta} - \delta')n \in \mathbb{N}$
   \item The rate\footnotemark of the protocol is $(\epsilon_1 \epsilon_2 - 5\delta - 2\tilde{\delta} - \delta')$
   \item $\mathcal{I} = \{1,2,\ldots,\beta n\}$
   \item $m = \left \lceil \log (\comb{\beta n}{\gamma n}) \right \rceil$ 
   \item $\mathcal{T} = \{ (A) : A \subset \mathcal{I}, |A| = \gamma n\}$
   \item $Q : \{0,1\}^m \longrightarrow\mathcal{T}$ is an onto map
 \end{itemize}
\end{minipage}

\begin{multicols}{2}
\begin{algorithmic}[1]

\STATE Alice transmits an $n$-tuple $\boldsymbol{X}$ of i.i.d. Bernoulli($1/2$) bits over the channel.

\STATE \label{step:abort0_0} Bob receives the $n$-tuple $\boldsymbol{Y}$ from BEC($\epsilon_1$). Bob forms the sets
\begin{align*}
\overline{E} & := \{ i \in \{1,2,\ldots,n\}: Y_i \neq \bot\} \\
E & := \{ i \in \{1,2,\ldots,n\}: Y_i = \bot\}
\end{align*}

If $|\overline{E}| < (\beta + \gamma)n$ or $|E| < (\beta - \gamma)n$, Bob aborts the protocol.

\STATE Bob chooses a bit-string $\boldsymbol{S} \thicksim \text{Unif}(\{0,1\}^m)$. Let  $\boldsymbol{J} = Q(\boldsymbol{S})$ and $\boldsymbol{J}^c = (\mathcal{I} \backslash \{\boldsymbol{J}\})$.
Bob forms the tuples $\boldsymbol{L}_U, \boldsymbol{L}_{\overline{U}} \in \{1,2,\ldots,n\}^{\beta n}$ as follows:
\begin{align*}
L_{U,i} & \thicksim \text{Unif}\{ \overline{E} \backslash \{\boldsymbol{L}^{i-1}_U\} \} \\
L_{\overline{U},J_i} & \thicksim \text{Unif}\{ \overline{E} \backslash \{ \{\boldsymbol{L}_U\} \cup \{ \boldsymbol{L}_{\overline{U}}|_{\boldsymbol{J}^{i-1}} \} \} \} \\ L_{\overline{U},J^c_i} & \thicksim \text{Unif}\{ E \backslash \{ \boldsymbol{L}_{\overline{U}}|_{\boldsymbol{J}^{c,i-1}} \} \}
\end{align*}

Bob reveals the tuples $\boldsymbol{L}_0,\boldsymbol{L}_1$ over the public channel.

\STATE \label{step:abort0_1} Alice checks to see that $\{ \boldsymbol{L}_0\} \cap \{\boldsymbol{L}_1\} = \emptyset$, otherwise Alice aborts the protocol.

\STATE Bob and Alice initiate interactive hashing, with the input to interactive hashing being the string $\boldsymbol{S}$ held by Bob.  Interactive hashing returns strings $\boldsymbol{S}_0, \boldsymbol{S}_1 \in \{0,1\}^m$, which are available to both Alice and Bob. Let $\Phi \in \{0,1\}$ such that $\boldsymbol{S}_{\Phi} = \boldsymbol{S}$. Let 
\begin{align*}
\boldsymbol{J}_0 & = Q(\boldsymbol{S}_0) \\
\boldsymbol{J}_1 & = Q(\boldsymbol{S}_1)
\end{align*}

\STATE Bob now sends the following to Alice over the public channel: $\Theta = \Phi \oplus U$, $\boldsymbol{Y}|_{\boldsymbol{L}_0|_{ \boldsymbol{J}_{\overline{\Theta}} }  }$, $\boldsymbol{Y}|_{\boldsymbol{L}_1|_{ \boldsymbol{J}_{\Theta} } }$.

\STATE \label{step:abort0_2} Alice checks that $\boldsymbol{Y}|_{\boldsymbol{L}_0|_{ \boldsymbol{J}_{\overline{\Theta}} }  }$ completely matches $\boldsymbol{X}|_{\boldsymbol{L}_0|_{ \boldsymbol{J}_{\overline{\Theta}} }  }$ and $\boldsymbol{Y}|_{\boldsymbol{L}_1|_{ \boldsymbol{J}_{\Theta} } }$ completely matches $\boldsymbol{X}|_{\boldsymbol{L}_1|_{ \boldsymbol{J}_{\Theta} } }$, otherwise Alice aborts the protocol.

\STATE Alice chooses functions $F_0,F_1$, randomly and independently, from a family $\mathcal{F}$ of universal$_2$ hash functions :
\[ F_0,F_1: \{0,1\}^{\beta n} \longrightarrow \{0,1\}^{(\epsilon_1 \epsilon_2 - 5\delta - 2\tilde{\delta} - \delta')n } \]

Alice finally send the following to Bob over the public channel: 
\[ F_0, \; F_1, \; \boldsymbol{K}_0 \oplus F_0(\boldsymbol{X}|_{\boldsymbol{L}_0}), \; \boldsymbol{K}_1 \oplus F_1(\boldsymbol{X}|_{\boldsymbol{L}_1}) \]

\STATE Bob knows $F_U$ and $\boldsymbol{Y}|_{\boldsymbol{L}_U}$ and can, therefore, recover $\boldsymbol{K}_U$.

\end{algorithmic}
\end{multicols}
\end{algorithm*}

\footnotetext{The parameters $\delta, |\tilde{\delta}|, \delta'$ can be chosen to be arbitrarily small so that this rate can take any desired value less than $\epsilon_1 \epsilon_2$. Note that when $\epsilon_1 = 1/2$, we need $\tilde{\delta} < 0$ and we can choose $\delta = -2 \tilde{\delta}$. For $\epsilon_1 < 1/2$, we choose $\tilde{\delta} > 0$.}


\subsection{Protocol when $\epsilon_1 > 1/2$}
\label{sec:protocol_malicious_geq}

Our protocol for this regime (Protocol~\ref{protocol:malicious_gt}) is the main novelty in this section. The above approach, where Bob gets to choose both the tuples of indices $\boldsymbol{L}_0,\boldsymbol{L}_1$, does not work in this regime. To see why this is the case, consider the setup with $\epsilon_1 = 2/3, \epsilon_2 = 1/2$. In this setup, $\boldsymbol{Y}$ is unerased at about $n/3$ indices. An honest Bob, therefore, will construct the tuples $\boldsymbol{L}_0, \boldsymbol{L}_1$ with each tuple consisting of about $n/3$ indices approximately. The good tuple $\boldsymbol{L}_U$ will have nearly all the unerased indices in $\boldsymbol{Y}$. A malicious Bob who wishes to remain undetected by Alice will also, hence, form tuples having about $n/3$ indices each. However, a malicious Bob colluding with Eve has access to about $2n/3$ indices at which he knows the bits transmitted by Alice. As a result, malicious Bob can form the two tuples $\boldsymbol{L}_0, \boldsymbol{L}_1$ consisting only of those indices at which he knows the bits transmitted by Alice. In such a situation, Bob will pass any check that Alice may impose, without getting caught, and will get to learn both of Alice's strings. At the root of this problem is Bob's ability to choose both tuples $\boldsymbol{L}_0,\boldsymbol{L}_1$. Our protocol takes away this ability from Bob, allowing Bob to form only one of the tuples, \emph{with the other tuple being provided to Bob by interactive hashing}. Thus, interactive hashing is used to output tuples using which the keys that encrypt Alice's strings are derived. We describe the protocol in more detail below.

The steps in this protocol are the same as for the protocol in the regime $\epsilon_1 \leq 1/2$, upto and including the formation of sets $E, \overline{E}$ by Bob. There are two main differences thereafter. Firstly, $L_0, L_1$ are now random sets, not random tuples. Secondly, Bob is allowed to construct only the \emph{good set} $L_U$, not the \emph{bad set} $L_{\overline{U}}$. The set $L_{\overline{U}}$ is obtained as an output of interactive hashing when interactive hashing is invoked with the bit-string representing $L_U$ as its input. Specifically, suppose the sets $L_U, L_{\overline{U}}$ are required to be of cardinality $\beta n$, where $0 < \beta < 1/2$. Let $m$ be the smallest integer required so that each $\beta n$-sized subset of $\{1,2,\ldots,n\}$ can be represented with a unique $m$-bit string. Bob selects one of these $m$-bit strings, say $\boldsymbol{S}$, to represent its choice of $L_U$. Of course, the choice of the string $\boldsymbol{S}$ should be such that $L_U \subset \overline{E}$. Alice and Bob now initiate interactive hashing. Bob holds $\boldsymbol{S}$ as the input to interactive hashing and both Alice and Bob receive as outputs some strings $\boldsymbol{S}_0, \boldsymbol{S}_1$, one of which is the same as $\boldsymbol{S}$. The strings $\boldsymbol{S}_0, \boldsymbol{S}_1$ identify subsets $L_0,L_1$ one of which is $L_U$ and the other is used as $L_{\overline{U}}$. The key property used to guarantee privacy against malicious Bob is the following: The sets $L_0,L_1$ cannot simultaneously have more than a threshold number each of indices at which either $\boldsymbol{Y}$ or $\boldsymbol{Z}$ or both are unerased. That is, at least one of $L_0, L_1$  has fewer than this threshold number of such indices. And our protocol \emph{effectively removes} that threshold number of such indices each from $L_0,L_1$. This \emph{removal} happens in two steps: in the first step, Bob is asked to reveal bits $\boldsymbol{X}|_{L_0 \cap L_1}$ as a check by Alice and indices $L_0 \cap L_1$ are not used thereafter. In the second step, sufficient privacy amplification is used over the bits $\boldsymbol{X}|_{L_0 \backslash L_0 \cap L_1}$ and $\boldsymbol{X}|_{L_1 \backslash L_0 \cap L_1}$, with the outputs of privacy amplification used as the keys to encrypt Alice's strings. This two-step process guarantees that a malicious Bob, colluding with Eve, can gain only a vanishingly small amount of information about at least one of the keys. Privacy against a malicious Alice, colluding with Eve, is based on the fact that Alice cannot make out which of the strings $\boldsymbol{S}_0,\boldsymbol{S}_1$ was the input string $\boldsymbol{S}$. Specifically, suppose $\Phi$ is a binary random variable such that $\boldsymbol{S}_{\Phi} = \boldsymbol{S}$. Then, conditioned on the combined views of Alice and Eve, $\Phi$ is uniformly distributed. Bob, who knows $\Phi$, uses $\Phi$ to mask any leakage of information about $U$ to a malicious Alice. Note that unlike Protocol~\ref{protocol:malicious_lt}, there is no $J$ used in the current protocol. Thus, only random sets of indices $L_0,L_1$, as opposed to random tuples of indices, suffice to help prevent leaking $U$ to Alice.

\begin{algorithm*}
\floatname{algorithm}{Protocol}
\caption{Malicious Users, $\epsilon_1 > 1/2$}
\label{protocol:malicious_gt}

Parameters : \begin{minipage}[t]{0.8\linewidth}
 \begin{itemize}
   \item $\beta \in [0, 1-\epsilon_1) \cap \mathbb{Q}$, $H(\beta) \in \mathbb{R} \backslash \mathbb{Q}$. Such a $\beta$ value, arbitrarily close to $(1 - \epsilon_1)$, exists as a consequence of Lemma~\ref{lem:desired_beta_hbeta}.
   \item $\delta = (1 - \epsilon_1 - \beta)$
   \item $\delta' \in (0,1)$ is such that $(\epsilon_1\epsilon_2 - 3 \delta - \delta') > 0$, $(\epsilon_1\epsilon_2 - 3 \delta - \delta') \in \mathbb{Q}$
    \item $\beta n, \beta n(\epsilon_1\epsilon_2 - 3 \delta - \delta') \in \mathbb{N}$, $<\log \comb{n}{\beta n}> \; \longrightarrow 1$ as $n \longrightarrow \infty$. Such a sequence of natural numbers is possible using Lemma~\ref{lem:log_dense}
   \item The rate\footnotemark of the protocol is $\beta (\epsilon_1\epsilon_2 - 3 \delta - \delta')$
   \item $\mathcal{I} = \{1,2,\ldots,n\}$
   \item $m = \lceil \log (\comb{n}{\beta n}) \rceil$
   \item $\mathcal{T} = \{A \subset \mathcal{I}: |A| = \beta n\}$
   \item $\mathcal{B} \subset \{0,1\}^m$ such that $|\mathcal{B}| = |\mathcal{T}|$, $\mathcal{B}^c = \{0,1\}^m \backslash \mathcal{B}$
   \item $Q : \mathcal{B} \longrightarrow \mathcal{T}$ is a bijective map
 \end{itemize}
\end{minipage}

\begin{multicols}{2}
\begin{algorithmic}[1]

\STATE Alice transmits an $n$-tuple $\boldsymbol{X}$ of i.i.d. Bernoulli($1/2$) bits over the channel.

\STATE \label{step:abort1_0} Bob receives the $n$-tuple $\boldsymbol{Y}$ from BEC($\epsilon_1$). Bob forms the sets
\begin{align*}
\overline{E} & := \{ i \in \{1,2,\ldots,n\}: Y_i \neq \bot\} \\
E & := \{ i \in \{1,2,\ldots,n\}: Y_i = \bot\}
\end{align*}

If $|\overline{E}| < \beta n$, Bob aborts the protocol.

\STATE Bob defines the collection of \emph{good} sets as:
\[ \mathcal{T}_G := \{ G \subset \overline{E} : |G| = \beta n \} \]

Let $\mathcal{B}_G = Q^{-1}(\mathcal{T}_G)$, where 
\[  Q^{-1}(\mathcal{T}_G) := \{  b \in \mathcal{B}: Q(b) \in \mathcal{T}_G \} \]

\STATE Bob picks a $m$-bit string $\boldsymbol{S} \in \mathcal{B}_G \cup \mathcal{B}^c$ as follows:
\[ P[\boldsymbol{S} = \boldsymbol{s}] = \left\{  \begin{array}{ll} \frac{1}{|\mathcal{B}_G|} \cdot \frac{|\mathcal{B}|}{2^m}, & \boldsymbol{s} \in \mathcal{B}_G \\ \frac{1}{|\mathcal{B}^c|} \cdot \left(1 - \frac{|\mathcal{B}|}{2^m} \right), & \boldsymbol{s} \in \mathcal{B}^c \\ 0, & \text{otherwise} \end{array} \right.  \]

Note that the channel acts independently on each bit transmitted by Alice and Alice does not know the erasure pattern seen by Bob. Thus, conditioned on Alice's view, $\boldsymbol{S}$ is uniform over all $m$-bit strings.

\STATE \label{step:abort1_1} Bob and Alice initiate interactive hashing with the input to interactive hashing being the string $\boldsymbol{S}$ held by Bob. As a result, both Alice and Bob receive $\boldsymbol{S}_0,\boldsymbol{S}_1 \in \{0,1\}^m$ as the output of interactive hashing. If either $\bold{S}_0 \in \mathcal{B}^c$ or $\boldsymbol{S}_1 \in \mathcal{B}^c$, then Alice and Bob abort the protocol. Otherwise, let $\Phi \in \{0,1\}$ such that $\boldsymbol{S}_{\Phi} = \boldsymbol{S}$ and let
\begin{align*}
L_0 & = Q(\boldsymbol{S}_0) \\
L_1 & = Q(\boldsymbol{S}_1)
\end{align*}

\STATE \label{step:abort1_2} If we have  
\[ \beta - \delta \leq \frac{1}{\beta n}|L_0 \cap L_1| \leq \beta + \delta \] 
then the protocol moves to the next step. Otherwise, Alice and Bob abort the protocol.

\STATE \label{step:reveal_common_part1} Bob reveals $\Theta=\Phi \oplus U$ and $\boldsymbol{Y}|_{L_0 \cap L_1}$ over the public channel.

\STATE \label{step:abort1_3} Alice checks to see that $\boldsymbol{Y}|_{L_0 \cap L_1}$ exactly matches $\boldsymbol{X}|_{L_0 \cap L_1}$, otherwise Alice aborts the protocol. 

\STATE Alice randomly and independently selects functions $F_0,F_1$ from a family $\mathcal{F}$ of universal$_2$ hash functions:
\[ F_0,F_1: \{0,1\}^{\beta n - |L_0 \cap L_1|} \longrightarrow \{0,1\}^{\beta n (\epsilon_1\epsilon_2 - 3 \delta - \delta')} \]
Alice finally sends the following information over the public channel:
\[ F_0, \; F_1, \; \boldsymbol{K}_0 \oplus F_0(\boldsymbol{X}|_{L_{\Theta} \backslash L_0 \cap L_1}), \; \boldsymbol{K}_1 \oplus F_1(\boldsymbol{X}|_{L_{\overline{\Theta}} \backslash L_0 \cap L_1}) \]

\STATE Bob knows $F_U$ and $\boldsymbol{Y}|_{L_{\Phi} \backslash L_0 \cap L_1}$ and can, therefore, recover $\boldsymbol{K}_U$.

\end{algorithmic}
\end{multicols}
\end{algorithm*}

\footnotetext{The $\beta$ can be chosen arbitrarily close to $(1 - \epsilon_1)$ and the $\delta'$ can be chosen to be arbitrarily small so that this rate can take any desired value less than $\epsilon_1 \epsilon_2 (1 - \epsilon_1)$.}

We prove the following lemma which, essentially, proves Theorem~\ref{thm:result_2p_malicious_wtap}.

\begin{lemma}
\label{lem:ach_rate_malicious_wtap}

Let $R < \left\{ \begin{array}{lr} \epsilon_1 \epsilon_2, & \epsilon_1 \leq \frac{1}{2}\\ \epsilon_1 \epsilon_2 (1 - \epsilon_1), & \epsilon_1 > \frac{1}{2} \end{array} \right\}$. Then, there exists a sequence of protocols $(\mathcal{P}_n)_{n \in \mathbb{N}}$, with corresponding rates $(r_n)_{n \in \mathbb{N}}$ such that $r_n \longrightarrow R$ and:

\begin{enumerate}[(a)]
\item When Alice and Bob are both honest, $\mathcal{P}_n$ aborts with vanishing probability and (\ref{eqn:ach_2p_wtap_0})-(\ref{eqn:ach_2p_wtap_3}) are satisfied for $(\mathcal{P}_n)_{n \in \mathbb{N}}$, as $n \longrightarrow \infty$. \label{lem:honest_rate_malicious_wtap}

\item When Alice is malicious and colludes with Eve and Bob is honest, let $V_n$ be the final view of a malicious Alice colluding with Eve at the end of $\mathcal{P}_n$. Then, $I(U ; V_n) \longrightarrow 0$ for $(\mathcal{P}_n)_{n \in \mathbb{N}}$, as $n \longrightarrow \infty$. \label{lem:malicious_alice_lemma}

\item When Alice is honest and Bob is malicious and colludes with Eve, let $V_n$ be the final view of a malicious Bob colluding with Eve at the end of $\mathcal{P}_n$. Then, $\min \{ I(\boldsymbol{K}_0 ; V_n),  I(\boldsymbol{K}_1 ; V_n)\} \longrightarrow 0$ as $n \longrightarrow \infty$. \label{lem:malicious_bob_lemma}

\end{enumerate}
\end{lemma}

This lemma is proved in Appendix~\ref{appndx:malicious_wtap_all_proofs}. A short sketch of its proof now follows. The protocol sequence $(\mathcal{P}_n)_{n \in \mathbb{N}}$ we consider is a sequence of Protocol~\ref{protocol:malicious_lt} instances when $\epsilon_1 \leq 1/2$ and of Protocol~\ref{protocol:malicious_gt} instances otherwise.

\begin{itemize}
\item In the statement of Lemma~\ref{lem:ach_rate_malicious_wtap}(\ref{lem:honest_rate_malicious_wtap}), Alice and Bob are assumed to be honest. When that is the case, we show that the numerous checks in $\mathcal{P}_n$ pass with high probability. The arguments showing that the checks pass w.h.p. use the Chernoff's bound, properties of interactive hashing or disjoint construction of sets/tuples depending on the particular check. Since all checks pass with high probability, effectively, these checks cease to matter in $\mathcal{P}_n$. We argue that in this case, $\mathcal{P}_n$ is essentially the same as Protocol~\ref{protocol:C2P} and $r_n \longrightarrow R$ as $n \longrightarrow \infty$. Specifically, in $\mathcal{P}_n$, just like in Protocol~\ref{protocol:C2P}, Alice creates two keys both secret from Eve and only one known to Bob. Alice uses these keys to encrypt her strings. As shown previously for Protocol~\ref{protocol:C2P}, such keys are sufficient for the protocol sequence to satisfy (\ref{eqn:ach_2p_wtap_0})-(\ref{eqn:ach_2p_wtap_3}) in this setup with honest-but-curious users. 

\item In the statement of Lemma~\ref{lem:ach_rate_malicious_wtap}(\ref{lem:malicious_alice_lemma}), it is assumed that Bob is honest and Alice is malicious and possibly colludes with Eve. The goal in $\mathcal{P}_n$ is to prevent such an Alice from learning $U$. In $\mathcal{P}_n$, Bob and Alice initiate interactive hashing where Bob holds an input for interactive hashing and the two outputs of interactive hashing are received by both Alice and Bob. One of these outputs is the same as the input held by Bob. The guarantee is that Alice cannot make out which of the two outputs is the one Bob held as input to interactive hashing. Specifically, suppose $\Phi$ is a binary random variable indicating which of the two outputs was the input to interactive hashing. Then, we show that conditioned on the combined views of Alice and Eve, $\Phi$ is uniformly distributed. Bob, of course, knows $\Phi$ and uses it to mask any leakage of information about $U$. As a result, Alice and Eve together cannot learn anything about $U$. 

\item In the statement of Lemma~\ref{lem:ach_rate_malicious_wtap}(\ref{lem:malicious_bob_lemma}), it is assumed that Alice is honest and Bob is malicious and possibly colludes with Eve. The goal in $\mathcal{P}_n$ is to prevent such a Bob from learning a non-negligible amount of information about both of Alice's strings. When $\epsilon_1 \leq 1/2$, a malicious Bob may swap some of the \emph{good} indices (unerased in $\boldsymbol{Y}$ or $\boldsymbol{Z}$ or both) from $\boldsymbol{L}_U$ with the \emph{bad} indices (erased in both  $\boldsymbol{Y}$ and  $\boldsymbol{Z}$) of $\boldsymbol{L}_{\overline{U}}$, to gain information about both of Alice's strings. This leaves both $\boldsymbol{L}_U, \boldsymbol{L}_{\overline{U}}$ with a large number of \emph{bad} indices. After Bob and Alice complete the interactive hashing, both of them receive as outputs some subsets $J_{\Phi}, J_{\overline{\Phi}}$ where $J_{\Phi}$ is the same as the subset Bob held as the input to interactive hashing. In a check imposed by Alice, Bob is asked to reveal the bits of $\boldsymbol{X}$ at indices $\boldsymbol{L}_U|_{J_{\overline{\Phi}}}, \boldsymbol{L}_{\overline{U}}|_{J_{\Phi}}$. An honest Bob knows the required bits, by design in $\mathcal{P}_n$. However, this check creates a problem for malicious Bob for the following reason. By a property of interactive hashing, Bob cannot influence the choice of $J_{\overline{\Phi}}$. If Bob has behaved maliciously, we show that w.h.p several of the indices in $\boldsymbol{L}_U|_{J_{\overline{\Phi}}}$ will be  the \emph{bad} indices of $\boldsymbol{L}_U$. As a result, w.h.p. malicious Bob cannot reveal all the bits sent by Alice at the indices $\boldsymbol{L}_U|_{J_{\overline{\Phi}}}$ and, therefore, fails this check. When $\epsilon_1 > 1/2$, both Alice and Bob receive the sets $L_0, L_1$ as the output of interactive hashing. Interactive hashing guarantees that w.h.p. at least one of the sets $L_0,L_1$ has fewer than a threshold number of \emph{good} indices. Through a two-step process, this threshold number of \emph{good} indices are effectively removed from both $L_0,L_1$. As a result, at least one out of $L_0,L_1$ has, effectively, no \emph{good} indices left at the end. Thus, at least one of keys created by Alice will be unknown to a malicious Bob. Consequently, malicious Bob cannot gain any information about at least one of Alice's strings. 
\end{itemize}


\section{Independent Oblivious Transfers over a broadcast channel : Proof of Theorem~\ref{thm:result_2p_hbc_indep}}
\label{sec:proofs_2}

There are three users Alice, Bob and Cathy in this setup (see Figure~\ref{fig:ot_hbc_indep}). The goal is to achieve independent OTs, with $2$-privacy, between Alice-Bob and Alice-Cathy. Specifically, we show that the rate-region $\mathcal{R}$ of independent pairs of OTs, with $2$-privacy, for honest-but-curious Alice, Bob and Cathy is such that
\[  \mathcal{R}_{\text{inner}} \subseteq \mathcal{R} \subseteq \mathcal{R}_{\text{outer}}  \]

where 
\begin{align*}
\mathcal{R}_{\text{inner}} = \Big\{ (R_B , R_C ) & \in \mathbb{R}_+^2 : 
 R_B \leq \epsilon_2  \min \{ \epsilon_1, 1 - \epsilon_1\},\\
 R_C &\leq \epsilon_1  \min \{ \epsilon_2, 1 - \epsilon_2\}, \\
R_B + R_C  &\leq  \epsilon_2\cdot \min\{\epsilon_1, 1 - \epsilon_1\}  \\ & \quad + \epsilon_1\cdot \min\{\epsilon_2, 1 - \epsilon_2\} \\
& \quad - \min\{\epsilon_1, 1 - \epsilon_1\} \cdot \min\{\epsilon_2, 1 - \epsilon_2\} \Big\}.
\end{align*}

and 
\begin{align*}
\mathcal{R}_{\text{outer}} = \Big\{
(R_B,R_C) \in \mathbb{R}_+^2 : R_B &\leq  \epsilon_2 \cdot \min \{ \epsilon_1, 1 - \epsilon_1\}, \\
                                 R_C &\leq  \epsilon_1 \cdot \min \{ \epsilon_2, 1 - \epsilon_2\}, \\
                           R_B + R_C &\leq \epsilon_1 \cdot \epsilon_2 \Big\}.
\end{align*}


\subsection{Proof of inner bound : $\mathcal{R}_{\text{inner}} \subseteq \mathcal{R}$}

It suffices to show that any rate pair $(r_B,r_C)$, with $r_B < C_{2P}$ and $r_C < \max\{0, (2\epsilon_1 - 1) \cdot \min\{\epsilon_2, 1 - \epsilon_2\} \}$, is an achievable $2$-private rate-pair. An analogous argument, with the roles of Bob and Cathy reversed, will show that any rate-pair $(r_B,r_C)$, with $r_B <  \max\{0, (2\epsilon_2 - 1) \cdot \min\{\epsilon_1, 1 - \epsilon_1\} \}$ and $r_C < C_{2P}$, is also an achievable $2$-private rate-pair. Coupled with a time-sharing argument, this proves the inner bound $\mathcal{R}_{\text{inner}} \subseteq \mathcal{R}$. Consequently, we describe a protocol (Protocol~\ref{protocol:C2P_symmetric}) for achieving any rate-pair $(r_B,r_C)$ when $r_B < C_{2P}, r_C < \max\{0, (2\epsilon_1 - 1) \cdot \min\{\epsilon_2, 1 - \epsilon_2\} \}$, in the setup of Figure~\ref{fig:ot_hbc_indep}. A sequence of Protocol~\ref{protocol:C2P_symmetric} instances, with rate-pair approaching $(r_B,r_C)$, is shown to satisfy (\ref{eqn:ach_2p_indep_0})-(\ref{eqn:ach_2p_indep_7}). This establishes that $(r_B,r_C)$ is an achievable $2$-private rate.

Our protocol has two distinct phases. The first phase is the same as Protocol~\ref{protocol:C2P}, achieving a rate $r_B < C_{2P}$ of OT for Bob with $2$-privacy. If $\epsilon_1 > 1/2$, a second phase begins after the first phase ends. This second phase is the two-party OT protocol of \cite{ot2007} (the two users being Alice and Cathy), which runs over the segment of Alice's transmissions that remained unused during the first phase. Note that this unused segment of Alice's transmissions is completely erased for Bob and is about $(2\epsilon_1 - 1)n$ bits long. This second phase, thus, achieves an OT rate of $r_C < \max\{0, (2\epsilon_1 - 1) \cdot \min\{\epsilon_2, 1 - \epsilon_2\} \}$, with $2$-privacy, for Cathy.

\begin{algorithm*}
\floatname{algorithm}{Protocol}
\caption{Protocol for achieving any rate pair $(r_B,r_C)$ such that $r_B < C_{2P}, r_C < \max\{0, (2\epsilon_1 - 1)\min\{\epsilon_2, 1 - \epsilon_2\}\}$}
\label{protocol:C2P_symmetric}

Parameters : \begin{minipage}[t]{0.8\linewidth}
 \begin{itemize}
  \item $\delta \in (0,1) $ such that $r_B < (\epsilon_2 - \delta)(\min \{\epsilon_1, 1 - \epsilon_1\} - \delta)$ and $(\epsilon_2 - \delta) \in \mathbb{Q}$
  \item $0 < \tilde{\delta} < r_B$,  $\tilde{\delta} \in \mathbb{Q}$
  \item $\beta = \frac{r_B}{\epsilon_2 - \delta}$
  \item $\beta n, n(r_B - \tilde{\delta}) \in \mathbb{N}$
  \item Bob's rate\footnotemark in the Protocol is $(r_B - \tilde{\delta})$
 \end{itemize}
\end{minipage}

\begin{multicols}{2}
\begin{algorithmic}[1]

\STATE Alice transmits an $n$-tuple $\boldsymbol{X}$ of i.i.d. Bernoulli($1/2$) bits over the channel.

\STATE \label{step:2p_indep_abort0} Bob receives the $n$-tuple $\boldsymbol{Y}$ from BEC($\epsilon_1$). Bob forms the sets
\begin{align*}
\overline{E} & := \{ i \in \{1,2,\ldots,n\}: Y_i \neq \bot\} \\
E & := \{ i \in \{1,2,\ldots,n\}: Y_i = \bot\}
\end{align*}

If $|\overline{E}| < \beta n$ or $|E| < \beta n$, Bob aborts the protocol.

\STATE Bob creates the following sets:
\begin{align*}
& L_U  \thicksim \text{Unif}\{A \subset \overline{E} : |A| = \beta n\} \\
& L_{\overline{U}}  \thicksim \text{Unif}\{A \subset E : |A| = \beta n\} \\
\text{If   } & \epsilon_1 > \frac{1}{2} \\
& L \thicksim \text{Unif}\{ A \subset E \backslash L_{\overline{U}} : |A| = (\epsilon_1 - \delta - \beta)n \} \\
\text{Else } & \\
& L = \emptyset
\end{align*}

Bob reveals $L_0,L_1,L$ to Alice over the public channel.

\STATE Alice randomly and independently chooses functions $F_0,F_1$ from a family $\mathcal{F}$ of universal$_2$ hash functions:
\[ F_0,F_1 : \{0,1\}^{\beta n} \longrightarrow \{0,1\}^{n(r_B - \tilde{\delta})} \]

Alice finally sends the following information on the public channel:
\[ F_0, \; F_1, \; \boldsymbol{K}_0 \oplus F_0(\boldsymbol{X}|_{L_0}), \; \boldsymbol{K}_1 \oplus F_1(\boldsymbol{X}|_{L_1}) \]

\STATE Bob knows $F_U,\boldsymbol{X}|_{L_U}$ and can, therefore, recover $\boldsymbol{K}_U$.

\STATE If $L \neq \emptyset$, Alice and Cathy follow the two-party OT protocol \cite{ot2007} over $\boldsymbol{X}|_L$, to obtain OT for Cathy at rate $r_C$.

\end{algorithmic}
\end{multicols}
\end{algorithm*}

\footnotetext{The parameters $\delta, \tilde{\delta}$ can be chosen to be arbitrarily small so that Bob's rate can take any desired value less than $C_{2P}$. The two-party OT protocol ensures that Cathy's rate $r_C$ can take any desired value less than $\max\{0, (2\epsilon_1 - 1)\min\{\epsilon_2, 1 - \epsilon_2\}\}$.}

\begin{lemma}
\label{lem:c2p_ach_indep}
Any rate-pair $(r_B,r_C)$, such that $r_B < C_{2P}, r_C < \max\{0, (2\epsilon_1 - 1)\min\{\epsilon_2, 1 - \epsilon_2\}\}$, is an achievable $2$-private rate-pair, with honest-but-curious users, for the setup of Figure~\ref{fig:ot_hbc_indep}
\end{lemma}

This lemma is proved in Appendix~\ref{appndx:proof_ach_2p_hbc_indep}.



\subsection{Proof of outer bound : $\mathcal{R} \subseteq \mathcal{R}_{\text{outer}}$}

We show that our outer-bound holds under a weaker privacy requirement, wherein the left-hand-sides of (\ref{eqn:ach_2p_indep_2}), (\ref{eqn:ach_2p_indep_6}) and (\ref{eqn:ach_2p_indep_7}) are multiplied by $1/n$. Let $(r_B,r_C)$ be an achievable $2$-private rate pair, for the setup in Figure~\ref{fig:ot_hbc_indep}. Then, the following are straightforward upperbounds as a consequence of Theorem~\ref{thm:result_2p_1p_hbc_wtap}:
\begin{align*}
r_B & \leq \epsilon_2 \cdot \min \{\epsilon_1, 1 - \epsilon_1\} \\
r_C & \leq \epsilon_1 \cdot \min \{\epsilon_2, 1 - \epsilon_2\}
\end{align*} 

To prove that $r_B + r_C \leq \epsilon_1 \epsilon_2$, we use the following lemma (proved in Appendix \ref{appndx:proof_small_quant_2p_hbc_indep}):

\begin{lemma}
\label{lem:small_quant_2p_hbc_indep}
\[
\frac{1}{n}H(\boldsymbol{K}_0,\boldsymbol{K}_1,\boldsymbol{J}_0,\boldsymbol{J}_1 | \boldsymbol{X},\boldsymbol{\Lambda}) \longrightarrow 0 \text{ as } n \longrightarrow \infty
\]
\end{lemma}

Intuitively, this lemma says that anyone observing Alice's interface to the rest of the system, namely signals $\boldsymbol{X},\boldsymbol{\Lambda}$, should be able to recover all the four strings $\boldsymbol{K}_0,\boldsymbol{K}_1,\boldsymbol{J}_0,\boldsymbol{J}_1$. Suppose this was not true and, say, $\boldsymbol{K}_0$ cannot be decoded from $\boldsymbol{X},\boldsymbol{\Lambda}$. In this case, Alice will infer that Bob wanted $\boldsymbol{K}_1$, that is $U = 1$, which violates (\ref{eqn:ach_2p_indep_3}). Similarly, if $\boldsymbol{J}_1$ cannot be decoded from $\boldsymbol{X},\boldsymbol{\Lambda}$, Alice will infer that Cathy wanted $\boldsymbol{J}_0$, that is $W = 0$, which violates (\ref{eqn:ach_2p_indep_4}).

Let $\tilde{E} := \{ i \in \{1,2,\ldots,n\}: Y_i = \bot \text{ and } Z_i = \bot \}$. Let $\tilde{e}$ denote a realization of $\tilde{E}$ and let $\overline{\tilde{e}} = \{1,2,\ldots,n\} \backslash \tilde{e}$. Now,
\begin{align*}
m_B + m_C & = H(\boldsymbol{K}_{\overline{U}}, \boldsymbol{J}_{\overline{W}}) \\
          & = I(\boldsymbol{K}_{\overline{U}}, \boldsymbol{J}_{\overline{W}} ; \boldsymbol{X},\boldsymbol{\Lambda},U,W) + H(\boldsymbol{K}_{\overline{U}}, \boldsymbol{J}_{\overline{W}} \mid \boldsymbol{X},\boldsymbol{\Lambda},U,W) \\
          & \leq I(\boldsymbol{K}_{\overline{U}}, \boldsymbol{J}_{\overline{W}} ; \boldsymbol{X},\boldsymbol{\Lambda},U,W) + H(\boldsymbol{K}_U, \boldsymbol{K}_{\overline{U}}, \boldsymbol{J}_W, \boldsymbol{J}_{\overline{W}} \mid \boldsymbol{X},\boldsymbol{\Lambda},U,W) \\
          & = I(\boldsymbol{K}_{\overline{U}}, \boldsymbol{J}_{\overline{W}} ; \boldsymbol{X},\boldsymbol{\Lambda},U,W) + H(\boldsymbol{K}_0, \boldsymbol{K}_1, \boldsymbol{J}_0, \boldsymbol{J}_1 \mid \boldsymbol{X},\boldsymbol{\Lambda},U,W) \\
          & \leq I(\boldsymbol{K}_{\overline{U}}, \boldsymbol{J}_{\overline{W}} ; \boldsymbol{X},\boldsymbol{\Lambda},U,W) + H(\boldsymbol{K}_0, \boldsymbol{K}_1, \boldsymbol{J}_0, \boldsymbol{J}_1 \mid \boldsymbol{X},\boldsymbol{\Lambda}) \\
          & \stackrel{\text{(a)}}{=} I(\boldsymbol{K}_{\overline{U}}, \boldsymbol{J}_{\overline{W}} ; \boldsymbol{X},\boldsymbol{\Lambda},U,W) + o(n) \\
          & \leq I(\boldsymbol{K}_{\overline{U}}, \boldsymbol{J}_{\overline{W}} ; \boldsymbol{X},\boldsymbol{\Lambda},U,W, \tilde{E}) + o(n) \\
          & \stackrel{\text{(b)}}{=} I(\boldsymbol{K}_{\overline{U}}, \boldsymbol{J}_{\overline{W}} ; \boldsymbol{X},\boldsymbol{\Lambda},U,W \mid \tilde{E}) + o(n) \\
          & = \sum_{\tilde{e} \subseteq \{1,2, \ldots,n\} } p_{\tilde{E}}(\tilde{e}) I(\boldsymbol{K}_{\overline{U}}, \boldsymbol{J}_{\overline{W}} ; \boldsymbol{X},\boldsymbol{\Lambda}, U, W \mid \tilde{E} = \tilde{e}) + o(n) \\
          & = \sum_{\tilde{e} \subseteq \{1,2, \ldots,n\} } p_{\tilde{E}}(\tilde{e}) I(\boldsymbol{K}_{\overline{U}}, \boldsymbol{J}_{\overline{W}} ; \boldsymbol{X}|_{\overline{\tilde{e}}},\boldsymbol{\Lambda}, U, W \; \mathlarger{\mid} \; \tilde{E} = \tilde{e}) + \sum_{\tilde{e} \subseteq \{1,2,\ldots,n\} } p_{\tilde{E}}(\tilde{e}) I(\boldsymbol{K}_{\overline{U}}, \boldsymbol{J}_{\overline{W}} ; \boldsymbol{X}|_{\tilde{e}}   \;\; \mathlarger{\mid} \; \;  \boldsymbol{X}|_{\overline{\tilde{e}}}, \boldsymbol{\Lambda}, U, W, \tilde{E} = \tilde{e}) \\ & \quad + o(n) \\
          & \leq \sum_{\tilde{e} \subseteq \{1,2,\ldots,n\} } p_{\tilde{E}}(\tilde{e}) I(\boldsymbol{K}_{\overline{U}}, \boldsymbol{J}_{\overline{W}} ; \boldsymbol{X}|_{\overline{\tilde{e}}},\boldsymbol{\Lambda}, U, W \; \mathlarger{\mid} \; \tilde{E} = \tilde{e}) + \sum_{\tilde{e} \subseteq \{1,2,\ldots,n\} } p_{\tilde{E}}(\tilde{e}) H(\boldsymbol{X}|_{\tilde{e}} \; \mathlarger{\mid} \; \tilde{E} = \tilde{e}) + o(n) \\
          & \leq I(\boldsymbol{K}_{\overline{U}}, \boldsymbol{J}_{\overline{W}} ; \boldsymbol{Y},\boldsymbol{Z}, \boldsymbol{\Lambda}, U,W \mid \tilde{E}) + \sum_{\tilde{e} \subseteq \{1,2,\ldots,n\} } p_{\tilde{E}}(\tilde{e}) |\tilde{e}| + o(n) \\
          &  \stackrel{\text{(c)}}{=}  I(\boldsymbol{K}_{\overline{U}}, \boldsymbol{J}_{\overline{W}} ; \boldsymbol{Y},\boldsymbol{Z}, \boldsymbol{\Lambda}, U,W, \tilde{E}) +  n\epsilon_1\epsilon_2  + o(n) \\
          & \stackrel{\text{(d)}}{=}  I(\boldsymbol{K}_{\overline{U}}, \boldsymbol{J}_{\overline{W}} ; \boldsymbol{Y},\boldsymbol{Z}, \boldsymbol{\Lambda}, U,W) + n\epsilon_1\epsilon_2 + o(n) \\
          & \stackrel{\text{(e)}}{=} n\epsilon_1\epsilon_2 + o(n)
\end{align*}
where (a) follows from Lemma~\ref{lem:small_quant_2p_hbc_indep}, (b) and (c) follow since $\tilde{E}$ is independent of $(\boldsymbol{K}_0,\boldsymbol{K}_1,\boldsymbol{J}_0,\boldsymbol{J}_1,U,W)$, (d) follows since $\tilde{E}$ is a function of $(\boldsymbol{Y},\boldsymbol{Z})$ and (e) follows from a weakened version (multiplication by $1/n$) of (\ref{eqn:ach_2p_indep_2}). As a result,
\begin{align*}
r_B + r_C & = \frac{m_B}{n} + \frac{m_C}{n} \\
               & \leq \epsilon_1 \epsilon_2 + \frac{o(n)}{n}
\end{align*}


\section{Oblivious Transfer Over a Degraded Wiretapped Channel : Proof of Theorem~\ref{thm:result_1p_hbc_degraded}}
\label{sec:proofs_3}

In this setup (see Figure~\ref{fig:ot_hbc_degraded}), Alice is connected to Bob and Eve by a broadcast channel made up of a cascade of two independent BECs. There is a BEC($\epsilon_1$) connecting Alice to Bob and a BEC($\epsilon_2$) connecting Bob to Eve. The goal is to achieve OT between Alice and Bob, with $1$-privacy. For the $1$-private OT capacity $C_{1P}$, we show that :
\[  \min \left\{\frac{1}{3}\epsilon_2(1 - \epsilon_1), \epsilon_1 \right\} \leq C_{1P} \leq \min\{\epsilon_2(1 - \epsilon_1), \epsilon_1 \}  \]


\subsection{Proof of lower bound: $\min \left\{(1/3) \cdot \epsilon_2(1 - \epsilon_1), \epsilon_1 \right\}$}

We describe a protocol (Protocol~\ref{protocol:C1P_degraded}) for achieving any $1$-private rate $r < \min \left\{(1/3) \cdot \epsilon_2(1 - \epsilon_1), \epsilon_1 \right\}$, with honest-but-curious users, in the setup of Figure~\ref{fig:ot_hbc_degraded}. For a sequence of Protocol~\ref{protocol:C1P_degraded} instances of rate $r < \min \left\{(1/3) \cdot \epsilon_2(1 - \epsilon_1), \epsilon_1 \right\}$, we show that (\ref{eqn:ach_1p_wtap_0})-(\ref{eqn:ach_1p_wtap_3}) hold. This establishes that any $r < \min \left\{(1/3) \cdot \epsilon_2(1 - \epsilon_1), \epsilon_1 \right\}$ is an achievable $1$-private rate. All the protocols seen thus far for honest-but-curious users critically depended on the fact that the erasure patterns received by Bob and Eve (or Cathy) were independent. This is the reason why Bob could send the sets $L_0,L_1$ over the public channel and Eve (or Cathy) could not deduce $U$ from these sets. However, the present setup has a physically degraded channel, degraded in favor of Bob. If Bob sends sets $L_0,L_1$ as in previous protocols, then Eve will see that one of the sets of indices corresponds entirely to erasures in $\boldsymbol{Z}$ (the \emph{bad} set $L_{\overline{U}}$) while the other set corresponds only partially to erasures in $\boldsymbol{Z}$ (the \emph{good} set $L_U$). As a result, Eve will learn $U$ as soon as Bob sends $L_0,L_1$ over the public channel. Our protocol overcomes this problem by having Bob efficiently encrypt the sets $L_0,L_1$, using a long secret key shared with Alice and secret from Eve, before transmitting the sets on the public channel. Furthermore, one of these sets of indices corresponds entirely to erasures both in $\boldsymbol{Y}$ and $\boldsymbol{Z}$, something that was not true when Bob and Eve received independent erasure patterns. Our protocol makes use of this feature to reduce the length of the secret key it needs to encrypt Alice's strings before transmitting them to Bob over the public channel. A more detailed description of the protocol now follows.

\begin{algorithm*}
\floatname{algorithm}{Protocol}
\caption{Protocol for achieving any $r < \min \left\{(1/3) \cdot \epsilon_2(1 - \epsilon_1), \epsilon_1 \right\}$}
\label{protocol:C1P_degraded}

Parameters : \begin{minipage}[t]{0.8\linewidth}
 \begin{itemize}
  \item $\delta \in (0,1) $ such that $r < \min \{ \frac{1}{3}(\epsilon_2 - \delta)(1 - \epsilon_1 - \delta), (\epsilon_1 - \delta) \}$, $(\epsilon_2 - \delta) \in \mathbb{Q}$
  \item $0 < \tilde{\delta} < r$, $\tilde{\delta} \in \mathbb{Q}$
  \item $\beta = \frac{r - \tilde{\delta}}{\epsilon_2 - \delta}$ 
  \item $\beta n, \beta n(\epsilon_2 - \delta), 2 \beta n \left( \frac{r}{r - \tilde{\delta}} \right), \beta n (1 - (\epsilon_2 - \delta)), n(r - 2 \tilde{\delta}) \in \mathbb{N}$
  \item The rate\footnotemark of the protocol is $(r - 2\tilde{\delta})$
 \end{itemize}
\end{minipage}

\begin{multicols}{2}
\begin{algorithmic}[1]

\STATE Alice transmits an $n$-tuple $\boldsymbol{X}$ of i.i.d. Bernoulli($1/2$) bits over the channel.

\STATE \label{step:1p_degraded_abort0} Bob receives the $n$-tuple $\boldsymbol{Y}$ from BEC($\epsilon_1$). Bob forms the sets
\begin{align*}
\overline{E} & := \{ i \in \{1,2,\ldots,n\}: Y_i \neq \bot\} \\
E & := \{ i \in \{1,2,\ldots,n\}: Y_i = \bot\}
\end{align*}

If $|\overline{E}| < (1 - \epsilon_1 - \delta) n$ or $|E| < (\epsilon_1 - \delta) n$, Bob aborts the protocol.

\STATE Bob creates the following sets:
\begin{align*}
&L_U  \thicksim \text{Unif}\{A \subset \overline{E} : |A| = \beta n (\epsilon_2 - \delta) \} \\
&L_{\overline{U}}  \thicksim \text{Unif}\{A \subset E : |A| = \beta n (\epsilon_2 - \delta)\} \\
&\tilde{G}  := \overline{E} \backslash L_U \\
&\tilde{B}  := E \backslash L_{\overline{U}}
\end{align*}

Bob reveals $\tilde{G},\tilde{B}$ to Alice over the public channel.

\STATE Let $\tilde{\boldsymbol{L}} = (L_0 \cup L_1)$. Bob forms the tuple $\boldsymbol{Q} \in \{0,1\}^{2 \beta n (\epsilon_2 - \delta) }$ such that :
\begin{equation*}
Q_i = \left\{ \begin{array}{ll} 0, & \tilde{L}_i \in L_0 \\ 1, & \tilde{L}_i \in L_1 \end{array} \right.
\end{equation*}

\STATE  Bob forms the set $\tilde{G}_L$ consisting of the first $ 2\beta n \cdot  r/(r - \tilde{\delta})$ elements from $(\tilde{G})$. Bob forms the set $\tilde{G}_S$ consisting of the next $\beta n (1 - (\epsilon_2 - \delta))$ elements from $(\tilde{G})$.

\STATE Bob randomly selects a function $F_L$ from a family $\mathcal{F}_L$ of universal$_2$ hash functions, given as:
\[  F_L : \{0,1\}^{ \left( \frac{r}{r - \tilde{\delta}} \right) 2\beta n} \longrightarrow \{0,1\}^{2 \beta n (\epsilon_2 - \delta) }   \]

Bob now sends the following to Alice over the public channel : $F_L, \boldsymbol{Q} \oplus F_L(\boldsymbol{X}|_{\tilde{G}_L})$. 

\STATE Alice recovers $\boldsymbol{Q}$ from $F_L, \boldsymbol{Q} \oplus F_L(\boldsymbol{X}|_{\tilde{G}_L})$. Using $\boldsymbol{Q}, \tilde{G}$ and $\tilde{B}$, Alice recovers $L_0,L_1$. Alice now randomly and independently chooses functions $F_0,F_1$ from a family $\mathcal{F}$ of universal$_2$ hash functions, given as:
\[ F_0,F_1 : \{0,1\}^{\beta n} \longrightarrow \{0,1\}^{n(r - 2 \tilde{\delta})} \]

Alice finally sends the following information on the public channel:
\[ F_0, \; F_1, \; \boldsymbol{K}_0 \oplus F_0(\boldsymbol{X}|_{L_0 \cup \tilde{G}_S}), \; \boldsymbol{K}_1 \oplus F_1(\boldsymbol{X}|_{L_1 \cup \tilde{G}_S}) \]

\STATE Bob knows $F_U,\boldsymbol{X}|_{L_U \cup \tilde{G}_S}$ and can, therefore, recover $\boldsymbol{K}_U$.

\end{algorithmic}
\end{multicols}
\end{algorithm*}

\footnotetext{The parameters $\delta, \tilde{\delta}$ can be chosen to be arbitrarily small so that this rate can take any desired value less than $\min \left\{(1/3) \cdot \epsilon_2(1 - \epsilon_1), \epsilon_1 \right\}$.}

Alice initiates the protocol by transmitting a sequence $\boldsymbol{X}$ of $n$ i.i.d. uniform bits over the channel. Bob and Eve receive the channel outputs $\boldsymbol{Y}$ and $\boldsymbol{Z}$ respectively, where $\boldsymbol{Z}$ is an erased version of $\boldsymbol{Y}$. Bob denotes by $E$ the set of indices at which $\boldsymbol{Y}$ was erased and by $\overline{E}$ the complement of $E$. Out of $\overline{E}$, Bob uniformly at random picks up a \emph{good set} $L_U$ of cardinality about $nr$. In a similar manner, Bob picks the \emph{bad set} $L_{\overline{U}}$ out of $E$, with $|L_{\overline{U}}| = |L_U|$. Let $\tilde{G} = \overline{E} \backslash L_U$ and $\tilde{B} = E \backslash L_{\overline{U}}$. Note that $|\tilde{G}|$ is about $n(1 - \epsilon_1 - r)$. Bob reveals the set of indices $\tilde{G}$ and $\tilde{B}$ to Alice over the public channel. Out of an ordered version of the set of indices $\tilde{G}$, Bob takes the first approximately $|L_0 \cup L_1|/\epsilon_2$ elements and calls it the set $\tilde{G}_L$ and takes the next approximately $|L_0|(1 - \epsilon_2)/\epsilon_2$ elements and calls it the set $\tilde{G}_S$. For $r < \min \left\{(1/3) \cdot \epsilon_2(1 - \epsilon_1), \epsilon_1 \right\}$, $|\tilde{G}_L| + |\tilde{G}_S| \leq |\tilde{G}|$ and so the sets $\tilde{G}_L, \tilde{G}_S$ of the required sizes can be derived from the set $\tilde{G}$. The purpose of forming $\tilde{G}_L$ and $\tilde{G}_S$ is to use them to form two different secret keys, known to Alice and Bob but secret from Eve.

Bob's goal now is to transmit $L_0,L_1$ to Alice without revealing them to Eve. Towards this goal, Bob does two things: Firstly, Bob considers the ordered version $\boldsymbol{L}$ of $L_0 \cup L_1$ and forms the binary $|L_0 \cup L_1|$-tuple $\boldsymbol{Q}$ such that $Q_i = 0$ when $L_i \in L_0$ and $Q_i = 1$ when $L_i \in L_1$, $i=1,2,\ldots, |L_0 \cup L_1|$. Secondly, Bob forms a secret key using $\tilde{G}_L$ that is $|L_0 \cup L_1|$ bits long, which Alice knows and Eve does not know, as follows. Bob randomly selects a function $F_L$ from a universal$_2$ class $\mathcal{F}_L$, whose input is about $|L_0 \cup L_1|/\epsilon_2$ bits long and whose output is $|L_0 \cup L_1|$ bits long. Then, $F_L(\boldsymbol{X}|_{\tilde{G}_L})$ is the secret key Bob is looking for. Bob sends $F_L, \boldsymbol{Q} \oplus F(\boldsymbol{X}|_{\tilde{G}_L})$ to Alice over the public channel. Alice recovers $L_0,L_1$ from this message while Eve cannot separate out $L_0,L_1$ from $L_0 \cup L_1$.

Alice now forms two independent keys to encrypt its strings and send these encrypted strings to Bob. One of these keys is known to Bob and none of the keys is known to Eve. For this, Alice randomly selects two functions $F_0,F_1$ from a family $\mathcal{F}$ of universal$_2$ hash functions, whose input is about $(|L_0| + |\tilde{G}_S|)/\epsilon_2$ bits long and whose output is about $|L_0|$ bits long. Then, $F_0(\boldsymbol{X}|_{L_0 \cup \tilde{G}_S})$ and $F_1(\boldsymbol{X}|_{L_1 \cup \tilde{G}_S})$ are the keys Alice wants. Note that Bob does not know the key $F_{\overline{U}}(\boldsymbol{X}|_{L_{\overline{U}} \cup \tilde{G}_S})$ even though Bob knows $\boldsymbol{X}|_{\tilde{G}_S}$. This is a direct consequence of privacy amplification on  $\boldsymbol{X}|_{L_{\overline{U}} \cup \tilde{G}_S}$ by $F_{\overline{U}}$, coupled with the facts that $\boldsymbol{X}|_{L_{\overline{U}}}$ is erased for Bob and the key $F_{\overline{U}}(\boldsymbol{X}|_{L_{\overline{U}} \cup \tilde{G}_S})$ is about the same length as $|L_{\overline{U}}|$. Alice finally sends $F_0,F_1, \boldsymbol{K}_0 \oplus F_0(\boldsymbol{X}|_{L_0 \cup \tilde{G}_S}), \boldsymbol{K}_1 \oplus F_1(\boldsymbol{X}|_{L_1 \cup \tilde{G}_S})$ to Bob over the public channel. Bob knows $\boldsymbol{X}|_{L_U \cup \tilde{G}_S}$ and can recover $\boldsymbol{K}_U$ from Alice's message.

\begin{lemma}
\label{lem:c1p_ach_degraded}
Any rate $r < \min \left\{(1/3) \cdot \epsilon_2(1 - \epsilon_1), \epsilon_1 \right\}$ is an achievable $1$-private rate, with honest-but-curious users, for the setup of Figure~\ref{fig:ot_hbc_degraded}. 
\end{lemma}

This lemma is proved in Appendix~\ref{appndx:proof_c1p_ach_degraded}. Here we give a sketch of this proof. Let $(\mathcal{P}_n)_{\{n \in \mathbb{N}\}}$ be a sequence of Protocol~\ref{protocol:C1P_degraded} instances, of rate $r - 2 \tilde{\delta}$. If the protocol does not abort, then Bob knows the key $\boldsymbol{S}_U = F_U(\boldsymbol{X}|_{L_U \cup \tilde{G}_S})$. This is because Bob knows $F_U,L_U,\tilde{G}_S,\boldsymbol{X}|_{L_U},\boldsymbol{X}|_{\tilde{G}_S}$. As a result, Bob can recover the string $\boldsymbol{K}_U$ from Alice's public message and so (\ref{eqn:ach_1p_wtap_0}) holds for $(\mathcal{P}_n)_{\{n \in \mathbb{N}\}}$. Bob does not know the key $\boldsymbol{S}_{\overline{U}} = F_{\overline{U}}(\boldsymbol{X}|_{L_{\overline{U}} \cup \tilde{G}_S})$, despite knowing $\boldsymbol{X}|_{\tilde{G}_S}$. This is because, by design, $|\boldsymbol{S}_{\overline{U}}| = |L_{\overline{U}}|$ and Bob does not know $\boldsymbol{X}|_{L_{\overline{U}}}$. Hence, privacy amplification (Lemma~\ref{lem:privacy_amplification}) by $F_{\overline{U}}$ on $\boldsymbol{X}|_{L_{\overline{U}} \cup \tilde{G}_S}$ ensures that $\boldsymbol{S}_{\overline{U}}$ appears nearly uniformly distributed to Bob. Thus, Bob does not learn anything about $\boldsymbol{K}_{\overline{U}}$ from Alice's message and so (\ref{eqn:ach_1p_wtap_1}) holds for $(\mathcal{P}_n)_{\{n \in \mathbb{N}\}}$. (\ref{eqn:ach_1p_wtap_2}) holds for $(\mathcal{P}_n)_{\{n \in \mathbb{N}\}}$ since Alice cannot learn $U$ upon receiving $L_0,L_1$ from Bob, as in previous protocols. Finally, note that the keys $\boldsymbol{S}_U, \boldsymbol{S}_{\overline{U}}$ are independent, despite $\boldsymbol{X}|_{\tilde{G}_S}$ being a common part of the inputs to functions $F_U,F_{\overline{U}}$ that generate these keys. Furthermore, privacy amplification by $F_U,F_{\overline{U}}$ ensures that Eve knows nothing about $\boldsymbol{S}_U, \boldsymbol{S}_{\overline{U}}$. Thus, for the same reasons as in previous protocols, (\ref{eqn:ach_1p_wtap_3}) holds for $(\mathcal{P}_n)_{\{n \in \mathbb{N}\}}$.


\subsection{Proof of upper bound: $C_{1P} \leq \min\{\epsilon_2(1 - \epsilon_1), \epsilon_1 \}$}

The upper bound follows by evaluating the upper bound in (\ref{eqn:c1p_general_outer}) for the setup of Figure~\ref{fig:ot_hbc_degraded}. Intuitively, the upper bound of $\epsilon_2(1 - \epsilon_1)$ follows from the fact that OT capacity is upper bounded by the secret key capacity of the wiretapped channel. This is because if Bob runs the protocol with the choice bit set deterministically to say 0, then $\boldsymbol{K}_0$ is a secret key between Alice and Bob. The upper bound follows from the fact that $\epsilon_2(1-\epsilon_1)$ is the secret key capacity of this wiretapped channel with public discussion~\cite{sec-key1993}. The upper bound of $\epsilon_1$ follows from the fact that this is an upper bound for two-party OT capacity of the binary erasure channel with erasure probability $\epsilon_1$~\cite{ot2007}.


\section{Summary}
\label{sec:conclusions}

In this work, we formulated and studied the problem of obtaining $1$-of-$2$ string OT between two users Alice and Bob in the presence of an eavesdropper Eve. The resource for OT is a broadcast channel from Alice to Bob and Eve. Apart from the usual OT constraints between the users, we additionally require that the eavesdropper learn nothing about any users' private data. The wiretapped channel model we introduced in this study (see Figure~\ref{fig:ot_hbc_wtap_bcast}) is a generalization of the two-pary OT model studied previously \cite{ot2007}, \cite{PintoDowsMorozNasc2011}. We studied the privacy requirements in our OT problem under two distinct privacy regimes : $2$-privacy, where Eve may collude with either user and $1$-privacy where no such collusion is allowed. When the broadcast channel in the model consists of two independent and parallel BECs (see Figure~\ref{fig:ot_hbc_wtap}), we derived the OT capacity both under $2$-privacy and under $1$-privacy for honest-but-curious users. These capacity results easily generalize for $1$-of-$N$ string OT. Our protocols were extensions of the scheme presented by Ahlswede and Csisz{\'a}r \cite{ot2007}, designed to additionally guarantee privacy against Eve. The corresponding converses were generalizations of the converse arguments in \cite{ot2007}. In the same model, we studied the problem of obtaining OT when Alice and Bob can behave maliciously and the malicious user can additionally collude with Eve. For this problem, we obtained an achievable rate which is optimal when $\epsilon_1 \leq 1/2$ and is no more than a fraction $\epsilon_1$ away from optimal when $\epsilon_1 > 1/2$. Our protocol for the regime $\epsilon_1 > 1/2$ makes novel use of interactive hashing to directly obtain the keys that encrypt Alice's strings. For $\epsilon_1 \leq 1/2$, our protocol is an extension of the protocol presented in \cite{savvides_thesis}, \cite{PintoDowsMorozNasc2011}, \cite{DowsNasc2014-arxiv} and is designed to maintain privacy against the malicious user colluding with Eve. We studied a generalization of the wiretapped OT model of Figure~\ref{fig:ot_hbc_wtap}, where the eavesdropper is replaced by a legitimate user Cathy (see Figure~\ref{fig:ot_hbc_indep}). Independent OT is required between Alice-Bob and Alice-Cathy. We derived inner and outer bounds for the region of achievable rate-pairs. These bounds match except when $\epsilon_1, \epsilon_2 > 1/2$. The final OT problem we studied considers a physically degraded broadcast channel as the OT resource (see Figure~\ref{fig:ot_hbc_degraded}). OT is required between Alice and Bob with $1$-privacy. Due to the degraded nature of the channel, Eve has more information about the noise process in the legitimate users' channel compared to previous models (where Bob, Eve got independent erasure patterns). This makes it harder to guarantee privacy for Bob but also presents an opportunity for reducing the amount of secret keys needed to encrypt Alice's strings. We obtain upper and lower bounds for the OT capacity for this problem. The bounds match when $\epsilon_1 \leq (1/3) \cdot \epsilon_2 (1 - \epsilon_1)$, otherwise the lower bound is within a factor of $1/3$ of the upper bound.


\section{Discussion and open problems}
\label{sec:open_problems}

\begin{itemize}

\item In the problem of obtaining OT over a wiretapped channel with malicious users (see Section~\ref{sec:prob_statement_malicious_wtap}), our achievable $2$-private rate is $\epsilon_1 \cdot C_{2P}$ when $\epsilon_1 > 1/2$. Here, $C_{2P}$ is the $2$-private OT capacity in the same setup when users are honest-but-curious. The main reason we loose rate in our protocol is that the sets $L_0,L_1$ obtained out of interactive hashing are not disjoint. In order to obtain disjoint sets, we never use the indices $L_0 \cap L_1$. This is quite a sizeable number of indices for an honest Bob to loose out of the good set $L_U$, leading to a rate loss. If we could get interactive hashing to provide us with disjoint $L_0,L_1$, then the achievable rate can be improved. Specifically, a useful version of interactive hashing would have a subset of $\{1,2,\ldots,n\}$ as its input and would provide two disjoint subsets of $\{1,2,\ldots,n\}$ as outputs. If that is possible without loosing any property of interactive hashing, we will only have to do the required privacy amplification on $L_0,L_1$ and a higher rate for our problem will become possible. 

\item In the problem of obtaining independent OTs over a broadcast erasure channel (see Section~\ref{sec:prob_statement_indep}), there is a gap between the achievable rate and the outer bound when $\epsilon_1,\epsilon_2 > 1/2$. Our converse technique of evaluating how much information must remain hidden from any user during OT, which is applicable to any broadcast channel, does not close this gap. We believe that a more channel-specific insight on the impossibility of meeting one of the OT requirements can tighten the converse. This kind of a channel-specific converse argument was successfully employed in characterizing the $1$-private OT capacity in the presence of Eve (see Section~\ref{sec:prob_statement_wtap}) in the regime $\epsilon_2/2 \leq \epsilon_1 < 1/2$. In the independent OTs problem itself, when all three users are malicious, it is quite tempting to consider a protocol where Alice invokes two independent interactive hashing based checks. One of the checks is for Bob and the other for Cathy, to catch malicious Bob or Cathy. This is unlike the wiretapper model (Alice and Bob can act maliciously, Eve remains passive) where Alice cannot get Eve to respond to such checks. This difference is what makes the problem of catching a malicious Bob, colluding with Eve, in the wiretapper model much harder to solve. However, the two-checks approach cannot prevent attacks such as a denial-of-service attack by Bob to prevent Cathy from getting any OT, or vice-versa. Evolving a technique to prevent such attacks in our setup has been deferred to a future study.

\item A single-server private information retrieval (PIR) problem, closely related to the independent OTs problem of this paper, was formulated and studied in \cite{MishraSDPisit15}. This problem has the same setup as the independent OTs problem, except that Alice holds only a single database of $N$ strings in the PIR problem. Bob and Cathy want a string of their respective choice from this single database, with $2$-privacy. For $N=2$, the data transfer capacity for this PIR problem was derived in \cite{MishraSDPisit15}. The result uses a novel achievable scheme which is quite different from the achievable scheme used in the independent OTs problem in this paper. In fact, the achievable scheme used in the independent OTs problem of this paper turns out to be sub-optimal for the PIR problem. However, for $N > 2$, it remains open to characterize the data transfer capacity for the PIR problem.

\item Considering OT using a physically degraded channel (see Section~\ref{sec:prob_statement_degraded}) presents several open problems. We do not know the $1$-private OT capacity, with honest-but-curious users, when $(1/3) \cdot \epsilon_2(1 - \epsilon_1) < \epsilon_1$. Characterizing the OT capacity in this regime seems to require a tighter converse, based on a more channel-specific impossibility argument. We suspect that the $2$-private OT capacity in this setup is zero. However, our brief attempt at showing this has not been successful. It would be quite interesting to devise a scheme for OT with $2$-privacy here. Finally, obtaining OT when users can behave maliciously in this setup appears to require newer techniques and this problem has been deferred to a future study. 

\item In all the problems studied in this work, we have assumed unlimited public discussion. It would be interesting to study these problems when the public discussion rate is also constrained.

\end{itemize}


\section{Acknowledgements}
\label{sec:acknowledgements}

M.~Mishra gratefully acknowledges the help received from Amitava Bhattacharya, Department of Mathematics, Tata Institute of Fundamental Research (TIFR) in proving Lemma~\ref{lem:dense}.




\appendices

\section{Universal Hash Functions, R\'enyi Entropy and Privacy Amplification}
\label{appndx:privacy_amplification}

\begin{definition}
A class $\mathcal{F}$ of functions mapping $\mathcal{A} \longrightarrow \mathcal{B}$ is \emph{universal}$_2$ if, for $F \thicksim \text{Unif}(\mathcal{F})$ and for any $a_0,a_1 \in \mathcal{A}, a_0 \neq a_1$, we have
\[  P[F(a_0) = F(a_1)] \leq \frac{1}{|\mathcal{B}|}   \]
\end{definition}

The class of all linear maps from $\{0,1\}^n$ to $\{0,1\}^r$ is a universal$_2$ class. Several other examples of universal$_2$ classes of functions are given in \cite{carter_wegman_1979, carter_wegman_1981}.

\begin{definition}
Let $A$ be a random variable with alphabet $\mathcal{A}$ and distribution $p_A$. The \emph{collision probability} $P_c(A)$ of $A$ is defined as the probability that $A$ takes the same value twice in two independent experiments. That is,
\[  P_c(A) = \underset{a \in \mathcal{A}}{\sum} p^2_A(a) \]
\end{definition}

\begin{definition}
The \emph{R\'enyi entropy of order two} of a random variable $A$ is
\[  R(A) = \log_2 \left( \frac{1}{P_c(A)}  \right)  \]
\end{definition}

For an event $\mathcal{E}$, the conditional distribution $p_{A|\mathcal{E}}$ is used to define the conditional collision probability $P_c(A|\mathcal{E})$ and the conditional R\'enyi entropy of order $2$, $R(A|\mathcal{E})$.

\begin{lemma}[Corollary 4 of \cite{generalized_privacy_ampl_1995}]
\label{lem:privacy_amplification}
Let $P_{AD}$ be an arbitrary probability distribution, with $A \in \mathcal{A}, D \in \mathcal{D}$, and let $d \in \mathcal{D}$. Suppose $R(A | D = d) \geq c$. Let $\mathcal{F}$ be a universal$_2$ class of functions mapping $\mathcal{A} \longrightarrow \{0,1\}^l$ and $F \thicksim \text{Unif}(\mathcal{F})$. Then,
\begin{align*}
  H(F(A) | F, D = d) & \geq l -  \log (1 + 2^{l - c}) \\
                             & \geq l - \frac{2^{l-c}}{\ln 2}
\end{align*}
\end{lemma}


\section{Interactive hashing}
\label{appndx:interactive_hashing}

Interactive hashing is an interactive protocol between two users,  a Sender and a Receiver. The input to interactive hashing is a string $\boldsymbol{S} \in \{0,1\}^k$ available with Sender. The output of interactive hashing are two strings $\boldsymbol{S}_0,\boldsymbol{S}_1 \in \{0,1\}^k$, available to both Sender and Receiver, satisfying the following properties :

\begin{property}
\label{prop:ih_0}
$\boldsymbol{S}_0 \neq \boldsymbol{S}_1$
\end{property}

\begin{property}
\label{prop:ih_1}
$\exists \Phi \in \{0,1\}$ such that $\boldsymbol{S}_{\Phi} = \boldsymbol{S}$.
\end{property}

\begin{property}
\label{prop:ih_2}
Suppose Sender and Receiver are both honest. Then, 
\[ \boldsymbol{S}_{\overline{\Phi}} \thicksim \text{ Unif} \{ \{0,1\}^m \backslash \{\boldsymbol{S}_{\Phi}\}   \}. \]
\end{property}

\begin{property}
\label{prop:ih_3}
Suppose Sender is honest and Receiver is malicious. Let $V_R,V^{IH}_R$ be Receiver's views at the beginning and end of interactive hashing, respectively. Then, for $\boldsymbol{s}_0,\boldsymbol{s}_1 \in \{0,1\}^k$, $\boldsymbol{s}_0 \neq \boldsymbol{s}_1$,
\[  P[ \boldsymbol{S} = \boldsymbol{s}_0 | V_R ] = P[ \boldsymbol{S} = \boldsymbol{s}_1 | V_R ]  \quad \Longrightarrow  \quad P[ \boldsymbol{S} = \boldsymbol{s}_0 | V^{IH}_R, \boldsymbol{S}_0 = \boldsymbol{s}_0, \boldsymbol{S}_1 = \boldsymbol{s}_1 ] = P[\boldsymbol{S} = \boldsymbol{s}_1 | V^{IH}_R, \boldsymbol{S}_0 = \boldsymbol{s}_0, \boldsymbol{S}_1 = \boldsymbol{s}_1] = \frac{1}{2} \]
\end{property}

\begin{property}
\label{prop:ih_4}
Suppose Sender is malicious and Receiver is honest. Let $\mathcal{G} \subset \{0,1\}^k$. Then,
\[  P[\boldsymbol{S}_0,\boldsymbol{S}_1 \in \mathcal{G}] \leq 15.6805 \times \frac{|\mathcal{G}|}{2^k}  \]
\end{property}

Protocol~\ref{protocol:IH} is a protocol for interactive hashing for which the above properties were proved in \cite{savvides_thesis}.

\begin{algorithm*}
\floatname{algorithm}{Protocol}
\caption{Interactive hashing}
\label{protocol:IH}

Let $\boldsymbol{S}$ be a $k$-bit string that Sender wishes to send to Receiver. All operations mentioned here are in the binary field $\mathbb{F}_2$.

\begin{multicols}{2}
\begin{algorithmic}[1]

\STATE Receiver chooses a $(k-1) \times k$ matrix $\boldsymbol{M}$ uniformly at random from amongst all binary matrices of rank $(k-1)$. Let $\boldsymbol{\Delta}_i$ denote the $i^{th}$ row of $\boldsymbol{M}$.

\STATE For $1 \leq i \leq k-1$ do:
  \begin{enumerate}[(a)]
     \item Receiver send $\boldsymbol{\Delta}_i$ to Sender.
     \item Sender responds back with the bit $\Pi_i = \boldsymbol{\Delta}_i \cdot \boldsymbol{S}$.
  \end{enumerate}

\STATE Given $\boldsymbol{M}$ and the vector $\boldsymbol{\Pi} = (\Pi_1,\Pi_2,\ldots,\Pi_{k-1})$, Sender and Receiver compute the two solutions of the linear system $\boldsymbol{M} \cdot \boldsymbol{\chi} = \boldsymbol{\Pi}$. These solutions are labelled $\boldsymbol{S}_0,\boldsymbol{S}_1$ according to lexicographic order.
\end{algorithmic}
\end{multicols}
\end{algorithm*}


\section{Oblivious transfer over a wiretapped channel with honest-but-curious users : Proofs of Lemmas~\ref{lem:c2p_ach_wtap},~\ref{lem:c1p_ach_wtap},~\ref{lem:small_quant_1p_hbc_wtap}}

\subsection{Notations and definitions}

\begin{itemize}

\item Recall that for both Protocol~\ref{protocol:C2P} and Protocol~\ref{protocol:C1P}:
\begin{align*}
V_A & = \boldsymbol{K}_0, \boldsymbol{K}_1, \boldsymbol{X}, \boldsymbol{\Lambda} \\
V_B & = U, \boldsymbol{Y}, \boldsymbol{\Lambda} \\
V_E & = \boldsymbol{Z}, \boldsymbol{\Lambda}
\end{align*}
where 
\[ \boldsymbol{\Lambda} = L_0,L_1,F_0,F_1,\boldsymbol{K}_0 \oplus F_0(\boldsymbol{X}|_{L_0}), \boldsymbol{K}_1 \oplus F_1(\boldsymbol{X}|_{L_1}) \]

\item Let $\boldsymbol{\Psi} = (\Psi_i : i = 1,2,\ldots,n)$, where, for $i=1,2,\ldots,n$:
\begin{equation}
\label{eqn:defn_of_psi}
\Psi_i := \left\{  \begin{array}{ll} Y_i, & Y_i \neq \bot \\ Z_i, & Z_i \neq \bot \\ \bot, & \text{otherwise}  \end{array} \right.
\end{equation}

\end{itemize}


\subsection{Proof of Lemma~\ref{lem:c2p_ach_wtap}}
\label{appndx:proof_ach_2p_hbc_wtap}

In this proof, we use a sequence $(\mathcal{P}_n)_{n \in \mathbb{N}}$ of Protocol~\ref{protocol:C2P} instances of rate $(r - \tilde{\delta})$ and we show that (\ref{eqn:ach_2p_wtap_0}) - (\ref{eqn:ach_2p_wtap_3}) are satisfied for $(\mathcal{P}_n)_{n \in \mathbb{N}}$.

Let $\Upsilon$ be the event that $\mathcal{P}_n$ aborts in Step~\ref{step:2p_wtap_abort0}. Then, due to Chernoff's bound, $P[\Upsilon = 1] \longrightarrow 0$ exponentially fast as $n \longrightarrow \infty$.

\begin{enumerate}

\item To show that (\ref{eqn:ach_2p_wtap_0}) is satisfied for $(\mathcal{P}_n)_{n \in \mathbb{N}}$, we first note that
\begin{align*}
P[\hat{\boldsymbol{K}}_U \neq \boldsymbol{K}_U] & = P[\Upsilon=0] \cdot P[\hat{\boldsymbol{K}}_U \neq \boldsymbol{K}_U | \Upsilon = 0] + P[\Upsilon = 1] \cdot P[\hat{\boldsymbol{K}}_C \neq \boldsymbol{K}_U | \Upsilon = 1]
\end{align*} 

Since $P[\Upsilon=1] \rightarrow 0$ exponentially fast, it is sufficient to show that $P[\hat{\boldsymbol{K}}_U \neq \boldsymbol{K}_U | \Upsilon = 0] \longrightarrow 0$ as $n \longrightarrow \infty$.

When $\Upsilon=0$, Bob knows $\boldsymbol{X}|_{L_U}$. Since Bob also knows $F_U$, Bob can compute the key $F_U(\boldsymbol{X}|_{L_U})$. As a result, Bob learns $\boldsymbol{K}_U$ from $\boldsymbol{K}_U \oplus F_U(\boldsymbol{X}|_{L_U})$ sent by Alice. Hence, $P[\hat{\boldsymbol{K}}_U \neq \boldsymbol{K}_U | \Upsilon = 0] = 0$.

\item To show that (\ref{eqn:ach_2p_wtap_1}) is satisfied for $(\mathcal{P}_n)_{n \in \mathbb{N}}$, we note that 
\begin{align*}
I(\boldsymbol{K}_{\overline{U}} ; V_B,V_E ) & \leq I(\boldsymbol{K}_{\overline{U}} ; V_B,V_E, \Upsilon) \\
& = \sum_{j=0,1}P[\Upsilon=j] I(\boldsymbol{K}_{\overline{U}} ; V_B,V_E| \Upsilon=j) + I(\boldsymbol{K}_{\overline{U}} ; \Upsilon) .
\end{align*}
Since $P[\Upsilon=1] \rightarrow 0$ exponentially fast and $I(\boldsymbol{K}_{\overline{U}} ; \Upsilon)=0$, it is sufficient to show that $I(\boldsymbol{K}_{\overline{U}}  ;  V_B,V_E | \Upsilon = 0) \longrightarrow 0$ as $n \longrightarrow \infty$. The rest of this argument is implicitly conditioned on the event $\Upsilon = 0$, though we do not explicitly write it in the expressions below.
\begin{align*}
I(\boldsymbol{K}_{\overline{U}} ; V_B,V_E) & = I(\boldsymbol{K}_{\overline{U}} ; U,\boldsymbol{Y},\boldsymbol{Z},\boldsymbol{\Lambda}) \\
& = I(\boldsymbol{K}_{\overline{U}} ; U,\boldsymbol{Y},\boldsymbol{Z},L_0,L_1,F_0,F_1,\boldsymbol{K}_0 \oplus F_0(\boldsymbol{X}|_{L_0}), \boldsymbol{K}_1 \oplus F_1(\boldsymbol{X}|_{L_1})) \\
& = I(\boldsymbol{K}_{\overline{U}} ; U,\boldsymbol{Y},\boldsymbol{Z},L_U,L_{\overline{U}},F_U,F_{\overline{U}},\boldsymbol{K}_U \oplus F_U(\boldsymbol{X}|_{L_U}), \boldsymbol{K}_{\overline{U}} \oplus F_{\overline{U}}(\boldsymbol{X}|_{L_{\overline{U}}})) \\
& \stackrel{\text{(a)}}{=}  I(\boldsymbol{K}_{\overline{U}} ; \boldsymbol{K}_{\overline{U}} \oplus F_{\overline{U}}(\boldsymbol{X}|_{L_{\overline{U}}}) | U,\boldsymbol{Y},\boldsymbol{Z},L_U,L_{\overline{U}},F_U,F_{\overline{U}}, \boldsymbol{K}_U \oplus F_U(\boldsymbol{X}|_{L_U})) \\
& = H(\boldsymbol{K}_{\overline{U}} \oplus F_{\overline{U}}(\boldsymbol{X}|_{L_{\overline{U}}}) | U,\boldsymbol{Y},\boldsymbol{Z},L_U,L_{\overline{U}},F_U,F_{\overline{U}}, \boldsymbol{K}_U \oplus F_U(\boldsymbol{X}|_{L_U})) \\ & \quad - H( F_{\overline{U}}(\boldsymbol{X}|_{L_{\overline{U}}}) |\boldsymbol{K}_{\overline{U}}, U,\boldsymbol{Y}, \boldsymbol{Z},L_U,L_{\overline{U}},F_U,F_{\overline{U}}, \boldsymbol{K}_U \oplus F_U(\boldsymbol{X}|_{L_U})) \\
& \stackrel{\text{(b)}}{\leq} n(r - \tilde{\delta}) - H( F_{\overline{U}}(\boldsymbol{X}|_{L_{\overline{U}}}) |\boldsymbol{K}_{\overline{U}}, U,\boldsymbol{Y},\boldsymbol{Z},L_U,L_{\overline{U}},F_U,F_{\overline{U}}, \boldsymbol{K}_U \oplus F_U(\boldsymbol{X}|_{L_U})) \\
& = n(r - \tilde{\delta}) - H( F_{\overline{U}}(\boldsymbol{X}|_{L_{\overline{U}}}) | \boldsymbol{\Psi}|_{L_{\overline{U}}}, \boldsymbol{K}_{\overline{U}}, U,\boldsymbol{Y},\boldsymbol{Z},L_U,L_{\overline{U}},F_U, F_{\overline{U}},\boldsymbol{K}_U \oplus F_U(\boldsymbol{X}|_{L_U})) \\
& \stackrel{\text{(c)}}{=} n(r - \tilde{\delta}) - H( F_{\overline{U}}(\boldsymbol{X}|_{L_{\overline{U}}}) | F_{\overline{U}}, \boldsymbol{\Psi}|_{L_{\overline{U}}})
\end{align*}

where (a) follows since $\boldsymbol{K}_{\overline{U}} \independent ( U,\boldsymbol{Y},\boldsymbol{Z},L_U,L_{\overline{U}},F_U,F_{\overline{U}},       \boldsymbol{K}_U \oplus F_U(\boldsymbol{X}|_{L_U}))$, (b) follows since $F_{\overline{U}}(\boldsymbol{X}|_{L_{\overline{U}}})$ is $n(r - \tilde{\delta})$ bits long and (c) follows since $ F_{\overline{U}}(\boldsymbol{X}|_{L_{\overline{U}}}) - F_{\overline{U}},\boldsymbol{\Psi}|_{L_{\overline{U}}} - \boldsymbol{K}_{\overline{U}}, U,\boldsymbol{Y},\boldsymbol{Z},L_U,L_{\overline{U}},F_U,\boldsymbol{K}_U \oplus F_U(\boldsymbol{X}|_{L_U})$ is a Markov chain.

Now, $R(\boldsymbol{X}|_{L_{\overline{U}}} \;  \mathlarger{\mid}  \; \boldsymbol{\Psi}|_{L_{\overline{U}}}  = \boldsymbol{\psi}|_{l_{\overline{u}}})  = \#_e(\boldsymbol{\psi}|_{l_{\overline{u}}})$.  Also, whenever $\#_e(\boldsymbol{\psi}|_{l_{\overline{u}}}) \geq (\epsilon_2 - \delta) |l_{\overline{u}}| = nr$, applying Lemma~\ref{lem:privacy_amplification} we get:
\begin{align*}
H(F_{\overline{U}}(\boldsymbol{X}|_{L_{\overline{U}}}) \;  \mathlarger{\mid}  \; F_{\overline{U}}, \boldsymbol{\Psi}|_{L_{\overline{U}}} = \boldsymbol{\psi}|_{l_{\overline{u}}}) & \geq n(r - \tilde{\delta}) - \frac{2^{n(r - \tilde{\delta}) - nr}}{\ln 2} \\
& = n(r - \tilde{\delta}) - \frac{2^{-\tilde{\delta}n }}{\ln 2}
\end{align*}

We know by Chernoff's bound that $P[\#_e(\boldsymbol{\Psi}|_{L_{\overline{U}}}) \geq (\epsilon_2 - \delta) |L_{\overline{U}}|] \geq 1 - \xi$, where $\xi \longrightarrow 0$ exponentially fast as $n \longrightarrow \infty$.  Note that there is an implicit conditioning on the event $\Upsilon = 0$ here too. Thus,
\begin{align*}
 I(\boldsymbol{K}_{\overline{U}} ; V_B,V_E) & \leq n(r - \tilde{\delta}) - H( F_{\overline{U}}(\boldsymbol{X}|_{L_{\overline{U}}}) \;  \mathlarger{\mid}  \; F_{\overline{U}}, \boldsymbol{\Psi}|_{L_{\overline{U}}}) \\
& \leq  n (r - \tilde{\delta}) - (1 - \xi) \left(  n(r - \tilde{\delta}) -  \frac{2^{-\tilde{\delta}n}}{\ln 2} \right) \\
& = \xi n(r - \tilde{\delta}) + (1 - \xi) \cdot \frac{2^{-\tilde{\delta}n}}{\ln 2}
\end{align*}

Thus, $I(\boldsymbol{K}_{\overline{U}} ; V_B,V_E) \longrightarrow 0$ as $n \longrightarrow \infty$.

\item To show that (\ref{eqn:ach_2p_wtap_2}) is satisfied for $(\mathcal{P}_n)_{n \in \mathbb{N}}$, we note that
\begin{align*}
I(U ; V_A,V_E ) & \leq I(U ; V_A, V_E, \Upsilon) \\
    & = \sum_{j=0,1}P[\Upsilon = j] I(U ; V_A,V_E| \Upsilon = j) + I(U ; \Upsilon)
\end{align*}
Since $P[\Upsilon = 1] \rightarrow 0$ exponentially fast and $I(U ; \Upsilon) = 0$, it is sufficient
to show that $I(U  ;  V_A,V_E | \Upsilon = 0) \longrightarrow 0$ as $n \longrightarrow \infty$. The rest of this argument is implicitly conditioned on the event $\Upsilon = 0$, though we do not explicitly write it in the expressions below.
\begin{align*}
I(U ; V_A,V_E) & = I(U ; \boldsymbol{K}_0, \boldsymbol{K}_1, \boldsymbol{X}, \boldsymbol{Z}, \boldsymbol{\Lambda}) \\
& = I(U ; \boldsymbol{K}_0, \boldsymbol{K}_1, \boldsymbol{X}, \boldsymbol{Z}, L_0,L_1,F_0,F_1,\boldsymbol{K}_0 \oplus F_0(\boldsymbol{X}|_{L_0}), \boldsymbol{K}_1 \oplus F_1(\boldsymbol{X}|_{L_1})) \\
& = I(U ; \boldsymbol{K}_0, \boldsymbol{K}_1, \boldsymbol{X}, \boldsymbol{Z}, L_0,L_1,F_0,F_1, F_0(\boldsymbol{X}|_{L_0}), F_1(\boldsymbol{X}|_{L_1})) \\
& \stackrel{\text{(a)}}{=} I(U ; \boldsymbol{X}, \boldsymbol{Z}, L_0,L_1,F_0,F_1, F_0(\boldsymbol{X}|_{L_0}), F_1(\boldsymbol{X}|_{L_1})) \\
& \stackrel{\text{(b)}}{=} I(U ; \boldsymbol{X}, \boldsymbol{Z}, L_0,L_1) \\
& \stackrel{\text{(c)}}{=} I(U ; L_0,L_1) \\
& \stackrel{\text{(d)}}{=} 0
\end{align*}

where (a) follows since $\boldsymbol{K}_0, \boldsymbol{K}_1 \independent (U, \boldsymbol{X}, \boldsymbol{Z}, L_0,L_1,F_0,F_1, F_0(\boldsymbol{X}|_{L_0}), F_1(\boldsymbol{X}|_{L_1}))$, (b) follows since $U - \boldsymbol{X}, \boldsymbol{Z}, L_0,L_1 - F_0,F_1, F_0(\boldsymbol{X}|_{L_0}), F_1(\boldsymbol{X}|_{L_1})$ is a Markov chain, (c) follows since $U - L_0,L_1 - \boldsymbol{X}, \boldsymbol{Z}$ is a Markov chain and (d) follows since the channel acts independently on each input bit and $|L_0| = |L_1|$.

\item To show that (\ref{eqn:ach_2p_wtap_3}) is satisfied for $(\mathcal{P}_n)_{n \in \mathbb{N}}$, we note that
\begin{align*}
I(\boldsymbol{K}_0,\boldsymbol{K}_1,U ; V_E) & \leq I(\boldsymbol{K}_0,\boldsymbol{K}_1,U ; V_E, \Upsilon) \\
 & = \sum_{j=0,1}P[\Upsilon=j] I(\boldsymbol{K}_0,\boldsymbol{K}_1,U ; V_E | \Upsilon = j) \\ & \qquad + I(\boldsymbol{K}_0,\boldsymbol{K}_1,U ; \Upsilon)
\end{align*}

Since $P[\Upsilon = 1] \longrightarrow 0$ exponentially fast and $I(\boldsymbol{K}_0,\boldsymbol{K}_1,U ; \Upsilon) = 0$, it is sufficient to show that $I(\boldsymbol{K}_0,\boldsymbol{K}_1,U  ;  V_E | \Upsilon = 0) \longrightarrow 0$ as $n \longrightarrow \infty$. The rest of this argument is implicitly conditioned on the event $\Upsilon = 0$, though we do not explicitly write it in the expressions below.
\begin{align*}
 I(\boldsymbol{K}_0,\boldsymbol{K}_1,U  ; V_E) & = I(U ; V_E) +  I(\boldsymbol{K}_0,\boldsymbol{K}_1 ; V_E | U) \\
& = I(U ; V_E) +  I(\boldsymbol{K}_0,\boldsymbol{K}_1 ; \boldsymbol{Z}, \boldsymbol{\Lambda} | U) \\
& = I(U ; V_E) +  I(\boldsymbol{K}_0,\boldsymbol{K}_1 ; \boldsymbol{Z}, L_0, L_1, F_0, F_1,  \boldsymbol{K}_0 \oplus F_0(\boldsymbol{X}|_{L_0}), \boldsymbol{K}_1 \oplus F_1(\boldsymbol{X}|_{L_1}) | U) \\
& = I(U ; V_E) +  I(\boldsymbol{K}_0,\boldsymbol{K}_1 ; \boldsymbol{K}_0 \oplus F_0(\boldsymbol{X}|_{L_0}), \boldsymbol{K}_1 \oplus F_1(\boldsymbol{X}|_{L_1})  |   U ,  \boldsymbol{Z}, L_0, L_1, F_0, F_1) \\
& = I(U ; V_E) +  H( \boldsymbol{K}_0 \oplus F_0(\boldsymbol{X}|_{L_0}), \boldsymbol{K}_1 \oplus F_1(\boldsymbol{X}|_{L_1}) | U, \boldsymbol{Z}, L_0, L_1, F_0, F_1) \\ & \quad -  H(F_0(\boldsymbol{X}|_{L_0}), F_1(\boldsymbol{X}|_{L_1}) | \boldsymbol{K}_0, \boldsymbol{K}_1, U, \boldsymbol{Z}, L_0, L_1, F_0, F_1) \\
& \stackrel{\text{(a)}}{\leq} I(U ; V_E) +  2n(r - \tilde{\delta}) -  H(F_0(\boldsymbol{X}|_{L_0}), F_1(\boldsymbol{X}|_{L_1}) | \boldsymbol{K}_0, \boldsymbol{K}_1,  U, \boldsymbol{Z}, L_0, L_1, F_0, F_1) \\
& = I(U ; V_E) +  2n(r - \tilde{\delta}) -  H(F_0(\boldsymbol{X}|_{L_0}) | \boldsymbol{K}_0, \boldsymbol{K}_1, U, \boldsymbol{Z}, L_0, L_1, F_0, F_1) \\ & \quad - H(F_1(\boldsymbol{X}|_{L_1}) | F_0(\boldsymbol{X}|_{L_0}), \boldsymbol{K}_0, \boldsymbol{K}_1, U, \boldsymbol{Z}, L_0, L_1, F_0, F_1) \\
& \stackrel{\text{(b)}}{=} I(U ; V_E) +  2n(r - \tilde{\delta}) -  H(F_0(\boldsymbol{X}|_{L_0}) | F_0, \boldsymbol{Z}|_{L_0}) - H(F_1(\boldsymbol{X}|_{L_1}) | F_1, \boldsymbol{Z}|_{L_1})
\end{align*}

where (a) follows since both $F_0(\boldsymbol{X}|_{L_0})$ and $F_1(\boldsymbol{X}|_{L_1})$ are $n(r - \tilde{\delta})$ bits each and (b) follows since $F_0(\boldsymbol{X}|_{L_0}) - F_0, \boldsymbol{Z}|_{L_0} - \boldsymbol{K}_0, \boldsymbol{K}_1, U, \boldsymbol{Z}, L_0, L_1, F_1$ and $F_1(\boldsymbol{X}|_{L_1}) - F_1, \boldsymbol{Z}|_{L_1} - F_0(\boldsymbol{X}|_{L_0}), \boldsymbol{K}_0, \boldsymbol{K}_1, U, \boldsymbol{Z}, L_0, L_1, F_0$ are Markov chains.

Now, $R(\boldsymbol{X}|_{L_0} \;  \mathlarger{\mid}  \; \boldsymbol{Z}|_{L_0} = \boldsymbol{z}|_{l_0})  = \#_e(\boldsymbol{z}|_{l_0})$. Also, whenever $\#_e(\boldsymbol{z}|_{l_0})  \geq (\epsilon_2 - \delta)|l_0|  = nr$,  then applying Lemma~\ref{lem:privacy_amplification}, we get:
\begin{align*}
H(F_0(\boldsymbol{X}|_{L_0}) \;  \mathlarger{\mid}  \; F_0, \boldsymbol{Z}|_{L_0} = \boldsymbol{z}|_{l_0}) & \geq n(r - \tilde{\delta}) - \frac{2^{n(r - \tilde{\delta}) - nr}}{\ln 2} \\
& = n(r - \tilde{\delta}) - \frac{2^{-\tilde{\delta}n}}{\ln 2}
\end{align*}

We know from Chernoff's bound that $P[\#_e(\boldsymbol{Z}|_{L_0}) \geq (\epsilon_2 - \delta)|L_0|] \geq 1 - \xi$, where $\xi \longrightarrow 0$ exponentially fast as $n \longrightarrow \infty$. Note that there is an implicit conditioning on the event $\Upsilon = 0$ here too. Thus ,
\begin{align*}
 I(K_0,K_1,U ; V_E)  & \leq I(U ; V_E) +  2n(r - \tilde{\delta}) -  H(F_0(\boldsymbol{X}|_{L_0}) \;  \mathlarger{\mid}  \; F_0, \boldsymbol{Z}|_{L_0}) - H(F_1(\boldsymbol{X}|_{L_1}) \;  \mathlarger{\mid}  \; F_1, \boldsymbol{Z}|_{L_1}) \\
& \leq I(U ; V_E) +     2n(r - \tilde{\delta})  -   (1 - \xi) \cdot 2 \left( n(r - \tilde{\delta})  - \frac{2^{-\tilde{\delta}n}}{\ln 2} \right) \\
& = I(U ; V_E) +    2 \xi n (r - \tilde{\delta})  +  2 (1 - \xi) \frac{2^{-\tilde{\delta}n}}{\ln 2}
\end{align*}

The first term above goes to $0$ since (\ref{eqn:ach_2p_wtap_2}) holds. Hence, $I(K_0,K_1,U ; V_E) \longrightarrow 0$ as $n \longrightarrow \infty$.

\end{enumerate}


\subsection{Proof of Lemma~\ref{lem:c1p_ach_wtap}}
\label{appndx:proof_ach_1p_hbc_wtap}

In this proof, we use a sequence $(\mathcal{P}_n)_{n \in \mathbb{N}}$ of Protocol~\ref{protocol:C1P} instances of rate $(r - \tilde{\delta})$ and we show that (\ref{eqn:ach_1p_wtap_0}) - (\ref{eqn:ach_1p_wtap_3}) are satisfied for $(\mathcal{P}_n)_{n \in \mathbb{N}}$.

Let $\Upsilon$ be the event that $\mathcal{P}_n$ aborts in Step~\ref{step:1p_wtap_abort0}. Then, due to Chernoff's bound, $P[\Upsilon = 1] \longrightarrow 0$ exponentially fast as $n \longrightarrow \infty$.

\begin{enumerate}

\item To show that (\ref{eqn:ach_1p_wtap_0}) is satisfied for $(\mathcal{P}_n)_{n \in \mathbb{N}}$, the proof is the same as that for showing that (\ref{eqn:ach_2p_wtap_0}) holds for Protocol~\ref{protocol:C2P} and is, therefore, omitted.

\item To show that (\ref{eqn:ach_1p_wtap_1}) is satisfied for $(\mathcal{P}_n)_{n \in \mathbb{N}}$, it suffices to show (as in the proof of Lemma~\ref{lem:c2p_ach_wtap}) that $I(\boldsymbol{K}_{\overline{U}}  ;  V_B | \Upsilon = 0) \longrightarrow 0$ as $n \longrightarrow \infty$. The rest of this argument is implicitly conditioned on the event $\Upsilon = 0$.
\begin{align*}
 I(K_{\overline{U}} ; V_B) & = I(K_{\overline{U}} ; U, \boldsymbol{Y}, \boldsymbol{\Lambda}) \\
& = I(K_{\overline{U}} ; U, \boldsymbol{Y}, L_0,L_1,F_0,F_1,\boldsymbol{K}_0 \oplus F_0(\boldsymbol{X}|_{L_0}), \boldsymbol{K}_1 \oplus F_1(\boldsymbol{X}|_{L_1})) \\
& = I(K_{\overline{U}} ; U, \boldsymbol{Y}, L_U,L_{\overline{U}},F_U,F_{\overline{U}},\boldsymbol{K}_U \oplus F_U(\boldsymbol{X}|_{L_U}),  \boldsymbol{K}_{\overline{U}} \oplus F_{\overline{U}}(\boldsymbol{X}|_{L_{\overline{U}}})) \\
& \stackrel{\text{(a)}}{=} I(K_{\overline{U}} ; \boldsymbol{K}_{\overline{U}} \oplus F_{\overline{U}}(\boldsymbol{X}|_{L_{\overline{U}}}) | U, \boldsymbol{Y}, L_U,L_{\overline{U}},F_U,F_{\overline{U}},  \boldsymbol{K}_U \oplus F_U(\boldsymbol{X}|_{L_U})) \\
& = H(\boldsymbol{K}_{\overline{U}} \oplus F_{\overline{U}}(\boldsymbol{X}|_{L_{\overline{U}}}) | U, \boldsymbol{Y}, L_U,L_{\overline{U}},F_U,F_{\overline{U}}, \boldsymbol{K}_U \oplus F_U(\boldsymbol{X}|_{L_U}) ) \\ & \quad - H(F_{\overline{U}}(\boldsymbol{X}|_{L_{\overline{U}}}) | K_{\overline{U}}, U, \boldsymbol{Y}, L_U,L_{\overline{U}}, F_U,F_{\overline{U}},\boldsymbol{K}_U \oplus F_U(\boldsymbol{X}|_{L_U})) \\
& \stackrel{\text{(b)}}{\leq} n(r - \tilde{\delta}) - H(F_{\overline{U}}(\boldsymbol{X}|_{L_{\overline{U}}}) | K_{\overline{U}}, U, \boldsymbol{Y}, L_U,L_{\overline{U}},F_U,F_{\overline{U}}, \boldsymbol{K}_U \oplus F_U(\boldsymbol{X}|_{L_U})) \\
& \stackrel{\text{(c)}}{=} n(r - \tilde{\delta}) - H(F_{\overline{U}}(\boldsymbol{X}|_{L_{\overline{U}}}) | F_{\overline{U}}, \boldsymbol{Y}|_{L_{\overline{U}}}) \\
\end{align*}

where (a) follows since $ K_{\overline{U}} \independent (U, \boldsymbol{Y}, L_U,L_{\overline{U}},F_U,F_{\overline{U}},  \boldsymbol{K}_U \oplus F_U(\boldsymbol{X}|_{L_U})) $, (b) follows since $F_{\overline{U}}(\boldsymbol{X}|_{L_{\overline{U}}})$ is $n(r - \tilde{\delta})$ bits long and (c) follows since $F_{\overline{U}}(\boldsymbol{X}|_{L_{\overline{U}}}) - F_{\overline{U}}, \boldsymbol{Y}|_{L_{\overline{U}}} - K_{\overline{U}}, U, \boldsymbol{Y}, L_U,L_{\overline{U}},F_U, \boldsymbol{K}_U \oplus F_U(\boldsymbol{X}|_{L_U}) $ is a Markov chain. 

Now,
\begin{align*}
R(\boldsymbol{X}|_{L_{\overline{U}}} | \boldsymbol{Y}|_{L_{\overline{U}}} = \boldsymbol{y}|_{l_{\overline{u}}} ) & = \#_e(\boldsymbol{y}|_{l_{\overline{u}}}) \\
& \geq nr
\end{align*}

since the construction of $L_{\overline{U}}$ contains at least $nr$ positions that are erased for Bob. Applying Lemma~\ref{lem:privacy_amplification}, we get:
\begin{align*}
H(F_{\overline{U}}(\boldsymbol{X}|_{L_{\overline{U}}}) | F_{\overline{U}}, \boldsymbol{Y}|_{L_{\overline{U}}} = \boldsymbol{y}|_{l_{\overline{u}}} ) & \geq n(r - \tilde{\delta}) - \frac{2^{n(r - \tilde{\delta}) - nr}}{\ln 2} \\
& = n(r - \tilde{\delta}) - \frac{2^{-\tilde{\delta} n}}{\ln 2}
\end{align*}

As a result,
\begin{align*}
 I(K_{\overline{U}} ; V_B) & \leq n(r - \tilde{\delta}) - H(F_{\overline{U}}(\boldsymbol{X}|_{L_{\overline{U}}}) | F_{\overline{U}}, \boldsymbol{Y}|_{L_{\overline{U}}}) \\
& \leq \frac{2^{-\tilde{\delta} n}}{\ln 2}
\end{align*}

Thus, $I(K_{\overline{U}} ; V_B) \longrightarrow 0$ as $n \longrightarrow \infty$.

\item To show that (\ref{eqn:ach_1p_wtap_2}) is satisfied for $(\mathcal{P}_n)_{n \in \mathbb{N}}$, it suffices to show (as in the proof of Lemma~\ref{lem:c2p_ach_wtap}) that $I(U  ;  V_A | \Upsilon = 0) \longrightarrow 0$ as $n \longrightarrow \infty$. The rest of this argument is implicitly conditioned on the event $\Upsilon = 0$.
\begin{align*}
 I(U ; V_A) & = I(U ; \boldsymbol{K}_0, \boldsymbol{K}_1, \boldsymbol{X}, \boldsymbol{\Lambda}) \\
& = I(U ; \boldsymbol{K}_0, \boldsymbol{K}_1, \boldsymbol{X}, L_0,L_1,F_0,F_1,\boldsymbol{K}_0 \oplus F_0(\boldsymbol{X}|_{L_0}), \boldsymbol{K}_1 \oplus F_1(\boldsymbol{X}|_{L_1})) \\
& = I(U ; \boldsymbol{K}_0, \boldsymbol{K}_1, \boldsymbol{X}, L_0,L_1,F_0,F_1, F_0(\boldsymbol{X}|_{L_0}), F_1(\boldsymbol{X}|_{L_1})) \\
& = I(U ; \boldsymbol{K}_0, \boldsymbol{K}_1, \boldsymbol{X}, L_0,L_1,F_0,F_1) \\
& \stackrel{\text{(a)}}{=} I(U ; \boldsymbol{X}, L_0,L_1) \\
& \stackrel{\text{(b)}}{=} I(U ; L_0,L_1) \\
& \stackrel{\text{(c)}}{=} 0
\end{align*}

where (a) follows since $\boldsymbol{K}_0, \boldsymbol{K}_1, F_0, F_1 \independent (U , \boldsymbol{X}, L_0,L_1)$, (b) follows since $\boldsymbol{X} \independent (U,L_0,L_1)$ and (c) follows since the channel acts independently on each input bit and $|L_0| = |L_1|$ .

\item To show that (\ref{eqn:ach_1p_wtap_3}) is satisfied for $(\mathcal{P}_n)_{n \in \mathbb{N}}$, the proof is the same as that for showing (\ref{eqn:ach_2p_wtap_3}) holds for Protocol~\ref{protocol:C2P}.

\end{enumerate}


\subsection{Proof of Lemma~\ref{lem:small_quant_1p_hbc_wtap}}
\label{appndx:proof_small_quant_1p_hbc_wtap}

We need two lemmas from \cite{ot2007}, which are stated here for completeness.

\begin{lemma}[\cite{ot2007}]
\label{lem:lemma3_ahl_csis}
Let A,B,C denote random variables with values in finite sets $\mathcal{A}$, $\mathcal{B}$ and $\mathcal{C}$ respectively. Suppose $c_1,c_2 \in \mathcal{C}$ with $P[C = c_1] = p > 0$ and $P[C = c_2] = q > 0$. Then,
\begin{IEEEeqnarray*}{rCl}
 | H(A|B, C = c_1) & - & H(A|B, C = c_2) |  \leq 3 \sqrt{\frac{(p+q)\ln 2}{2pq}I(A,B ; C)} log|\mathcal{A}| + 1.
\end{IEEEeqnarray*}
\end{lemma}

\begin{lemma}[\cite{ot2007}, Lemma 2.2 of \cite{sec-key1993}]
\label{lem:cond_indep_hbc_wtap}
\[ I(\boldsymbol{K}_0, \boldsymbol{K}_1 ;  U, \boldsymbol{Y} | \boldsymbol{X},\boldsymbol{\Lambda}) = 0 \]
\end{lemma}

Note that (\ref{eqn:ach_1p_wtap_2}) and Lemma~\ref{lem:lemma3_ahl_csis} together imply
\begin{align*}
H(\boldsymbol{K}_0 | \boldsymbol{X},\boldsymbol{\Lambda}, U = 0) - H(\boldsymbol{K}_0 | \boldsymbol{X},\boldsymbol{\Lambda}, U = 1) & = o(n) \\
H(\boldsymbol{K}_1 | \boldsymbol{X},\boldsymbol{\Lambda}, U = 0) - H(\boldsymbol{K}_1 | \boldsymbol{X},\boldsymbol{\Lambda}, U = 1) & = o(n)
\end{align*}

Multiplying both equations by $1/2$ and subtracting, we get
\begin{equation}
\label{eqn:lemma3_infr3}
H(\boldsymbol{K}_U | \boldsymbol{X},\boldsymbol{\Lambda}, U) - H(\boldsymbol{K}_{\overline{U}} | \boldsymbol{X},\boldsymbol{\Lambda}, U) = o(n).
\end{equation}

Lemma~\ref{lem:cond_indep_hbc_wtap} implies that $I(\boldsymbol{K}_0,\boldsymbol{K}_1 ; U | \boldsymbol{X},\boldsymbol{\Lambda}) = 0$. Hence,
\begin{align*}
H(\boldsymbol{K}_0,\boldsymbol{K}_1 | \boldsymbol{X}, \boldsymbol{\Lambda}) &=  H(\boldsymbol{K}_0,\boldsymbol{K}_1 | \boldsymbol{X}, \boldsymbol{\Lambda}, U) \\
                     &=  H(\boldsymbol{K}_U,\boldsymbol{K}_{\overline{U}} | \boldsymbol{X}, \boldsymbol{\Lambda}, U) \\
                     &=  H(\boldsymbol{K}_U | \boldsymbol{X}, \boldsymbol{\Lambda}, U) + H(\boldsymbol{K}_{\overline{U}} | \boldsymbol{X},\boldsymbol{\Lambda}, U, \boldsymbol{K}_U) \\
                     &\leq  H(\boldsymbol{K}_U | \boldsymbol{X}, \boldsymbol{\Lambda}, U) + H(\boldsymbol{K}_{\overline{U}} | \boldsymbol{X},\boldsymbol{\Lambda}, U).
\end{align*}

In light of (\ref{eqn:lemma3_infr3}), this lemma will be proved if we show either $H(\boldsymbol{K}_U | \boldsymbol{X}, \boldsymbol{\Lambda}, U)$ or $H(\boldsymbol{K}_{\overline{U}} | \boldsymbol{X},\boldsymbol{\Lambda}, U)$ to be $o(n)$.

For this we note that Lemma~\ref{lem:cond_indep_hbc_wtap} implies \[I(\boldsymbol{K}_0,\boldsymbol{K}_1 ; \boldsymbol{Y} | \boldsymbol{X},\boldsymbol{\Lambda}, U) = 0.\] This, in turn, implies that \[I(\boldsymbol{K}_U,\boldsymbol{K}_{\overline{U}} ; \boldsymbol{Y} | \boldsymbol{X},\boldsymbol{\Lambda}, U) = 0.\] Hence, $I(\boldsymbol{K}_U ; \boldsymbol{Y} | \boldsymbol{X}, \boldsymbol{\Lambda}, U) = 0$. Therefore,
\begin{align*}
H(\boldsymbol{K}_U | \boldsymbol{X}, \boldsymbol{\Lambda}, U) &= H(\boldsymbol{K}_U | \boldsymbol{X},\boldsymbol{\Lambda}, U, \boldsymbol{Y}) \\
                   &\stackrel{\text{(a)}}{=}  H(\boldsymbol{K}_U | \boldsymbol{X},\boldsymbol{\Lambda}, U, \boldsymbol{Y}, \hat{\boldsymbol{K}}_U) \\
                   &\leq  H(\boldsymbol{K}_U | \hat{\boldsymbol{K}}_U) \\
                   &\stackrel{\text{(b)}}{=}  o(n), 
\end{align*}

where (a) follows from the fact that $\hat{\boldsymbol{K}}_U$ is a function of $(U, \boldsymbol{Y}, \boldsymbol{\Lambda})$, and (b) from (\ref{eqn:ach_1p_wtap_0}) and Fano's inequality.


\section{Oblivious transfer over a wiretapped channel with malicious users : Proof of Lemma~\ref{lem:ach_rate_malicious_wtap}}
\label{appndx:malicious_wtap_all_proofs}

For this proof, the protocol sequence $(\mathcal{P}_n)_{n \in \mathbb{N}}$ we consider is a sequence of Protocol~\ref{protocol:malicious_lt} instances when $\epsilon_1 \leq 1/2$ and of Protocol~\ref{protocol:malicious_gt} instances otherwise. The rate of each Protocol~\ref{protocol:malicious_lt} instance is $(\epsilon_1 \epsilon_2 - 5 \delta - 2 \tilde{\delta} - \delta')$ and the rate of each  Protocol~\ref{protocol:malicious_gt} instance is $(1 - \epsilon_1 - \delta)(\epsilon_1 \epsilon_2 - 3 \delta - \delta')$. We show that $(\mathcal{P}_n)_{n \in \mathbb{N}}$ satisfies the conditions required in the statement of Lemma~\ref{lem:ach_rate_malicious_wtap}.

\subsection{Notation and definitions}

\begin{itemize}

\item The definition of $\boldsymbol{\Psi}$ is given in (\ref{eqn:defn_of_psi}).

\item Let $\boldsymbol{M}$ be the matrix chosen by Alice and let $\boldsymbol{\Pi}$ be the corresponding bit sequence Bob sends to Alice during interactive hashing.

\item Let the view of Eve, just before Alice sends the encrypted strings and hash functions, be $V_E$. Then,
\[ V_E = \left\{  \hspace{-0.2cm} \begin{array}{ll}   (\boldsymbol{Z},\boldsymbol{L}_0,\boldsymbol{L}_1, \boldsymbol{M}, \boldsymbol{\Pi},\Theta,\boldsymbol{Y}|_{\boldsymbol{L}_0|_{\boldsymbol{J}_{\overline{\Theta}}}}, \boldsymbol{Y}|_{\boldsymbol{L}_1|_{\boldsymbol{J}_{\Theta}}}), & \epsilon_1 \leq \frac{1}{2} \\ (\boldsymbol{Z}, \boldsymbol{M},\boldsymbol{\Pi},\Theta,\boldsymbol{Y}|_{L_0 \cap L_1}), & \epsilon_1 > \frac{1}{2}      \end{array}     \right. \]

\item Let the combined views of Bob and Eve, just before Alice sends the encrypted strings and hash functions, be $V_{BE}$. Then,
\[  V_{BE} = \left\{  \hspace{-0.2cm}   \begin{array}{ll}   (U,\boldsymbol{Y},\boldsymbol{Z},\boldsymbol{L}_0,\boldsymbol{L}_1,\boldsymbol{S},\boldsymbol{M}), & \epsilon_1 \leq \frac{1}{2}, \text{Bob honest} \\   (\boldsymbol{Y},\boldsymbol{Z},\boldsymbol{L}_0,\boldsymbol{L}_1,\boldsymbol{M},\boldsymbol{\Pi},\Theta), & \epsilon_1 \leq \frac{1}{2}, \text{Bob malicious} \\ (U,\boldsymbol{Y},\boldsymbol{Z},\boldsymbol{S},\boldsymbol{M}), & \epsilon_1 > \frac{1}{2}, \text{Bob honest} \\  (\boldsymbol{Y},\boldsymbol{Z},\boldsymbol{M},\boldsymbol{\Pi},\Theta), & \epsilon_1 > \frac{1}{2}, \text{Bob malicious} \\ \end{array}   \right.   \]
Note that if the output of interactive hashing corresponding to the input string $\boldsymbol{S}$ are the strings $\boldsymbol{S}_0,\boldsymbol{S}_1$, then $\boldsymbol{S}_0,\boldsymbol{S}_1$ are functions of $\boldsymbol{S}, \boldsymbol{M}$ as well as functions of $\boldsymbol{M}, \boldsymbol{\Pi}$.

\end{itemize}


\subsection{Proof of Lemma~\ref{lem:ach_rate_malicious_wtap}(\ref{lem:honest_rate_malicious_wtap})}


\subsubsection{$\mathcal{P}_n$ aborts with vanishing probability} \*
\label{sec:protocol_doesnt_abort_for_hbc}

\underline{$\mathbf{\epsilon_1 \leq 1/2}$}: 

The protocol can abort at steps (\ref{step:abort0_0}), (\ref{step:abort0_1}) or (\ref{step:abort0_2}). We show that when Alice and Bob are honest, each of these aborts happens only with vanishing probability. 

\begin{itemize}

\item In step (\ref{step:abort0_0}), we note that $\beta + \gamma = 1 - \epsilon_1 - \delta - \tilde{\delta}$. As a consequence of Chernoff's bound, $\#_{\overline{e}}(\boldsymbol{Y})  \geq (\beta + \gamma)n$ w.h.p.. Similarly, since $\beta -\gamma = \epsilon_1 - \delta$, then w.h.p., $\#_e(\boldsymbol{Y}) \geq (\beta - \gamma)n$. Thus, an abort happens in this step with only vanishing probability.

\item In step (\ref{step:abort0_1}), an abort never happens since $\{\boldsymbol{L}_0\},\{\boldsymbol{L}_1\}$ are disjoint by construction.

\item In step (\ref{step:abort0_2}), we note that the strings $\boldsymbol{Y}|_{\boldsymbol{L}_0|_{\boldsymbol{J}_{\overline{\Theta}}}}$, $\boldsymbol{Y}|_{\boldsymbol{L}_1|_{\boldsymbol{J}_{\Theta}}}$ are, in fact, the strings $\boldsymbol{Y}|_{\boldsymbol{L}_U|_{\boldsymbol{J}_{\overline{\Phi}}}}$, $\boldsymbol{Y}|_{\boldsymbol{L}_{\overline{U}}|_{\boldsymbol{J}_{\Phi}}}$. By construction, $\#_e(\boldsymbol{Y}|_{\boldsymbol{L}_U|_{\boldsymbol{J}_{\overline{\Phi}}}}) = 0$ and  $\#_e(\boldsymbol{Y}|_{\boldsymbol{L}_{\overline{U}}|_{\boldsymbol{J}_{\Phi}}}) = 0$. So, Bob correctly reveals the bits $\boldsymbol{X}|_{\boldsymbol{L}_U|_{\boldsymbol{J}_{\overline{\Phi}}}}$ and $\boldsymbol{X}|_{\boldsymbol{L}_{\overline{U}}|_{\boldsymbol{J}_{\Phi}}}$. Hence, an abort never happens in this step since Alice's check in this step always passes.

\end{itemize}

\underline{$\mathbf{\epsilon_1 > 1/2}$}: 

The protocol can abort at steps (\ref{step:abort1_0}), (\ref{step:abort1_1}), (\ref{step:abort1_2}) or (\ref{step:abort1_3}). We show that when Alice and Bob are honest, each of these aborts happens only with vanishing probability. 

\begin{itemize}
\item In step (\ref{step:abort1_0}), we note that $\beta = 1 - \epsilon_1 - \delta$. As a consequence of Chernoff's bound, $\#_{\overline{e}}(\boldsymbol{Y})  \geq \beta n$ w.h.p. Thus, an abort happens in this step with only vanishing probability.

\item In step (\ref{step:abort1_1}), an abort happens when either $\boldsymbol{S} \in \mathcal{B}^c$ or $\boldsymbol{S}_{\overline{\Phi}} \in \mathcal{B}^c$. 
Now, $P[\boldsymbol{S} \in \mathcal{B}^c] = 1 - |\mathcal{B}|/2^m$. As a consequence of Lemma~\ref{lem:log_dense} (in Appendix~\ref{appndx:supporting_lemmas}), the fractional part of $\log |\mathcal{B}| = \log |\mathcal{T}| = \log (\comb{n}{\beta n})$ converges to $1$ over an appropriate choice of a sequence of natural numbers. As a result, $1 - |\mathcal{B}|/2^m$ can be made as small as desired by choosing a sufficiently large $n$ from this sequence.
Similarly, due to Property~\ref{prop:ih_2} of interactive hashing, the string $\boldsymbol{S}_{\overline{\Phi}}$ is uniformly distributed over all strings other than $\boldsymbol{S}$. As a result, $P[\boldsymbol{S}_{\overline{\Phi}} \in \mathcal{B}^c] = |\mathcal{B}^c|/(2^m - 1) = (2^m - |\mathcal{B}|)/(2^m - 1)$ which can be made arbitrarily small, again as a consequence of Lemma~\ref{lem:log_dense}.

\item To see that the protocol aborts in step (\ref{step:abort1_2}) only with vanishing probability, we begin by noting that when Alice and Bob are honest, interactive hashing guarantees that $\boldsymbol{S}_{\overline{\Phi}} \thicksim \text{Unif} \{ \boldsymbol{s} \in \{0,1\}^m : \boldsymbol{s} \neq  \boldsymbol{S} \}$. Since the protocol did not abort in step (\ref{step:abort1_1}), this implies that $\boldsymbol{S}_{\overline{\Phi}} \thicksim \text{Unif}\{ \boldsymbol{s} \in \mathcal{B} : \boldsymbol{s} \neq \boldsymbol{S}\}$. Since $Q$ is a bijective map, $L_{\overline{\Phi}}$ is uniform over $\mathcal{T} \backslash \{ L_{\Phi} \}$. Let $l_{\phi}$ be a specific realization of $L_{\Phi}$. Lemma~\ref{lem:overlap_within_bounds_corr} (in Appendix~\ref{appndx:supporting_lemmas}, applied with $k = n, \varphi = \beta, \rho = \delta, \upsilon_0 = l_{\phi}, \Upsilon_1 = L_{\overline{\Phi}}$) proves that the probability with which the overlap size $|L_{\Phi} \cap L_{\overline{\Phi}}| = |L_0 \cap L_1|$ is outside the specified bounds falls exponentially in $n$.

\item The protocol never aborts in step (\ref{step:abort1_3}) when users are honest. This is because $\#_e(\boldsymbol{Y}|_{L_{\Phi}}) = 0$. Since $L_0 \cap L_1 = L_{\Phi} \cap L_{\overline{\Phi}} \subset L_{\Phi} $, $\boldsymbol{Y}|_{L_0 \cap L_1} = \boldsymbol{X}|_{L_0 \cap L_1}$. Hence, Alice's check in step  (\ref{step:abort1_3}) always passes when users are honest.
\end{itemize}


\subsubsection{(\ref{eqn:ach_2p_wtap_0})-(\ref{eqn:ach_2p_wtap_3}) hold} \*

The rate $r_n$ of Protocol $\mathcal{P}_n$ is:
\[  r_n = \left\{ \begin{array}{ll} \epsilon_1 \epsilon_2 - 5\delta - 2\tilde{\delta}  - \delta', & \epsilon_1 \leq \frac{1}{2} \\ (1 - \epsilon_1 - \delta)(\epsilon_1 \epsilon_2 - 3\delta - \delta'), & \epsilon_1 > \frac{1}{2} \end{array} \right. \]

To show that (\ref{eqn:ach_2p_wtap_0})-(\ref{eqn:ach_2p_wtap_3}) hold over $(\mathcal{P}_n)_{n \in \mathbb{N}}$, we first note that when Alice and Bob are honest, all the checks in $\mathcal{P}_n$ pass with high probability (as proved above). Thus, these checks cease to matter when Alice and Bob are honest. In such a setting, we show that $\mathcal{P}_n$ then is fundamentally no different from Protocol~\ref{protocol:C2P} which satisfies (\ref{eqn:ach_2p_wtap_0})-(\ref{eqn:ach_2p_wtap_3}) for honest Bob and Alice in this setup.

\underline{$\mathbf{\epsilon_1 \leq 1/2}$}: 

In this case of the protocol, Bob sends disjoint tuples $\boldsymbol{L}_0,\boldsymbol{L}_1$ to Alice over the public channel. One of these tuples comprises of unerased positions and the other comprises of mostly erased positions from $\boldsymbol{Y}$. Alice uses the bits $\boldsymbol{X}|_{\boldsymbol{L}_0}, \boldsymbol{X}|_{\boldsymbol{L}_1}$ to form keys (using functions $F_0, F_1$) that she uses to encrypt her strings $\boldsymbol{K}_0,\boldsymbol{K}_1$. We show that both keys are secret from Eve (even if Eve additionally know $U$) and there is one key not known to colluding Bob and Eve. The steps involved in proving these statements are not very different from those proving a similar property for Protocol~\ref{protocol:C2P}, except that $V_E, V_{BE}$ here have some additional variables in them including the ones used during interactive hashing.

\begin{itemize}

\item We show in Lemma~\ref{lem:eve_knows_little_lt} (in Appendix~\ref{appndx:supporting_lemmas})  that 
\begin{align*}
H(F_0(\boldsymbol{X}|_{\boldsymbol{L}_0}) \;  \mathlarger{\mid} \; F_0, U, V_E) & \geq  (1 - \xi) \cdot \left(     (\epsilon_1 \epsilon_2 - 5 \delta - 2\tilde{\delta} - \delta')n - \frac{2^{-(\delta + \delta') n}}{\ln 2}     \right) \\
H(F_1(\boldsymbol{X}|_{\boldsymbol{L}_1}) \;  \mathlarger{\mid} \; F_1, U, V_E) & \geq  (1 - \xi) \cdot \left(     (\epsilon_1 \epsilon_2 - 5 \delta - 2\tilde{\delta} - \delta')n - \frac{2^{-(\delta + \delta') n}}{\ln 2}     \right) 
\end{align*}
where $\xi \longrightarrow 0$ exponentially fast as $n \longrightarrow \infty$.

Thus, even if Eve knows $U$, Eve gains only about $2^{-(\delta + \delta') n}/\ln 2$ bits of information about either of the keys that encrypt $\boldsymbol{K}_0$ and $\boldsymbol{K}_1$. 

 \item We show in Lemma~\ref{lem:bob_honest_bobeve_know_little_lt} (in Appendix~\ref{appndx:supporting_lemmas}) that 
\[  H( F_{\overline{U}}(\boldsymbol{X}|_{\boldsymbol{L}_{\overline{U}}}) \;  \mathlarger{\mid} \; F_{\overline{U}}, V_{BE} ) \geq (1 - \xi) \cdot \left(   (\epsilon_1 \epsilon_2 - 5 \delta - 2\tilde{\delta} - \delta')n - \frac{2^{-(\delta + \delta') n}}{\ln 2}  \right)  \]
where $\xi \longrightarrow 0$ exponentially fast as $n \longrightarrow \infty$. 

Thus, colluding Bob and Eve learn only about $2^{-(\delta + \delta') n}/\ln 2$ bits of information about the key $F_{\overline{U}}(\boldsymbol{X}|_{\boldsymbol{L}_{\overline{U}}})$. 

\end{itemize}

Hence, both keys are secret from Eve and one of the keys is secret from colluding Bob and Eve. Such keys encrypting Alice's strings are sufficient to satisfy (\ref{eqn:ach_2p_wtap_0})-(\ref{eqn:ach_2p_wtap_3}) in this setup, as seen previously in Protocol~\ref{protocol:C2P}. As a result, (\ref{eqn:ach_2p_wtap_0})-(\ref{eqn:ach_2p_wtap_3}) are satisfied for $(\mathcal{P}_n)_{n \in \mathbb{N}}$ as well.

\underline{$\mathbf{\epsilon_1 > 1/2}$}: 

This case of the protocol is the same as that described for $\epsilon_1 \leq 1/2$, except that the disjoint tuples being used are now $(L_0 \backslash L_0 \cap L_1)$ and $(L_1 \backslash L_0 \cap L_1)$. We again show that both keys are secret from Eve (even if Eve additionally knows $U$) and there is one key not known to colluding Bob and Eve.

\begin{itemize}

\item By Lemma~\ref{lem:eve_knows_little_gt} (in Appendix~\ref{appndx:supporting_lemmas}), we have 
\begin{align*}
H(F_U( \boldsymbol{X}|_{L_{\Phi} \backslash L_0 \cap L_1}  ) \;  \mathlarger{\mid} \; F_U, U, V_E) & \geq (1 - \xi) \cdot \left(    \beta n(\epsilon_1 \epsilon_2 - 3 \delta - \delta') - \frac{2^{-(\delta + \delta')\beta n}}{\ln 2}    \right) \\
H(F_{\overline{U}}( \boldsymbol{X}|_{L_{\overline{\Phi}} \backslash L_0 \cap L_1}  ) \;  \mathlarger{\mid} \; F_{\overline{U}}, U, V_E)  & \geq (1 - \xi) \cdot \left(    \beta n(\epsilon_1 \epsilon_2 - 3 \delta - \delta') - \frac{2^{-(\delta + \delta')\beta n}}{\ln 2}    \right)
\end{align*}
where $\xi \longrightarrow 0$ exponentially fast as $n \longrightarrow \infty$.

Thus, even if Eve knows $U$, Eve gains only about $2^{-(\delta + \delta')\beta n}/\ln 2$ bit of information about either of the keys that encrypt $\boldsymbol{K}_0$ and $\boldsymbol{K}_1$.

\item  We show in Lemma~\ref{lem:bob_honest_bobeve_know_little_gt} (in Appendix~\ref{appndx:supporting_lemmas}) that
\[  H( F_{\overline{U}}(\boldsymbol{X}|_{L_{\overline{\Phi}} \backslash L_0 \cap L_1}) \;  \mathlarger{\mid} \; F_{\overline{U}}, V_{BE} ) \geq (1 - \xi) \cdot \left(   \beta n(\epsilon_1 \epsilon_2 - 3 \delta - \delta') - \frac{2^{-(\delta + \delta')\beta n}}{\ln 2}  \right)   \]
where $\xi \longrightarrow 0$ exponentially fast as $n \longrightarrow \infty$.

Thus, colluding Bob and Eve learn only about $2^{-(\delta + \delta')\beta n}/\ln 2$ bits of information about the key $F_{\overline{U}}(\boldsymbol{X}|_{L_{\overline{\Phi}} \backslash L_0 \cap L_1})$.

\end{itemize}

This scheme, like Protocol~\ref{protocol:C2P}, produces two keys both of which are not known to Eve and one of which is not known to colluding Bob and Eve. Thus, for the same reasons as for Protocol~\ref{protocol:C2P}, (\ref{eqn:ach_2p_wtap_0})-(\ref{eqn:ach_2p_wtap_3}) are satisfied for $(\mathcal{P}_n)_{n \in \mathbb{N}}$.


\subsection{Proof of Lemma~\ref{lem:ach_rate_malicious_wtap}(\ref{lem:malicious_alice_lemma})}

A malicious Alice, colluding with Eve, can present arbitrary values for $\boldsymbol{X}$ and can adopt an arbitrary strategy during interactive hashing, in $\mathcal{P}_n$. Let $V_{AE}$ be the combined views of Alice and Eve at the start of interactive hashing and let $V^{IH}_{AE}$ be the combined views of Alice and Eve at the end of interactive hashing. Then, Property~\ref{prop:ih_3} of interactive hashing guarantees that if, $\forall \boldsymbol{s} \in \{0,1\}^m, P \left[ \boldsymbol{S} = \boldsymbol{s} | V_{AE} \right] = 1/2^m$, then $\forall \boldsymbol{s}_0, \boldsymbol{s}_1 \in \{0,1\}^m$, $P \left[  \boldsymbol{S} = \boldsymbol{s}_0 | V^{IH}_{AE}, \boldsymbol{S}_0 = \boldsymbol{s}_0, \boldsymbol{S}_1 = \boldsymbol{s}_1 \right] = P \left[  \boldsymbol{S} = \boldsymbol{s}_1 | V^{IH}_{AE}, \boldsymbol{S}_0 = \boldsymbol{s}_0, \boldsymbol{S}_1 = \boldsymbol{s}_1 \right] = 1/2$. In other words, 
\[P \left[  \Phi = 0 | V^{IH}_{AE}, \boldsymbol{S}_0 = \boldsymbol{s}_0, \boldsymbol{S}_1 = \boldsymbol{s}_1 \right] = P \left[  \Phi = 1 | V^{IH}_{AE}, \boldsymbol{S}_0 = \boldsymbol{s}_0, \boldsymbol{S}_1 = \boldsymbol{s}_1 \right] = 1/2.\] 
This is the main property that we use to guarantee privacy for honest Bob against malicious Alice who is potentially colluding with Eve.  Specifically, we show that malicious Alice cannot influence or guess honest Bob's choices and, as a result, does not learn $U$.


\subsubsection{\underline{$\epsilon_1 \leq 1/2$}} \*

In this regime, $(\boldsymbol{X}, \boldsymbol{Z}, \boldsymbol{K}_0, \boldsymbol{K}_1) - (E,\overline{E}, \boldsymbol{S}, \boldsymbol{J}) - (\boldsymbol{L}_0,\boldsymbol{L}_1)$ is a Markov chain. Furthermore,

\begin{itemize}
\item $(E,\overline{E}, \boldsymbol{S}, \boldsymbol{J}) \independent (\boldsymbol{X}, \boldsymbol{Z}, \boldsymbol{K}_0, \boldsymbol{K}_1)$
\item $(\boldsymbol{S},\boldsymbol{J}) \independent (E,\overline{E})$, since honest Bob chose the string $\boldsymbol{S}$ independently of $E,\overline{E}$.
\item $(\boldsymbol{L}_0,\boldsymbol{L}_1) \independent (\boldsymbol{S},\boldsymbol{J})$, since $\boldsymbol{L}_0,\boldsymbol{L}_1$ are randomly ordered tuples, conveying no information about $\boldsymbol{S},\boldsymbol{J}$.
\end{itemize}

Thus, when malicious Alice receives $\boldsymbol{L}_0,\boldsymbol{L}_1$, it gains no information about $\boldsymbol{S},\boldsymbol{J}$ and certainly no information about $U$.

Now, $V_{AE} = (\boldsymbol{X},\boldsymbol{Z},\boldsymbol{K}_0,\boldsymbol{K}_1,\boldsymbol{L}_0,\boldsymbol{L}_1)$. Hence,
\begin{align*}
P \left[  \boldsymbol{S} = \boldsymbol{s} | V_{AE} \right] & = P \left[  \boldsymbol{S} = \boldsymbol{s} | \boldsymbol{X},\boldsymbol{Z},\boldsymbol{K}_0,\boldsymbol{K}_1,\boldsymbol{L}_0,\boldsymbol{L}_1 \right] \\
                                                                                        & = P \left[  \boldsymbol{S} = \boldsymbol{s} | \boldsymbol{X},\boldsymbol{Z},\boldsymbol{K}_0,\boldsymbol{K}_1 \right] \\
                                                                                        & = P \left[  \boldsymbol{S} = \boldsymbol{s} \right] \\
                                                                                        & = \frac{1}{2^m}
\end{align*}

As a result, $P \left[  \Phi = 0 | V^{IH}_{AE} \right] = P \left[  \Phi = 1 | V^{IH}_{AE} \right] = 1/2$. Thus, when Bob communicates $\Theta = \Phi \oplus U$, malicious Alice does not learn $U$.


\subsubsection{\underline{$\epsilon_1 > 1/2$}} \*

In this regime, $(\boldsymbol{X},\boldsymbol{Z},\boldsymbol{K}_0,\boldsymbol{K}_1) - (E,\overline{E}) - \boldsymbol{S}$ is a Markov chain. Since $(E,\overline{E}) \independent (\boldsymbol{X}, \boldsymbol{Z}, \boldsymbol{K}_0, \boldsymbol{K}_1)$, we have $\boldsymbol{S} \independent (\boldsymbol{X},\boldsymbol{Z},\boldsymbol{K}_0,\boldsymbol{K}_1)$. Importantly, $V_{AE} = (\boldsymbol{X},\boldsymbol{Z},\boldsymbol{K}_0,\boldsymbol{K}_1)$ and, thus, $\boldsymbol{S} \independent V_{AE}$. As a result,
\begin{align*}  
  P[\boldsymbol{S}  = \boldsymbol{s} | V_{AE}]  = P[\boldsymbol{S} = \boldsymbol{s}]  & = \left\{   \begin{array}{ll}  \underset{b_g}{\sum} P[\boldsymbol{S} = \boldsymbol{s} | \;\; |\mathcal{B}_G| = b_g] \cdot P[|\mathcal{B}_G| = b_g], & \boldsymbol{s} \in \mathcal{B} \\ \frac{1}{|\mathcal{B}^c|} \cdot \left(1 - \frac{|\mathcal{B}|}{2^m} \right), & \boldsymbol{s} \in \mathcal{B}^c    \end{array}   \right. \\
   & = \left\{   \begin{array}{ll}  \underset{b_g}{\sum} P[\boldsymbol{S} = \boldsymbol{s} | \boldsymbol{s} \in \mathcal{B}_G, |\mathcal{B}_G| = b_g] & \\ \quad \cdot P[\boldsymbol{s} \in \mathcal{B}_G | \;\; |\mathcal{B}_G| = b_g] \cdot P[|\mathcal{B}_G| = b_g], & \boldsymbol{s} \in \mathcal{B} \\ \frac{1}{2^m} , & \boldsymbol{s} \in \mathcal{B}^c    \end{array}   \right. \\
\end{align*}

Note that :

\begin{itemize}
\item $P[\boldsymbol{S} =\boldsymbol{s} | \boldsymbol{s} \in \mathcal{B}_G, |\mathcal{B}_G| = b_g]  =  \frac{1}{b_g} \cdot \frac{|\mathcal{B}|}{2^m}$
\item $\displaystyle\begin{aligned}[t]
             P[\boldsymbol{s} \in \mathcal{B}_G | \;\; |\mathcal{B}_G| = b_g]  & = 1 -  P[\boldsymbol{s} \notin \mathcal{B}_G | \;\; |\mathcal{B}_G| = b_g] \\
                                                & =  1 - \frac{ \comb{|\mathcal{B}|-1}{b_g}  }{\comb{|\mathcal{B}|}{b_g}} \\
                                                & = 1 - \frac{|\mathcal{B}| - b_g}{|\mathcal{B}|} \\
                                                & = \frac{ b_g }{|\mathcal{B}|} \end{aligned}$
\end{itemize}

Thus, 
\begin{align*}
   P[\boldsymbol{S} = \boldsymbol{s} | V_{AE}]  & =  \left\{  \begin{array}{ll}   \underset{b_g}{\sum}  \frac{1}{b_g} \frac{|\mathcal{B}|}{2^m} \cdot \frac{ b_g }{|\mathcal{B}|} \cdot  P[|\mathcal{B}_G| = b_g], & \boldsymbol{s} \in \mathcal{B} \\  \frac{1}{2^m}, & \boldsymbol{s} \in \mathcal{B}^c   \end{array}    \right. \\
                         &  = \left\{  \begin{array}{ll}   \frac{1}{2^m}, & \boldsymbol{s} \in \mathcal{B} \\  \frac{1}{2^m}, & \boldsymbol{s} \in \mathcal{B}^c   \end{array}    \right. 
\end{align*}

That is, $\forall \boldsymbol{s} \in \{0,1\}^m, P \left[  \boldsymbol{S} = \boldsymbol{s} | V_{AE} \right] = P[\boldsymbol{S} = \boldsymbol{s}] = 1/2^m$. As a result, $P \left[  \Phi = 0 | V^{IH}_{AE} \right] = P \left[  \Phi = 1 | V^{IH}_{AE} \right] = 1/2$. Hence, when Alice receives $\Theta = \Phi \oplus U$, Alice learns nothing about $U$.


\subsection{Proof of Lemma~\ref{lem:ach_rate_malicious_wtap}(\ref{lem:malicious_bob_lemma})}



\subsubsection{\underline{$\epsilon_1 \leq 1/2$}} \*

A malicious Bob, in collusion with Eve, can produce arbitrary values for $(\boldsymbol{L}_0,\boldsymbol{L}_1,\Theta)$ during $\mathcal{P}_n$. In order to pass the check in step~\ref{step:abort0_1}, $\{\boldsymbol{L}_0\}, \{\boldsymbol{L}_1\}$ have to be disjoint. Importantly, Bob has to reveal $\boldsymbol{L}_0,\boldsymbol{L}_1$ before it initiates interactive hashing.

We consider the following two exhaustive cases on $\#_e(\boldsymbol{\Psi}|_{\boldsymbol{L}_0}), \#_e(\boldsymbol{\Psi}|_{\boldsymbol{L}_1})$ (see (\ref{eqn:defn_of_psi}) for the definition of $\boldsymbol{\Psi}$). 


\underline{Case 1}: \underline{$(\#_e(\boldsymbol{\Psi}|_{\boldsymbol{L}_0}) < \delta n)$ OR  $(\#_e(\boldsymbol{\Psi}|_{\boldsymbol{L}_1}) < \delta n)$} \*

W.l.o.g. let  $\#_e(\boldsymbol{\Psi}|_{\boldsymbol{L}_0}) < \delta n$. A lower bound on $\#_e(\boldsymbol{\Psi}|_{\boldsymbol{L}_1})$ is computed as follows :
\begin{align*}
\#_{\overline{e}}(\boldsymbol{\Psi}|_{\boldsymbol{L}_1}) & \leq \#_{\overline{e}}(\boldsymbol{\Psi}) - \#_{\overline{e}}(\boldsymbol{\Psi}|_{\boldsymbol{L}_0}) \\
                              & \leq (1 - \epsilon_1 \epsilon_2 + \delta)n - \#_{\overline{e}}(\boldsymbol{\Psi}|_{\boldsymbol{L}_0}) \text{  [due to Chernoff's bound, w.h.p.]}\\
                                    & \leq (1 - \epsilon_1 \epsilon_2 + \delta)n - (\beta n - \delta n)
\end{align*}

Therefore,
\begin{align*}
\#_e(\boldsymbol{\Psi}|_{\boldsymbol{L}_1}) & = \beta n - \#_{\overline{e}}(\boldsymbol{\Psi}|_{\boldsymbol{L}_1}) \\
                 & \geq \beta n - (1 - \epsilon_1 \epsilon_2 + \delta)n + (\beta n - \delta n) \\
                                               & = 2 \beta n - (1 - \epsilon_1 \epsilon_2 + 2 \delta)n \\
                                               & = 2 \left( \frac{1}{2} - \delta  - \tilde{\delta}\right) n -  (1 - \epsilon_1 \epsilon_2 + 2 \delta)n \\
                                               & = (\epsilon_1 \epsilon_2 - 4 \delta - 2\tilde{\delta})n
\end{align*}

We show in Lemma~\ref{lem:bob_malicious_bobeve_know_little_lt} (in Appendix~\ref{appndx:supporting_lemmas}) that whenever $\#_e(\boldsymbol{\Psi}|_{\boldsymbol{L}_1}) \geq (\epsilon_1 \epsilon_2 - 4 \delta - 2\tilde{\delta})n$, 
 \[ R(\boldsymbol{X}|_{\boldsymbol{L}_1} \;  \mathlarger{\mid} \; V_{BE} = v_{BE}) \geq (\epsilon_1 \epsilon_2 - 4 \delta - 2\tilde{\delta})n. \]

As a consequence of Lemma~\ref{lem:privacy_amplification}, we get:
\begin{align*}  
 H(F_1(\boldsymbol{X}|_{\boldsymbol{L}_1}) \;  \mathlarger{\mid} \; F_1, V_{BE} = v_{BE}) & \geq  (\epsilon_1 \epsilon_2 - 5 \delta - 2\tilde{\delta} - \delta')n - \frac{2^{ (\epsilon_1 \epsilon_2 - 5 \delta - 2\tilde{\delta} - \delta')n - (\epsilon_1 \epsilon_2 - 4 \delta - 2\tilde{\delta})n }}{\ln 2} \\
   & = (\epsilon_1 \epsilon_2 - 5 \delta - 2\tilde{\delta} - \delta')n - \frac{2^{-(\delta + \delta') n }}{\ln 2}. 
\end{align*}

That is, colluding Bob and Eve learn no more than $2^{-(\delta + \delta') n}/\ln 2$ bits of information about $F_1(\boldsymbol{X}|_{\boldsymbol{L}_1})$. Hence, malicious Bob colluding with Eve learns only a vanishingly small amount of information about $\boldsymbol{K}_1$ from $\boldsymbol{K}_1 \oplus F_1(\boldsymbol{X}|_{\boldsymbol{L}_1})$.


\underline{Case 2} : \underline{$(\#_e(\boldsymbol{\Psi}|_{\boldsymbol{L}_0}) \geq \delta n)$ AND  $(\#_e(\boldsymbol{\Psi}|_{\boldsymbol{L}_1}) \geq \delta n)$} \*

The key idea in this part of the proof is the following: Bob cannot control the tuple $\boldsymbol{J}_{\overline{\Phi}}$ produced by interactive hashing. However, Bob has to correctly reveal to Alice either $\boldsymbol{X}_{\boldsymbol{L}_0|_{\boldsymbol{J}_{\overline{\Phi}}}}$ or $\boldsymbol{X}_{\boldsymbol{L}_1|_{J_{\overline{\Phi}}}}$ (depending on $\Theta$). We show that both $\boldsymbol{\Psi}_{\boldsymbol{L}_0|_{\boldsymbol{J}_{\overline{\Phi}}}}$ and $\boldsymbol{\Psi}_{\boldsymbol{L}_1|_{\boldsymbol{J}_{\overline{\Phi}}}}$ have a substantial number of erasures and Bob can reveal these erased bits correctly with only exponentially small probability.

Define 
\[  \mathcal{T}_{e} := \left\{   \boldsymbol{a} \in \mathcal{T} : \#_e(\boldsymbol{\Psi}|_{\boldsymbol{L}_0|_{\boldsymbol{a}}}) < \gamma \delta n  \text{ OR }  \#_e(\boldsymbol{\Psi}|_{\boldsymbol{L}_1|_{\boldsymbol{a}}}) <  \gamma \delta n \right\}  \]
 
Let $\boldsymbol{\psi}|_{\boldsymbol{l}_0}, \boldsymbol{\psi}|_{\boldsymbol{l}_1}$ be specific realizations of $\boldsymbol{\Psi}|_{\boldsymbol{L}_0}, \boldsymbol{\Psi}|_{\boldsymbol{L}_1}$ respectively. Applying Lemma~\ref{lem:rare_property0_1} (with $\boldsymbol{\mathfrak{p}} = \boldsymbol{\psi}|_{\boldsymbol{l}_0}, \boldsymbol{\mathfrak{q}} = \boldsymbol{\psi}|_{\boldsymbol{l}_1}, k = \beta n, \varphi = \delta/\beta, \alpha = \gamma/\beta, \rho = \delta$), we get:
\[  \frac{| \mathcal{T}_e |}{|\mathcal{T}|} \leq 2 e^{-2 \gamma n \delta^2}  \]

Now, let $\mathcal{B}_e = Q^{-1}(\mathcal{T}_e)$. Then we have:
\begin{align*}
\frac{ |\mathcal{B}_e| }{2^m} & \leq \frac{ |\mathcal{B}_e| }{|\mathcal{T}|} \\
                                             & \leq \frac{ 2 |\mathcal{T}_e| }{|\mathcal{T}|} \\
                                             & \leq 4 e^{-2 \gamma n \delta^2}
\end{align*}

Property~\ref{prop:ih_4} of interactive hashing now gives:
\begin{align*}  
P\left[ \boldsymbol{S}_0, \boldsymbol{S}_1 \in \mathcal{B}_e  \right] & \leq 16 \times \frac{ |\mathcal{B}_e| }{2^m} \\
                                                                 & \leq 64 \times  e^{-2 \gamma n \delta^2}
\end{align*}

Thus, with high probability either  $\boldsymbol{S}_0 \notin \mathcal{B}_e$ or  $\boldsymbol{S}_1 \notin \mathcal{B}_e$. Let us assume that $\boldsymbol{S}_0 \notin \mathcal{B}_e$. Recall that $\boldsymbol{J}_0 = Q(\boldsymbol{S}_0)$. Therefore, $\boldsymbol{J}_0 \notin \mathcal{T}_e$. This means $\#_e(\boldsymbol{\Psi}|_{\boldsymbol{L}_0|_{\boldsymbol{J}_0}}) \geq \gamma \delta n$ and $\#_e(\boldsymbol{\Psi}|_{\boldsymbol{L}_1|_{\boldsymbol{J}_0}}) \geq \gamma \delta n$. Since Bob has to reveal one of these bit strings correctly to Alice, Bob has to guess at least $\gamma \delta n$ unknown i.i.d. bits correctly. Bob can make the correct guess with probability $2^{-\gamma \delta n}$. A similar argument holds if we assume $\boldsymbol{S}_1 \notin \mathcal{B}_e$. Hence, Bob fails the test in step~\ref{step:abort0_2} with very high probability.


\subsubsection{\underline{$\epsilon_1 > 1/2$}} \*

In this case, malicious Bob in collusion with Eve, can present an arbitrary value for $\Theta$ during $\mathcal{P}_n$. Bob initiates interactive hashing and gets the output strings $\boldsymbol{S}_0,\boldsymbol{S}_1$. Recall that $L_0 = Q(\boldsymbol{S}_0)$ and $L_1 = Q(\boldsymbol{S}_1)$. We show that it is only with negligibly small probability that both $\#_{\overline{e}}(\boldsymbol{\Psi}|_{L_0})$ and  $\#_{\overline{e}}(\boldsymbol{\Psi}|_{L_1})$ exceed a certain threshold. That is, with high probability, at least one of $\#_{\overline{e}}(\boldsymbol{\Psi}|_{L_0})$ or $\#_{\overline{e}}(\boldsymbol{\Psi}|_{L_1})$ is below that threshold.

We condition the following arguments on no abort happening in step~\ref{step:abort1_1}, which means that $\boldsymbol{S}_0,\boldsymbol{S}_1 \in \mathcal{B}$.

Define 
\[  \mathcal{T}_{\overline{e}} := \left\{   A \in \mathcal{T} : \#_{\overline{e}}(\boldsymbol{\Psi}|_A) \geq \beta n (1 - \epsilon_1 \epsilon_2 + 2\delta)  \right\}  \]

Note that w.h.p. (due to Chernoff's bound), $\#_{\overline{e}}(\boldsymbol{\Psi}) \leq (1 - \epsilon_1 \epsilon_2 + \delta)n$. Let $\boldsymbol{\psi}$ be a typical realization of $\boldsymbol{\Psi}$. Applying Lemma~\ref{lem:rare_property0_0} (with $\boldsymbol{\mathfrak{p}} = \boldsymbol{\psi}, k = n, 1 - \varphi = (1 - \epsilon_1 \epsilon_2 + \delta), \alpha = \beta, \rho = \delta$), we have w.h.p.:
\[  \frac{|\mathcal{T}_{\overline{e}}|}{|\mathcal{T}|} \leq e^{-2 \beta n \delta^2}   \]

Now, let $\mathcal{B}_{\overline{e}} = Q^{-1}(\mathcal{T}_{\overline{e}})$. Then we have:
\begin{align*}
\frac{ |\mathcal{B}_{\overline{e}}| }{|\mathcal{B}|}  & = \frac{ |\mathcal{T}_{\overline{e}}| }{|\mathcal{B}|} \\
                                                                             & = \frac{ |\mathcal{T}_{\overline{e}}| }{|\mathcal{T}|} \\
                                                                             & \leq e^{-2 \beta n \delta^2}
\end{align*}

Property~\ref{prop:ih_4} of interactive hashing now gives:
\begin{align*}  
P\left[  \boldsymbol{S}_0, \boldsymbol{S}_1 \in \mathcal{B}_{\overline{e}}  \right] & \leq 16 \times \frac{ |\mathcal{B}_{\overline{e}}| }{\mathcal{B}} \\
                                                                 & \leq 16 \times  e^{-2 \beta n \delta^2}
\end{align*}

This implies that w.h.p. either $\boldsymbol{S}_0 \notin \mathcal{B}_{\overline{e}}$ or $\boldsymbol{S}_1 \notin \mathcal{B}_{\overline{e}}$. That is, w.h.p. either $L_0 \notin \mathcal{T}_{\overline{e}}$ or $L_1 \notin \mathcal{T}_{\overline{e}}$. As a result, w.h.p. either $\#_{\overline{e}}(\boldsymbol{\Psi}|_{L_0}) < \beta n (1 - \epsilon_1 \epsilon_2 + 2\delta)$ or $\#_{\overline{e}}(\boldsymbol{\Psi}|_{L_1}) < \beta n (1 - \epsilon_1 \epsilon_2 + 2\delta)$.

Since the protocol uses $\boldsymbol{X}|_{L_0 \backslash L_0 \cap L_1}$ and $\boldsymbol{X}|_{L_1 \backslash L_0 \cap L_1}$ for creating the keys (using the hash functions $F_0,F_1$), the keys obtained are independent. To ensure that at least one of these keys remains unknown to malicious Bob, the protocol has the following two steps:

\begin{enumerate}

\item In step~\ref{step:reveal_common_part1}, Bob has to correctly reveal $\boldsymbol{X}|_{L_0 \cap L_1}$. Asking Bob to reveal $\boldsymbol{X}|_{L_0 \cap L_1}$ prevents malicious Bob from manipulating interactive hashing to have $\boldsymbol{\Psi}|_{L_0 \cap L_1}$ comprise of erasures, thereby packing $\boldsymbol{\Psi}|_{L_0 \backslash L_0 \cap L_1}$ with non-erasures of of $\boldsymbol{\Psi}|_{L_0}$ and packing $\boldsymbol{\Psi}|_{L_1 \backslash L_0 \cap L_1}$ with non-erasures of $\boldsymbol{\Psi}|_{L_1}$. Such a packing would allow malicious Bob to learn non-negligible information about the keys being created using $\boldsymbol{X}|_{L_0 \backslash L_0 \cap L_1}$ and $\boldsymbol{X}|_{L_1 \backslash L_0 \cap L_1}$. As a result,
\begin{align*}
\#_{\overline{e}}(\boldsymbol{\Psi}|_{L_0 \backslash L_0 \cap L_1}) & = \#_{\overline{e}}(\boldsymbol{\Psi}|_{L_0}) - |L_0 \cap L_1| \\
\#_{\overline{e}}(\boldsymbol{\Psi}|_{L_1 \backslash L_0 \cap L_1}) & = \#_{\overline{e}}(\boldsymbol{\Psi}|_{L_1}) - |L_0 \cap L_1|
\end{align*}

\item We know that after the previous step,
\begin{align*}
\#_e(\boldsymbol{\Psi}|_{L_0 \backslash L_0 \cap L_1}) & = |L_0 \backslash L_0 \cap L_1| - \#_{\overline{e}}(\boldsymbol{\Psi}|_{L_0 \backslash L_0 \cap L_1}) \\
                                                                 & = \beta n - |L_0 \cap L_1| - \#_{\overline{e}}(\boldsymbol{\Psi}|_{L_0 \backslash L_0 \cap L_1}) \\
                                                                 & = \beta n - \#_{\overline{e}}(\boldsymbol{\Psi}|_{L_0})
\end{align*}

and 
\begin{align*}
\#_e(\boldsymbol{\Psi}|_{L_1 \backslash L_0 \cap L_1}) & = |L_1 \backslash L_0 \cap L_1| - \#_{\overline{e}}(\boldsymbol{\Psi}|_{L_1 \backslash L_0 \cap L_1}) \\
                                                                 & = \beta n - |L_0 \cap L_1| - \#_{\overline{e}}(\boldsymbol{\Psi}|_{L_1 \backslash L_0 \cap L_1}) \\
                                                                 & = \beta n - \#_{\overline{e}}(\boldsymbol{\Psi}|_{L_1})
\end{align*}

Thus, w.h.p. either $\#_e(\boldsymbol{\Psi}|_{L_0 \backslash L_0 \cap L_1}) \geq \beta n (\epsilon_1 \epsilon_2 - 2 \delta)$ or $\#_e(\boldsymbol{\Psi}|_{L_1 \backslash L_0 \cap L_1}) \geq \beta n (\epsilon_1 \epsilon_2 - 2 \delta)$. W.l.o.g. suppose $\#_e(\boldsymbol{\Psi}|_{L_0 \backslash L_0 \cap L_1}) \geq \beta n (\epsilon_1 \epsilon_2 - 2 \delta)$. We show in Lemma~\ref{lem:bob_malicious_bobeve_know_little_gt} (in Appendix~\ref{appndx:supporting_lemmas}) that whenever $\#_e(\boldsymbol{\Psi}|_{L_0 \backslash L_0 \cap L_1}) \geq \beta n (\epsilon_1 \epsilon_2 - 2 \delta)$, we have
\[ R(\boldsymbol{X}|_{L_0 \backslash L_0 \cap L_1} | V_{BE} = v_{BE}) \geq \beta n (\epsilon_1 \epsilon_2 - 2 \delta). \]

Furthermore, suppose that $\boldsymbol{X}|_{L_0 \backslash L_0 \cap L_1}$ is the input to $F_1$. Then, applying Lemma~\ref{lem:privacy_amplification} gives us:
\begin{align*}
H(F_1(\boldsymbol{X}|_{L_0 \backslash L_0 \cap L_1}) | F_1, V_{BE} = v_{BE}) & \geq \beta n (\epsilon_1 \epsilon_2 - 3 \delta - \delta') - \frac{2^{ \beta n (\epsilon_1 \epsilon_2 - 3 \delta - \delta') - \beta n (\epsilon_1 \epsilon_2 - 2 \delta) }}{\ln 2} \\
  & = \beta n (\epsilon_1 \epsilon_2 - 3 \delta - \delta') - \frac{2^{ -(\delta + \delta')\beta n }}{\ln 2}
\end{align*}

\end{enumerate}

This shows that malicious Bob, in collusion with Eve, cannot learn more than $2^{ -(\delta + \delta')\beta n }/\ln 2$ bits of information for at least one of the keys created using $F_0,F_1$. That is, for at least one of strings $\boldsymbol{K}_0,\boldsymbol{K}_1$, colluding Bob and Eve gain only a vanishingly small amount of information, regardless of what value of $\Theta$ Bob chooses to reveal during $\mathcal{P}_n$.


\subsection{Supporting Lemmas}
\label{appndx:supporting_lemmas}


\begin{lemma}[Hoeffding's Inequality, \protect{\cite[Proposition 1.2]{hoeffding_source}}]
\label{lem:hoeffding}
Let $\mathcal{A} = ( a_1,a_2,\ldots,a_N )$ be a finite population of $N$ points and $A_1,A_2,\ldots,A_k$ be a random sample drawn without  replacement from $\mathcal{A}$. Let 
\[  a = \underset{1 \leq i \leq N}{\min} a_i \quad \text{and} \quad  \tilde{a} = \underset{1 \leq i \leq N}{\max} a_i \]
Then, for all $\rho > 0$,
\[   P \left[   \frac{1}{k} \overset{k}{\underset{i=1}{\sum}} A_i  - \mu \geq \rho   \right] \leq  e^{  - \frac{2k\rho^2}{(\tilde{a}-a)^2}  }    \]
where $\mu = (1/N) \cdot \overset{N}{\underset{i=1}{\sum}} a_i$ is the mean of $\mathcal{A}$.
\end{lemma}

\begin{lemma}
\label{lem:dense}
Let $\vartheta \in \mathbb{Q} \cap [0, \infty)$, $\alpha \in \mathbb{R} \backslash \mathbb{Q} \cap [0, \infty)$ and $\varphi \in \mathbb{R} \cap [0, \infty)$ be constants. Then, $(<n \alpha - \varphi \log n>)_{\{n \in \mathbb{N} : \vartheta n \in \mathbb{N}\}}$ is dense in $[0,1]$.
\end{lemma}

\begin{proof}
We prove the claim in three parts:
\begin{enumerate}
\item We show that $(<n \alpha>)_{n \in \mathbb{N}}$ is dense in $[0,1]$.
\item Using the fact that $\log n$ increases very slowly at large values of $n$, we show that $(<n \alpha - \varphi \log n>)_{n \in \mathbb{N}}$ is dense in $[0,1]$.
\item Finally, we show that the statements above hold even with the restriction that $\vartheta n$ should be an integer, i.e. we show that $(<n \alpha - \varphi \log n>)_{\{n \in \mathbb{N} : \vartheta n \in \mathbb{N}\}}$ is dense in $[0,1]$.
\end{enumerate}

\underline{Proof of 1):} We first note that the sequence $(<n \alpha>)_{n \in \mathbb{N}}$ comprises of distinct numbers. If not and suppose $<n_1 \alpha> \; = \; <n_2 \alpha>$, $\; n_2 > n_1$, then:

\begin{align*}
n_2 \alpha  - n_1 \alpha & = (\lfloor n_2 \alpha \rfloor + <n_2 \alpha>)  - (\lfloor n_1 \alpha \rfloor + <n_1 \alpha>) \\
                                       & = \lfloor n_2 \alpha \rfloor - \lfloor n_1 \alpha \rfloor
\end{align*}

which implies
\[ \alpha = \frac{\lfloor n_2 \alpha \rfloor - \lfloor n_1 \alpha \rfloor}{n_2 - n_1} \]
 which is clearly a contradiction since $\alpha$ is irrational.

Then, we note that for any $\gamma > 0$, $\exists k_1,k_2 \in \mathbb{N}$, $k_2 > k_1$, such that 
\[ | <k_2 \alpha> - <k_1 \alpha> | \leq \gamma \]

If not, then the sequence $(<n \alpha>)_{n \in \mathbb{N}}$ cannot be an infinite sequence of distinct elements, again a contradiction. 

Let $K = k_2 - k_1$ and $ | <k_2 \alpha> - <k_1 \alpha> | = \tilde{\gamma} \leq \gamma $. Consider the sequence $( <jK \alpha> )_{j=1}^{\lfloor 1/\tilde{\gamma} \rfloor}$. This is the sequence $(j \tilde{\gamma})_{j=1}^{\lfloor 1/\tilde{\gamma} \rfloor}$ if $<k_2 \alpha> \quad > \quad <k_1 \alpha>$ or the sequence $(1 - j \tilde{\gamma})_{j=1}^{\lfloor 1/\tilde{\gamma} \rfloor}$ if $<k_2 \alpha> \quad < \quad <k_1 \alpha>$. 

Thus, for any $x \in [0,1]$, $\exists j_0 \in \{1,2,\ldots, \lfloor 1/\tilde{\gamma} \rfloor \}$ such that $x \in \mathcal{N}_{\tilde{\gamma}}(<j_0 K \alpha>)$. Clearly then, $x \in \mathcal{N}_{\gamma}(<j_0 K \alpha>)$. Since $\gamma$ was chosen arbitrarily, the sequence $(<jK \alpha > )_{j=1}^{\lfloor 1/\tilde{\gamma} \rfloor}$ is dense in $[0,1]$. This, in turn, implies that $(<n \alpha>)_{n \in \mathbb{N}}$ is dense in $[0,1]$.

\underline{Proof of 2):} Let $M \in \mathbb{N}$ such that:
\[ \varphi \log (M + \lfloor \frac{1}{\tilde{\gamma}} \rfloor K ) - \varphi \log M  < \gamma \]
and
\begin{equation}
\label{eqn:positivity}
M \alpha - \varphi \log M > 0.
\end{equation}

For example, consider any $M > \varphi \cdot \lfloor 1/\tilde{\gamma} \rfloor K / \gamma \ln 2$ for which (\ref{eqn:positivity}) holds. Clearly, the sequence $( <(M + jK)\alpha>)_{j=1}^{\lfloor 1/\tilde{\gamma} \rfloor}$ still approximates any number in $[0,1]$ to within a precision of $\gamma$. As a result, the elements of the sequence $( < (M + jK)\alpha - \varphi \log (M + jK) >)_{j=1}^{\lfloor 1/\tilde{\gamma} \rfloor}$ approximate any element of $[0,1]$ to within a precision of $2 \gamma$. Again, since $\gamma$ was arbitrary,  $( < (M + jK)\alpha - \varphi \log (M + jK) >)_{j=1}^{\lfloor 1/\tilde{\gamma} \rfloor}$ is dense in $[0,1]$. As a result, the sequence $( < n \alpha - \varphi \log n >)_{n \in \mathbb{N}}$ is dense in $[0,1]$.

\underline{Proof of 3):} Let $\vartheta = p/q$. We show that the reasoning developed so far holds even when all the integers, sequences and offsets considered are multiplied by $q$. For clarity, we repeat the previous arguments with this change included in them.

 We first note that the sequence $( <qn \alpha>)_{n \in \mathbb{N}}$ comprises of distinct numbers. If not and, say, $<qn_1 \alpha> \; = \; <qn_2 \alpha>$, $\; n_2 > n_1$, then:
\begin{align*}
qn_2 \alpha  - qn_1 \alpha  & = (\lfloor qn_2 \alpha \rfloor + <qn_2 \alpha>)  - (\lfloor qn_1 \alpha \rfloor + <qn_1 \alpha>) \\
  & = \lfloor qn_2 \alpha \rfloor - \lfloor qn_1 \alpha \rfloor
\end{align*}
which implies
\[ \alpha = \frac{\lfloor qn_2 \alpha \rfloor - \lfloor qn_1 \alpha \rfloor}{qn_2 - qn_1} \]
 which is clearly a contradiction since $\alpha$ is irrational.

Then, we note that for any $\gamma > 0$, $\exists k_1,k_2 \in \mathbb{N}$, $k_2 > k_1$, such that 
\[ | <qk_2 \alpha> - <qk_1 \alpha> | \leq \gamma \]

If not, then the sequence $(<qn \alpha>)_{n \in \mathbb{N}}$ cannot be an infinite sequence of distinct elements, again a contradiction.

Let $K = qk_2 - qk_1$ and $ | <qk_2 \alpha> - <qk_1 \alpha> | = \tilde{\gamma} \leq \gamma $. Consider the sequence $(<jK \alpha> )_{j=1}^{\lfloor 1/\tilde{\gamma} \rfloor}$. This is the sequence $(j \tilde{\gamma})_{j=1}^{\lfloor 1/\tilde{\gamma} \rfloor}$ if $<qk_2 \alpha> \quad  > \quad <qk_1 \alpha>$ or the sequence $(1 - j \tilde{\gamma})_{j=1}^{\lfloor 1/\tilde{\gamma} \rfloor}$ if $<qk_2 \alpha> \quad  < \quad <qk_1 \alpha>$.

Thus, for any $x \in [0,1]$, $\exists j_0 \in \{1,2,\ldots, \lfloor 1/\tilde{\gamma} \rfloor \}$ such that $x \in \mathcal{N}_{\tilde{\gamma}}(<j_0 K \alpha>)$. Clearly then, $x \in \mathcal{N}_{\gamma}(<j_0 K \alpha>)$. Since $\gamma$ was chosen arbitrarily, the sequence $(<jK \alpha> )_{j=1}^{\lfloor 1/\tilde{\gamma} \rfloor}$ is dense in $[0,1]$. Importantly, $\vartheta j K = p j (k_2 - k_1) \in \mathbb{N}$ This now implies that $(<n \alpha>)_{ \{ n \in \mathbb{N} : \vartheta n \in \mathbb{N} \} }$ is dense in $[0,1]$.

Now consider the offset $M \in \mathbb{N}$ such that:
\[ \varphi \log (qM + \lfloor \frac{1}{\tilde{\gamma}} \rfloor K ) - \varphi \log(qM)  < \gamma \]
and
\begin{equation}
\label{eqn:more_positivity}
qM \alpha - \varphi \log (qM) > 0.
\end{equation}

For example, consider any $M > (1/q) \cdot (\varphi \lfloor 1/\tilde{\gamma} \rfloor K/\gamma \ln 2)$ for which (\ref{eqn:more_positivity}) holds. Clearly, the sequence $(<(qM + jK)\alpha>)_{j=1}^{\lfloor 1/\tilde{\gamma} \rfloor}$ still approximates any number in $[0,1]$ to within a precision of $\gamma$. As a result, the elements of the sequence $( < (qM + jK)\alpha - \varphi \log (qM + jK) >)_{j=1}^{\lfloor 1/\tilde{\gamma} \rfloor}$ approximate any element of $[0,1]$ to within a precision of $2 \gamma$. Again, since $\gamma$ was arbitrary,  $( < (qM + jK)\alpha - \varphi \log (qM + jK) >)_{j=1}^{\lfloor 1/\tilde{\gamma} \rfloor}$ is dense in $[0,1]$. Importantly again, $\vartheta(qM + jK) = p(M + j(k_2-k_1)) \in \mathbb{N}$. As a result, the sequence $( < n \alpha - \varphi \log n >)_{\{n \in \mathbb{N} : \vartheta n \in \mathbb{N}\}}$ is dense in $[0,1]$.

\end{proof}


\begin{lemma}
\label{lem:log_dense}
For $H(\beta) \in \mathbb{R} \backslash \mathbb{Q} \cap [0, \infty)$, the sequence $(<\log (\comb{n}{\beta n} )>)_{ \{ n \in \mathbb{N} : \beta n \in \mathbb{N} \}}$ is dense in $[0,1]$.
\end{lemma}

\begin{proof}
By Sterling's approximation,
\begin{align*} 
  \log (\comb{n}{\beta n} ) & = nH(\beta) - \frac{1}{2} \log (2 \pi \beta (1 - \beta)n) - \log (1 + O(\frac{1}{n})) \\
  & = nH(\beta) - \frac{1}{2} \log n - \frac{1}{2} \log (2 \pi \beta (1 - \beta)) - \log (1 + O(\frac{1}{n}))
\end{align*}

By Lemma~\ref{lem:dense}, the sequence $(<nH(\beta) - (1/2) \log n>)_{ \{ n \in \mathbb{N} : \beta n \in \mathbb{N}\} }$ is dense in $[0,1]$. Since $(1/2) \cdot \log (2 \pi \beta (1 - \beta))$ is a constant, the sequence $(<nH(\beta) - (1/2) \log n - (1/2) \log (2 \pi \beta (1 - \beta))> )_{ \{ n \in \mathbb{N} : \beta n \in \mathbb{N} \} }$ is also dense in $[0,1]$. The claim now follows since for large $n$, the term $\log (1 + O(1/n))$ is negligibly small.

\end{proof}


\begin{lemma}
\label{lem:rare_property0_0}

Let $\boldsymbol{\mathfrak{p}} \in \{0,1, \bot\}^k$ be fixed. Let $\varphi \in (0,1]$ be such that $\#_e(\boldsymbol{\mathfrak{p}}) \geq \varphi k$. Let $\overline{\varphi} = 1 - \varphi$, $\rho > 0$ and $\alpha \in [0,1]$. Define $\mathcal{T}_{\alpha} := \{ A \subset \{ 1,2,\ldots,k \} : |A| = \alpha k \}$, $\mathcal{T}_e := \{ A \in \mathcal{T}_{\alpha} : \#_e(\boldsymbol{\mathfrak{p}}|_A) \leq \alpha k (\varphi - \rho)  \}$ and $\mathcal{T}_{\overline{e}} := \{ A \in \mathcal{T}_{\alpha} : \#_{\overline{e}}(\boldsymbol{\mathfrak{p}}|_A) \geq \alpha k (\overline{\varphi} + \rho)  \}$. Then,
\[  \frac{|\mathcal{T}_{\overline{e}}|}{|\mathcal{T}_{\alpha}|} =  \frac{|\mathcal{T}_e|}{|\mathcal{T}_{\alpha}|} \leq  e^{-2 \alpha k \rho^2} \]

\end{lemma}

\begin{proof}
Let $\upsilon := \{ j \in \{1,2,\ldots,k\} : \mathfrak{p}_j = \bot \}$ and let $\Upsilon \thicksim \text{Unif}(\mathcal{T}_{\alpha})$. Then, $ \#_e(\boldsymbol{\mathfrak{p}}|_{\Upsilon}) =  \left|  \upsilon \cap \Upsilon \right|$.

Let 
\[ \boldsymbol{\omega} := \left( \omega_j \in \{0,1\}, j = 1,2,\ldots, k : \begin{array}{ll} \omega_j = 0, & \mathfrak{p}_j = \bot \\ \omega_j = 1, & \mathfrak{p}_j \neq \bot \end{array} \right) \]

Let $\mu  = 1 - (1/k) \cdot \overset{k}{\underset{j=0}{\sum}} \omega_j$. Note that $\mu \geq \varphi$.

Let $\Omega_i, i=1,2,\ldots,\alpha k$ be random samples drawn from $\boldsymbol{\omega}$ without replacement. Clearly then,  $\#_e(\boldsymbol{\mathfrak{p}}|_{\Upsilon}) =  \left|  \upsilon \cap \Upsilon \right| \sim  \alpha k - \overset{\alpha k}{\underset{i=1}{\sum}} \Omega_i$.

Using Lemma~\ref{lem:hoeffding}, we get:
\[  P \left[  \frac{1}{\alpha k} \overset{\alpha k}{\underset{i=1}{\sum}} \Omega_i \geq (1 - \mu) + \rho  \right] \leq e^{-2 \alpha k \rho^2}  \]

This implies:
\[  P \left[  \frac{1}{\alpha k}  \left| \upsilon \cap \Upsilon \right| \leq \mu - \rho  \right] \leq e^{-2 \alpha k \rho^2}   \]

Since $\mu \geq \varphi$, we get:
\[  P \left[  \left| \upsilon \cap \Upsilon \right| \leq \alpha k (\varphi - \rho)  \right] \leq e^{-2 \alpha k \rho^2}   \]

Since $\Upsilon$ was a random choice from $\mathcal{T}_{\alpha}$, it follows that :
\[  \frac{|\mathcal{T}_e|}{|\mathcal{T}_{\alpha}|} \leq  e^{-2 \alpha k \rho^2}  \]

Furthermore, for $A \in \mathcal{T}_{\alpha}$, $\#_{\overline{e}}(\boldsymbol{\mathfrak{p}}|_A) = \alpha k - \#_e(\boldsymbol{\mathfrak{p}}|_A)$. As a result, $\mathcal{T}_{\overline{e}} = \mathcal{T}_e$ and the result follows.

\end{proof}


\begin{lemma}
\label{lem:rare_property0_1}

Let $\boldsymbol{\mathfrak{p}}, \boldsymbol{\mathfrak{q}} \in \{0,1, \bot\}^k$ be fixed. Let $\varphi > 0$ be such that $\#_e(\boldsymbol{\mathfrak{p}}) \geq \varphi k$ and $\#_e(\boldsymbol{\mathfrak{q}}) \geq \varphi k$. Let $\rho > 0$ and $\alpha \in [0,1]$. Define $\mathcal{T}_{\alpha} := \{ A \subset \{ 1,2,\ldots,k \} : |A| = \alpha k \}$ and $\mathcal{T}_e := \{ A \in \mathcal{T}_{\alpha} : \#_e(\boldsymbol{\mathfrak{p}}|_A) < \alpha k (\varphi - \rho)  \text{ OR }  \#_e(\boldsymbol{\mathfrak{q}}|_A) < \alpha k (\varphi - \rho) \}$. Then,
\[ \frac{|\mathcal{T}_e|}{|\mathcal{T}_{\alpha}|} \leq  2e^{-2 \alpha k \rho^2} \]

\end{lemma}

\begin{proof}
\begin{align*}
|\mathcal{T}_e| & \leq \left|  \{  A \in \mathcal{T}_{\alpha} : \#_e(\boldsymbol{\mathfrak{p}}|_A) < \alpha k (\varphi - \rho)   \} \right|  \\ & \quad+ \left|  \{  A \in \mathcal{T}_{\alpha} : \#_e(\boldsymbol{\mathfrak{q}}|_A) < \alpha k (\varphi - \rho)   \}  \right| \\
                        & \leq e^{-2 \alpha k \rho^2} \cdot \left| \mathcal{T}_{\alpha} \right| + e^{-2 \alpha k \rho^2} \cdot \left| \mathcal{T}_{\alpha} \right|   \text {  [using Lemma~\ref{lem:rare_property0_0}] }
\end{align*}

As a result,
\[  \frac{|\mathcal{T}_e|}{|\mathcal{T}_{\alpha}|} \leq  2e^{-2 \alpha k \rho^2}  \]

\end{proof}


\begin{lemma}
\label{lem:overlap_within_bounds}
Let $\varphi, \alpha \in (0,1], \rho > 0$. Let $\mathcal{T}_{\varphi} := \{ A \subset \{1,2,\ldots,k\} : |A| = \varphi k \}$ and $\mathcal{T}_{\alpha} := \{ A \subset \{1,2,\ldots,k\} : |A| = \alpha k \}$. Let $\upsilon \in \mathcal{T}_{\varphi}, \Upsilon \thicksim \text{Unif}(\mathcal{T}_{\alpha} )$. 

Then, 
\[ P \left[  \left| \frac{1}{\alpha k} \left| \upsilon \cap \Upsilon \right| - \varphi \right| > \rho \right] \leq 2 e^{-2 \alpha k \rho^2}  \]
\end{lemma}

\begin{proof}
Let 
\[  \boldsymbol{\omega} := \left( \omega_i \in \{0,1\}, i = 1,2,\ldots, k : \begin{array}{ll} \omega_i = 1, & i \in \upsilon \\ \omega_i = 0, & i \notin \upsilon \end{array} \right) \]
Let $\Omega_j, j=1,2,\ldots,\alpha k$ be random samples drawn from $\boldsymbol{\omega}$ without replacement. Clearly then, $\overset{\alpha k}{\underset{j=1}{\sum}} \Omega_j \sim \left|  \upsilon \cap \Upsilon \right|$.

Applying Lemma~\ref{lem:hoeffding}, we have:
\[ P \left[  \frac{1}{\alpha k} \overset{\alpha k}{\underset{j=1}{\sum}} \Omega_j  - \varphi \geq \rho  \right] \leq  e^{-2 \alpha k\rho^2} \]

If, in $\boldsymbol{\omega}$, we now change all $1$'s to $0$'s and vice-versa and proceed as above, we get:
\[  P \left[  \frac{1}{\alpha k} \overset{\alpha k}{\underset{j=1}{\sum}} \Omega_j  - \varphi \leq -\rho  \right] \leq  e^{-2 \alpha k\rho^2} \]

Combining these two inequalities using the union bound, we get:
\[ P \left[  \left| \frac{1}{\alpha k} \left|  \upsilon \cap \Upsilon \right|  - \varphi \right|  \geq \rho  \right] \leq  2e^{-2 \alpha k \rho^2} \]

\end{proof}


\begin{lemma}
\label{lem:overlap_within_bounds_corr}
Let $\varphi \in (0,1], \rho > 0$. Let $\mathcal{T}_{\varphi} := \{ A \subset \{1,2,\ldots,k\} : |A| = \varphi k \}, \upsilon \in \mathcal{T}_{\varphi}, \Upsilon \thicksim \text{Unif}( \mathcal{T}_{\varphi} \backslash \{\upsilon\} )$.

 Then,
\[   P \left[  \left| \frac{1}{\varphi k} \left|\upsilon \cap \Upsilon \right| - \varphi \right| > \rho \right] \leq \frac{2 e^{-2 \varphi k \rho^2}}  {1 -  \frac{1}{| \mathcal{T}_{\varphi} |} }  \]
\end{lemma}

\begin{proof}
Using Lemma~\ref{lem:overlap_within_bounds},
\[ \left| \left\{ A \in \mathcal{T}_{\varphi} : \left| \frac{1}{\varphi k} \left|\upsilon \cap A \right| - \varphi \right| > \rho  \right\} \right|  \leq  |\mathcal{T}_{\varphi} | \cdot 2e^{-2 \varphi k \rho^2}  \]

So,
\[ \left| \left\{ A \in \mathcal{T}_{\varphi} \backslash \{\upsilon\} : \left| \frac{1}{\varphi k } \left|\upsilon \cap A \right| - \varphi \right| > \rho  \right\} \right|  \leq  |\mathcal{T}_{\varphi}| \cdot 2e^{-2 \varphi k \rho^2}  \]

The result now follows, considering that $\Upsilon$ is uniform over $|\mathcal{T}_{\varphi}| - 1$ possibilities.

\end{proof}

\begin{lemma}
\label{lem:desired_beta_hbeta}

Let $\alpha_0 \in (0,1]$. Then, $\exists (\alpha_n)_{n \in \mathbb{N}}$ such that  $\forall n \in \mathbb{N}, \alpha_n\in [0, \alpha_0) \cap \mathbb{Q}$, $\alpha_n \longrightarrow \alpha_0$ as $n \longrightarrow \infty$ and $H(\alpha_n) \in \mathbb{R} \backslash \mathbb{Q}$. 

\end{lemma}

\begin{proof}

For any $n \in \mathbb{N}$, let $k_n \in \mathbb{N}$ be such that $1/3^{k_n} < 1/n$. Then, it is easy to check that $\exists a_n \in \mathbb{N}$ such that $\alpha_0  - 1/n < a_n / 3^{k_n} < \alpha_0$. Let $\alpha_n = a_n/3^{k_n}$, $n \in \mathbb{N}$ such that $\alpha_0  - 1/n < \alpha_n < \alpha_0$. Clearly, $\alpha_n \in [0, \alpha_0) \cap \mathbb{Q}$ and $\alpha_n \longrightarrow \alpha_0$ as $n \longrightarrow \infty$. Also,
\begin{align*}
-H(\alpha_n) & = -H \left( \frac{a_n}{3^{k_n}} \right) \\
                    & =  \frac{a_n}{3^{k_n}} \log \left( \frac{a_n}{3^{k_n}} \right) + \left( 1 - \frac{a_n}{3^{k_n}} \right) \log \left( 1 - \frac{a_n}{3^{k_n}} \right) \\
                    & = \log \left(   \left( \frac{a_n}{3^{k_n}} \right)^{\frac{a_n}{3^{k_n}}}  \cdot  \left( \frac{3^{k_n} - a_n}{3^{k_n}} \right)^{\frac{3^{k_n} - a_n}{3^{k_n}}}   \right)  \\
                    & = \frac{1}{3^{k_n}} \log \left(    \left( \frac{a_n}{3^{k_n}} \right)^{a_n}  \cdot  \left( \frac{3^{k_n} - a_n}{3^{k_n}} \right)^{3^{k_n} - a_n}    \right) \\
                    & = \frac{1}{3^{k_n}} \log \left(    \frac{b_n}{3^{j_n}}  \right) 
\end{align*}

where $b_n,j_n \in \mathbb{N}$. Thus,
\begin{equation*}
H(\alpha_n) = \frac{1}{3^{k_n}} \log \left(    \frac{3^{j_n}}{b_n}  \right) 
\end{equation*}

Suppose that $\log \left(    3^{j_n}/b_n    \right)$ is rational. That is, $\exists p,q \in \mathbb{N}, q \neq 0$ such that :
\[   \log \left(    \frac{3^{j_n}}{b_n}    \right) = \frac{p}{q}  \]

This implies that $2^{\frac{p}{q}} = 3^{j_n}/b_n$. That is, $2^p = 3^{qj_n}/b^q_n$. Hence,
\[ b^q_n = \frac{3^{q j_n}}{2^p}  \]

This is a contradiction since the RHS cannot be an integer, as its numerator is an odd number while the denominator is an even number. Thus, $\log \left(    3^{j_n}/b_n    \right)$ is irrational. As a result, $H(\alpha_n) = (1/3^{k_n}) \cdot \log \left(    3^{j_n}/b_n  \right)$ is also irrational.

\end{proof}

\begin{lemma}
\label{lem:bob_honest_bobeve_know_little_lt}
Suppose $\epsilon_1 \leq 1/2$ and Bob is honest. Then, 
\begin{equation*}
H( F_{\overline{U}}(\boldsymbol{X}|_{\boldsymbol{L}_{\overline{U}}}) | F_{\overline{U}}, V_{BE} ) \geq (1 - \xi) \cdot \left(   (\epsilon_1 \epsilon_2 - 5 \delta - 2\tilde{\delta} - \delta')n - \frac{2^{-(\delta + \delta')  n}}{\ln 2}  \right)
\end{equation*}
where $\xi \longrightarrow 0$ exponentially fast as $n \longrightarrow \infty$.
\end{lemma}

\begin{proof}
Suppose $v_{BE} = (u, \boldsymbol{y}, \boldsymbol{z}, \boldsymbol{l}_0, \boldsymbol{l}_1, \boldsymbol{s}, \boldsymbol{m})$. Then,
\begin{align*}
R(\boldsymbol{X}|_{\boldsymbol{L}_{\overline{U}}} | V_{BE} = v_{BE}) & = R(\boldsymbol{X}|_{\boldsymbol{l}_{\overline{U}}} | u, \boldsymbol{y}, \boldsymbol{z}, \boldsymbol{l}_0, \boldsymbol{l}_1, \boldsymbol{s}, \boldsymbol{m}) \\
& = R(\boldsymbol{X}|_{\boldsymbol{l}_{\overline{U}}} | u, \boldsymbol{y}, \boldsymbol{z}, \boldsymbol{l}_u, \boldsymbol{l}_{\overline{u}}, \boldsymbol{s}, \boldsymbol{m}) \\
& \stackrel{\text{(a)}}{=} R(\boldsymbol{X}|_{\boldsymbol{l}_{\overline{U}}} | u, \boldsymbol{y}, \boldsymbol{z}, \boldsymbol{l}_{\overline{u}}) \\
& = R(\boldsymbol{X}|_{\boldsymbol{l}_{\overline{U}}} | u, \boldsymbol{y}, \boldsymbol{z}, \boldsymbol{l}_{\overline{u}}, \boldsymbol{\psi}|_{\boldsymbol{l}_{\overline{u}}}) \\
& \stackrel{\text{(b)}}{=} R(\boldsymbol{X}|_{\boldsymbol{l}_{\overline{U}}} | \boldsymbol{\psi}|_{\boldsymbol{l}_{\overline{u}}}) \\
& = \#_e(\boldsymbol{\psi}|_{\boldsymbol{l}_{\overline{u}}})
\end{align*}

where (a) follows since $\boldsymbol{X}|_{\boldsymbol{L}_{\overline{U}}} - U,\boldsymbol{X},\boldsymbol{Y},\boldsymbol{L}_{\overline{U}} - \boldsymbol{L}_U,\boldsymbol{S},\boldsymbol{M}$ is a Markov chain and (b) follows since $\boldsymbol{X}|_{\boldsymbol{L}_{\overline{U}}} -  \boldsymbol{\Psi}|_{\boldsymbol{L}_{\overline{U}}} - U,\boldsymbol{Y},\boldsymbol{Z},\boldsymbol{L}_{\overline{U}}$ is a Markov chain.

Whenever $\#_e(\boldsymbol{\psi}|_{\boldsymbol{l}_{\overline{u}}}) \geq (\epsilon_2 - \delta)(\beta n -\gamma n) = (\epsilon_2 - \delta)(\epsilon_1 -\delta)n$, then by applying Lemma~\ref{lem:privacy_amplification} we get:
\begin{align*}
  H( F_{\overline{U}}(\boldsymbol{X}|_{\boldsymbol{L}_{\overline{U}}}) | F_{\overline{U}}, V_{BE} = v_{BE} ) & \geq (\epsilon_1 \epsilon_2 - 5 \delta - 2\tilde{\delta} - \delta')n - \frac{ 2^{( \epsilon_1 \epsilon_2 - 5 \delta - 2\tilde{\delta} - \delta' - (\epsilon_2 - \delta)(\epsilon_1 - \delta)  )n} }{\ln 2}  \\
 & \geq (\epsilon_1 \epsilon_2 - 5 \delta - 2\tilde{\delta} - \delta')n - \frac{2^{-(\delta + \delta') n}}{\ln 2}
\end{align*}

Also, by Chernoff's bound $P[\#_e(\boldsymbol{\Psi}|_{\boldsymbol{L}_{\overline{U}}}) \geq (\epsilon_2 - \delta)(\beta n -\gamma n)] \geq 1 - \xi$, where $\xi \longrightarrow 0$ exponentially fast as $n \longrightarrow \infty$. Thus, we have 
\begin{equation*}
H( F_{\overline{U}}(\boldsymbol{X}|_{\boldsymbol{L}_{\overline{U}}}) | F_{\overline{U}}, V_{BE} ) \geq (1 - \xi) \cdot \left(   (\epsilon_1 \epsilon_2 - 5 \delta - 2\tilde{\delta} - \delta')n - \frac{2^{-(\delta + \delta') n}}{\ln 2}  \right)
\end{equation*}

\end{proof}

\begin{lemma}
\label{lem:bob_honest_bobeve_know_little_gt}
Suppose $\epsilon_1 > 1/2$ and Bob is honest. Then,
\[       H( F_{\overline{U}}(\boldsymbol{X}|_{L_{\overline{\Phi}} \backslash L_0 \cap L_1}) | F_{\overline{U}}, V_{BE} ) \geq (1 - \xi) \cdot \left(   \beta n(\epsilon_1 \epsilon_2 - 3 \delta - \delta') - \frac{2^{-(\delta + \delta') \beta n}}{\ln 2}  \right)     \]
where $\xi \longrightarrow 0$ exponentially fast as $n \longrightarrow \infty$.
\end{lemma}

\begin{proof}
Suppose $ v_{BE} = (u,\boldsymbol{y},\boldsymbol{z},\boldsymbol{s},\boldsymbol{m})$. Then,
\begin{align*}
R(\boldsymbol{X}|_{L_{\overline{\Phi}} \backslash L_0 \cap L_1} | V_{BE} = v_{BE}) & = R(\boldsymbol{X}|_{l_{\overline{\Phi}} \backslash l_0 \cap l_1} | u,\boldsymbol{y},\boldsymbol{z},\boldsymbol{s},\boldsymbol{m}) \\
& \stackrel{\text{(a)}}{=} R(\boldsymbol{X}|_{l_{\overline{\Phi}} \backslash l_0 \cap l_1} | u,\boldsymbol{y},\boldsymbol{z},\boldsymbol{s},\boldsymbol{m}, l_{\overline{\Phi}} \backslash l_0 \cap l_1) \\
& \stackrel{\text{(b)}}{=} R(\boldsymbol{X}|_{l_{\overline{\Phi}} \backslash l_0 \cap l_1} | \boldsymbol{y},\boldsymbol{z},l_{\overline{\Phi}} \backslash l_0 \cap l_1) \\
& = R(\boldsymbol{X}|_{l_{\overline{\Phi}} \backslash l_0 \cap l_1} | \boldsymbol{y},\boldsymbol{z},l_{\overline{\Phi}} \backslash l_0 \cap l_1, \boldsymbol{\psi}|_{L_{\overline{\phi}} \backslash l_0 \cap l_1}) \\
& \stackrel{\text{(c)}}{=}  R(\boldsymbol{X}|_{l_{\overline{\Phi}} \backslash l_0 \cap l_1} | \boldsymbol{\psi}|_{l_{\overline{\phi}} \backslash l_0 \cap l_1}) \\
 & = \#_e(\boldsymbol{\psi}|_{l_{\overline{\phi}} \backslash l_0 \cap l_1})
\end{align*}

where (a) follows since $L_0,L_1,\Phi$ are functions of $(\boldsymbol{S},\boldsymbol{M})$, (b) follows since $\boldsymbol{X}|_{L_{\overline{\Phi}} \backslash L_0 \cap L_1} - \boldsymbol{Y},\boldsymbol{Z},L_{\overline{\Phi}} \backslash L_0 \cap L_1 - U,\boldsymbol{S},\boldsymbol{M}$ is a Markov chain and (c) follows since $\boldsymbol{X}|_{L_{\overline{\Phi}} \backslash L_0 \cap L_1}  - \boldsymbol{\Psi}|_{L_{\overline{\Phi}} \backslash L_0 \cap L_1} - \boldsymbol{Y},\boldsymbol{Z},L_{\overline{\Phi}} \backslash L_0 \cap L_1$ is a Markov chain.

Whenever $\#_e(\boldsymbol{\psi}|_{l_{\overline{\phi}} \backslash l_0 \cap l_1}) \geq (\epsilon_1 \epsilon_2 - 2\delta)\beta n$ , then by applying Lemma~\ref{lem:privacy_amplification}, we get
\begin{align*}
  H( F_{\overline{U}}(\boldsymbol{X}|_{L_{\overline{\Phi}} \backslash L_0 \cap L_1}) | F_{\overline{U}}, V_{BE} = v_{BE} ) & \geq \beta n(\epsilon_1 \epsilon_2 - 3 \delta - \delta') - \frac{ 2^{ \beta n(\epsilon_1 \epsilon_2 - 3 \delta - \delta') - (\epsilon_1 \epsilon_2 - 2\delta)\beta n  } }{\ln 2}  \\
    & = \beta n(\epsilon_1 \epsilon_2 - 3 \delta - \delta') - \frac{2^{-(\delta + \delta') \beta n}}{\ln 2}
\end{align*}

Recall that $L_{\overline{\Phi}} \thicksim \text{Unif}(\mathcal{T} \backslash L_{\Phi})$.  By a simple application of Lemma~\ref{lem:hoeffding} alongwith Chernoff's bound, we get $P[\#_e(\boldsymbol{\Psi}|_{L_{\overline{\Phi}}}) \geq (\epsilon_1 \epsilon_2 - 2\delta)\beta n] \geq 1 - \xi$, where $\xi \longrightarrow 0$ exponentially fast as $n \longrightarrow \infty$. Since $L_0 \cap L_1$ positions are unerased in $\boldsymbol{Y}$, we have $P[\#_e(\boldsymbol{\Psi}|_{L_{\overline{\Phi}} \backslash L_0 \cap L_1}) \geq (\epsilon_1 \epsilon_2 - 2\delta)\beta n] \geq 1 - \xi$. As a result,
\[   H( F_{\overline{U}}(\boldsymbol{X}|_{L_{\overline{\Phi}} \backslash L_0 \cap L_1}) | F_{\overline{U}}, V_{BE} ) \geq (1 - \xi) \cdot \left(   \beta n(\epsilon_1 \epsilon_2 - 3 \delta - \delta') - \frac{2^{-(\delta + \delta') \beta n}}{\ln 2}  \right) \]

\end{proof}

\begin{lemma}
\label{lem:eve_knows_little_lt}
Suppose $\epsilon_1 \leq 1/2$ and suppose Alice, Bob are honest. Then,
\begin{enumerate}
  \item $H(F_0(\boldsymbol{X}|_{\boldsymbol{L}_0}) | F_0, U, V_E) \geq  (1 - \xi) \cdot \left(     (\epsilon_1 \epsilon_2 - 5 \delta - 2\tilde{\delta} - \delta')n - \frac{2^{-(\delta + \delta') n}}{\ln 2}     \right) $
  \item $H(F_1(\boldsymbol{X}|_{\boldsymbol{L}_1}) | F_1, U, V_E) \geq  (1 - \xi) \cdot \left(     (\epsilon_1 \epsilon_2 - 5 \delta - 2\tilde{\delta} - \delta')n - \frac{2^{-(\delta + \delta') n}}{\ln 2}     \right) $
\end{enumerate}
where $\xi \longrightarrow 0$ exponentially fast as $n \longrightarrow \infty$.
\end{lemma}

\begin{proof}
Suppose $ v_E = \left(  \boldsymbol{z},\boldsymbol{l}_0,\boldsymbol{l}_1,\boldsymbol{m},\boldsymbol{\pi},\theta,\boldsymbol{y}|_{\boldsymbol{l}_0|_{\boldsymbol{j}_{\overline{\theta}}}}, \boldsymbol{y}|_{\boldsymbol{l}_1|_{\boldsymbol{j}_{\theta}}} \right)$. Then,

\begin{enumerate}

\item  
\begin{align*}
R(\boldsymbol{X}|_{\boldsymbol{L}_0} | U = u, V_E = v_E) & = R \left( \boldsymbol{X}|_{\boldsymbol{l}_0} | u,\boldsymbol{z},\boldsymbol{l}_0,\boldsymbol{l}_1,\boldsymbol{m},\boldsymbol{\pi},\theta,\boldsymbol{y}|_{\boldsymbol{l}_0|_{\boldsymbol{j}_{\overline{\theta}}}}, \boldsymbol{y}|_{\boldsymbol{l}_1|_{\boldsymbol{j}_{\theta}}} \right) \\
& =  R \left( \boldsymbol{X}|_{\boldsymbol{l}_0} | u,\boldsymbol{z},\boldsymbol{l}_0,\boldsymbol{l}_1,\boldsymbol{m},\boldsymbol{\pi},\theta,\boldsymbol{y}|_{\boldsymbol{l}_0|_{\boldsymbol{j}_{\overline{\theta}}}}, \boldsymbol{y}|_{\boldsymbol{l}_1|_{\boldsymbol{j}_{\theta}}}, \boldsymbol{j}_{\overline{\theta}} \right) \\
& \stackrel{\text{(a)}}{=}  R \left( \boldsymbol{X}|_{\boldsymbol{l}_0} | \boldsymbol{z},\boldsymbol{l}_0,\boldsymbol{y}|_{\boldsymbol{l}_0|_{\boldsymbol{j}_{\overline{\theta}}}}, \boldsymbol{j}_{\overline{\theta}} \right) \\
& \stackrel{\text{(b)}}{=} R \left( \boldsymbol{X}|_{\boldsymbol{l}_0|_{\boldsymbol{j}^c_{\overline{\theta}} } } | \boldsymbol{z},\boldsymbol{l}_0,\boldsymbol{y}|_{\boldsymbol{l}_0|_{\boldsymbol{j}_{\overline{\theta}}}}, \boldsymbol{j}_{\overline{\theta}} \right) \\
& \stackrel{\text{(c)}}{=} R \left( \boldsymbol{X}|_{\boldsymbol{l}_0|_{\boldsymbol{j}^c_{\overline{\theta}}} } | \boldsymbol{z}|_{\boldsymbol{l}_0|_{\boldsymbol{j}^c_{\overline{\theta}}} } \right) \\
& = \#_e \left( \boldsymbol{z}|_{\boldsymbol{l}_0|_{\boldsymbol{j}^c_{\overline{\theta}}} } \right)
\end{align*}

where (a) follows since $\boldsymbol{X}|_{\boldsymbol{L}_0} - \boldsymbol{Z},\boldsymbol{L}_0,\boldsymbol{J}_{\overline{\theta}}, \boldsymbol{Y}|_{\boldsymbol{L}_0|_{\boldsymbol{J}_{\overline{\Theta}}}} - U,\boldsymbol{L}_1,\boldsymbol{M},\boldsymbol{\Pi},\Theta,\boldsymbol{Y}|_{\boldsymbol{L}_1|_{\boldsymbol{J}_{\Theta}}}$ is a Markov chain, (b) follows since $\boldsymbol{Y}|_{\boldsymbol{L}_0|_{\boldsymbol{J}_{\overline{\Theta}}}}$ is the same as $\boldsymbol{X}|_{\boldsymbol{L}_0|_{\boldsymbol{J}_{\overline{\Theta}}}}$ and (c) follows since $\boldsymbol{X}|_{\boldsymbol{L}_0|_{\boldsymbol{J}^c_{\overline{\Theta}}} } - \boldsymbol{Z}|_{\boldsymbol{L}_0|_{\boldsymbol{J}^c_{\overline{\Theta}}} } - \boldsymbol{Z},\boldsymbol{L}_0,\boldsymbol{J}_{\overline{\Theta}},\boldsymbol{Y}|_{\boldsymbol{L}_0|_{\boldsymbol{J}_{\overline{\Theta}}}}$ is a Markov chain.

Whenever $\#_e \left( \boldsymbol{z}|_{\boldsymbol{l}_0|_{\boldsymbol{j}^c_{\overline{\theta}}} } \right) \geq (\epsilon_2 - \delta) \left| \boldsymbol{l}_0|_{\boldsymbol{j}^c_{\overline{\theta}}} \right| = (\epsilon_2 - \delta)(\beta n - \gamma n)$, then applying Lemma~\ref{lem:privacy_amplification}, we get:
\begin{align*}
H(F_0(\boldsymbol{X}|_{\boldsymbol{L}_0}) | F_0, U = u, V_E = v_E) & \geq (\epsilon_1 \epsilon_2 - 5 \delta - 2\tilde{\delta} - \delta')n - \frac{2^{(\epsilon_1 \epsilon_2 - 5 \delta - 2\tilde{\delta} - \delta')n -  (\epsilon_2 - \delta)(\epsilon_1 - \delta)n}}{\ln 2} \\
 & \geq (\epsilon_1 \epsilon_2 - 5 \delta - 2\tilde{\delta} - \delta')n - \frac{2^{-(\delta + \delta') n}}{\ln 2}
\end{align*}

Also, by Chernoff's bound, $P \left[\#_e \left( \boldsymbol{Z}|_{\boldsymbol{L}_0|_{\boldsymbol{J}^c_{\overline{\Theta}}} } \right) \geq (\epsilon_2 - \delta) \left| \boldsymbol{L}_0|_{\boldsymbol{J}^c_{\overline{\Theta}}} \right| \right] \geq 1 - \xi$, where $\xi \longrightarrow 0$ exponentially fast as $n \longrightarrow \infty$. Thus, we get 
\[    H(F_0(\boldsymbol{X}|_{\boldsymbol{L}_0}) | F_0, U, V_E) \geq  (1 - \xi) \cdot \left(     (\epsilon_1 \epsilon_2 - 5 \delta - 2\tilde{\delta} - \delta')n - \frac{2^{-(\delta + \delta') n}}{\ln 2}     \right)   \]

\item The argument to show that 
\[    H(F_1(\boldsymbol{X}|_{\boldsymbol{L}_1}) | F_1, U, V_E) \geq  (1 - \xi) \cdot \left(     (\epsilon_1 \epsilon_2 - 5 \delta - 2\tilde{\delta} - \delta')n - \frac{2^{-(\delta + \delta') n}}{\ln 2}     \right)   \]
is very similar to the above argument.

\end{enumerate}

\end{proof}

\begin{lemma}
\label{lem:eve_knows_little_gt}
Suppose $\epsilon_1 > 1/2$ and suppose Alice, Bob are honest. Then,
\begin{align*}
H(F_U( \boldsymbol{X}|_{L_{\Phi} \backslash L_0 \cap L_1}  ) | F_U, U, V_E) & \geq (1 - \xi) \cdot \left(    \beta n(\epsilon_1 \epsilon_2 - 3 \delta - \delta') - \frac{2^{-(\delta + \delta') \beta n}}{\ln 2}    \right) \\
H(F_{\overline{U}}( \boldsymbol{X}|_{L_{\overline{\Phi}} \backslash L_0 \cap L_1}  ) | F_{\overline{U}}, U, V_E)  & \geq (1 - \xi) \cdot \left(    \beta n(\epsilon_1 \epsilon_2 - 3 \delta - \delta') - \frac{2^{-(\delta + \delta') \beta n}}{\ln 2}    \right)
\end{align*}
where $\xi \longrightarrow 0$ exponentially fast as $n \longrightarrow \infty$.
\end{lemma}

\begin{proof}
Suppose $v_E = (\boldsymbol{z},\boldsymbol{m},\boldsymbol{\pi},\theta,\boldsymbol{y}|_{l_0 \cap l_1})$. Then, for $i = 0,1$ :
\begin{align*}
R(\boldsymbol{X}|_{L_i \backslash L_0 \cap L_1} | U = u, V_E = v_E) & = R(\boldsymbol{X}|_{l_i \backslash l_0 \cap l_1} | u,\boldsymbol{z},\boldsymbol{m},\boldsymbol{\pi},\theta,\boldsymbol{y}|_{l_0 \cap l_1}) \\
& = R(\boldsymbol{X}|_{l_i \backslash l_0 \cap l_1} | u,\boldsymbol{z},\boldsymbol{m},\boldsymbol{\pi},\theta,\boldsymbol{y}|_{l_0 \cap l_1}, l_i \backslash l_0 \cap l_1) \\
& \stackrel{\text{(a)}}{=} R(\boldsymbol{X}|_{l_i \backslash l_0 \cap l_1} | \boldsymbol{z},l_i \backslash l_0 \cap l_1) \\
& = R(\boldsymbol{X}|_{l_i \backslash l_0 \cap l_1} | \boldsymbol{z}|_{l_i \backslash l_0 \cap l_1}) \\
& = \#_e(\boldsymbol{z}|_{l_i \backslash l_0 \cap l_1})
\end{align*}

where (a) follows since $\boldsymbol{X}|_{L_i \backslash L_0 \cap L_1} - \boldsymbol{Z},L_i \backslash L_0 \cap L_1 - U,\boldsymbol{M},\boldsymbol{\Pi},\Theta,\boldsymbol{Y}|_{L_0 \cap L_1}$ is a Markov chain.

Whenever $\#_e(\boldsymbol{z}|_{l_i \backslash l_0 \cap l_1}) \geq (\epsilon_2 - \delta)|l_i \backslash l_0 \cap l_1| = (\epsilon_2 - \delta)(\beta n - |l_0 \cap l_1|)$, then by applying Lemma~\ref{lem:privacy_amplification} we get:
\begin{align*}
H(F_U( \boldsymbol{X}|_{L_{\Phi} \backslash L_0 \cap L_1}  ) | F_U, U = u, V_E = v_E) & \geq  \beta n(\epsilon_1 \epsilon_2 - 3 \delta - \delta') - \frac{2^{ \beta n(\epsilon_1 \epsilon_2 - 3 \delta - \delta') -  (\epsilon_2 - \delta)(\beta n - |l_0 \cap l_1|) }}{\ln 2} \\
 & \geq \beta n(\epsilon_1 \epsilon_2 - 3 \delta - \delta') - \frac{2^{-(\delta + \delta') \beta n}}{\ln 2}
\end{align*}

By Chernoff's bound we know that, for $i=0,1$, $P[\#_e(\boldsymbol{Z}|_{L_i \backslash L_0 \cap L_1}) \geq (\epsilon_2 - \delta)|L_i \backslash L_0 \cap L_1|] \geq 1 - \xi$, where $\xi \longrightarrow 0$ exponentially fast as $n \longrightarrow \infty$. As a result, 
\[  H(F_U( \boldsymbol{X}|_{L_{\Phi} \backslash L_0 \cap L_1}  ) | F_U, U, V_E) \geq (1 - \xi) \cdot \left(    \beta n(\epsilon_1 \epsilon_2 - 3 \delta - \delta') - \frac{2^{-(\delta + \delta')\beta n}}{\ln 2}    \right).   \]

By a similar argument,
\[   H(F_{\overline{U}}( \boldsymbol{X}|_{L_{\overline{\Phi}} \backslash L_0 \cap L_1}  ) | F_{\overline{U}}, U, V_E)  \geq (1 - \xi) \cdot \left(    \beta n(\epsilon_1 \epsilon_2 - 3 \delta - \delta') - \frac{2^{-(\delta + \delta')\beta n}}{\ln 2}    \right).  \]

\end{proof}


\begin{lemma}
\label{lem:bob_malicious_bobeve_know_little_lt}
Suppose $\epsilon_1 \leq 1/2$, Bob is malicious and $\#_e(\boldsymbol{\Psi}|_{\boldsymbol{L}_1}) \geq (\epsilon_1 \epsilon_2 - 4\delta - 2\tilde{\delta})n$. Then,
\[  R(\boldsymbol{X}|_{\boldsymbol{L}_1} | V_{BE} = v_{BE})  \geq (\epsilon_1 \epsilon_2 - 4\delta - 2\tilde{\delta})n \]
\end{lemma}

\begin{proof}
Suppose $v_{BE} = (\boldsymbol{y},\boldsymbol{z},\boldsymbol{l}_0,\boldsymbol{l}_1,\boldsymbol{m},\boldsymbol{\pi},\theta)$. Then,
\begin{align*}
R(\boldsymbol{X}|_{\boldsymbol{L}_1} | V_{BE} = v_{BE}) & = R(\boldsymbol{X}|_{\boldsymbol{l}_1} | \boldsymbol{y},\boldsymbol{z},\boldsymbol{l}_0,\boldsymbol{l}_1,\boldsymbol{m},\boldsymbol{\pi},\theta) \\
& \stackrel{\text{(a)}}{=} R(\boldsymbol{X}|_{\boldsymbol{l}_1} | \boldsymbol{y},\boldsymbol{z},\boldsymbol{l}_1) \\
& = R(\boldsymbol{X}|_{\boldsymbol{l}_1} | \boldsymbol{y},\boldsymbol{z},\boldsymbol{l}_1, \boldsymbol{\psi}|_{\boldsymbol{l}_1}) \\
& \stackrel{\text{(b)}}{=} R(\boldsymbol{X}|_{\boldsymbol{l}_1} | \boldsymbol{\psi}|_{\boldsymbol{l}_1}) \\
& = \#_e(\boldsymbol{\psi}|_{\boldsymbol{l}_1}) \\
& \geq (\epsilon_1 \epsilon_2 - 4\delta - 2\tilde{\delta})n
\end{align*}

where (a) follows since $\boldsymbol{X}|_{\boldsymbol{L}_1} - \boldsymbol{Y},\boldsymbol{Z},\boldsymbol{L}_1 - \boldsymbol{L}_0,\boldsymbol{M},\boldsymbol{\Pi},\Theta$ is a Markov chain and (b) follows since $\boldsymbol{X}|_{\boldsymbol{L}_1} - \boldsymbol{\Psi}|_{\boldsymbol{L}_1} -  \boldsymbol{Y},\boldsymbol{Z},\boldsymbol{L}_1$ is a Markov chain.

\end{proof}


\begin{lemma}
\label{lem:bob_malicious_bobeve_know_little_gt}
Suppose $\epsilon_1 > 1/2$, Bob is malicious and $\#_e(\boldsymbol{\Psi}|_{L_0 \backslash L_0 \cap L_1}) \geq \beta n (\epsilon_1\epsilon_2 - 2\delta)$. Then,
\[  R(\boldsymbol{X}|_{L_0 \backslash L_0 \cap L_1} | V_{BE} = v_{BE})  \geq \beta n(\epsilon_1 \epsilon_2 - 2\delta) \]
\end{lemma}

\begin{proof}
Suppose $v_{BE} = (\boldsymbol{y},\boldsymbol{z},\boldsymbol{m},\boldsymbol{\pi},\theta)$. Then,
\begin{align*}
R(\boldsymbol{X}|_{L_0 \backslash L_0 \cap L_1} | V_{BE} = v_{BE}) & = R(\boldsymbol{X}|_{l_0 \backslash l_0 \cap l_1} | \boldsymbol{y},\boldsymbol{z},\boldsymbol{m},\boldsymbol{\pi},\theta) \\
& = R(\boldsymbol{X}|_{l_0 \backslash l_0 \cap l_1} | \boldsymbol{y},\boldsymbol{z},\boldsymbol{m},\boldsymbol{\pi},\theta, l_0 \backslash l_0 \cap l_1) \\
& \stackrel{\text{(a)}}{=} R(\boldsymbol{X}|_{l_0 \backslash l_0 \cap l_1} | \boldsymbol{y},\boldsymbol{z},l_0 \backslash l_0 \cap l_1) \\
& =  R(\boldsymbol{X}|_{l_0 \backslash l_0 \cap l_1} | \boldsymbol{y},\boldsymbol{z},l_0 \backslash l_0 \cap l_1, \boldsymbol{\psi}|_{l_0 \backslash l_0 \cap l_1}) \\
& \stackrel{\text{(b)}}{=} R(\boldsymbol{X}|_{l_0 \backslash l_0 \cap l_1} | \boldsymbol{\psi}|_{l_0 \backslash l_0 \cap l_1}); \\
& = \#_e(\boldsymbol{\psi}|_{l_0 \backslash l_0 \cap l_1}) \\
& \geq \beta n (\epsilon_1\epsilon_2 - 2\delta)
\end{align*}

where (a) follows since $\boldsymbol{X}|_{L_0 \backslash L_0 \cap L_1} - Y^,\boldsymbol{Z},L_0 \backslash L_0 \cap L_1 - \boldsymbol{M},\boldsymbol{\Pi},\Theta$ is a Markov chain and (b) follows since $\boldsymbol{X}|_{L_0 \backslash L_0 \cap L_1} - \boldsymbol{\Psi}|_{L_0 \backslash L_0 \cap L_1} - \boldsymbol{Y},\boldsymbol{Z},L_0 \backslash L_0 \cap L_1$ is a Markov chain.

\end{proof}


\section{Independent oblivious transfers over a broadcast channel: Proofs of Lemmas~\ref{lem:c2p_ach_indep},~\ref{lem:small_quant_2p_hbc_indep}}


\subsection{Proof of Lemma~\ref{lem:c2p_ach_indep}}
\label{appndx:proof_ach_2p_hbc_indep}

When $\epsilon_1 \leq 1/2$, Protocol~\ref{protocol:C2P_symmetric} is the same as Protocol~\ref{protocol:C2P}. And so, this proof is the same as the proof of Lemma~\ref{lem:c2p_ach_wtap}. As a result, we consider only the case when $\epsilon_1 > 1/2$ in this proof. We use a sequence $(\mathcal{P}_n)_{n \in \mathbb{N}}$ of Protocol~\ref{protocol:C2P_symmetric} instances and we show that (\ref{eqn:ach_2p_indep_0}) - (\ref{eqn:ach_2p_indep_7}) are satisfied for $(\mathcal{P}_n)_{n \in \mathbb{N}}$.

Note that:
\begin{align*}
V_A & = \{ \boldsymbol{K}_0, \boldsymbol{K}_1, \boldsymbol{J}_0, \boldsymbol{J}_1, \boldsymbol{X}, \boldsymbol{\Lambda} \} \\
V_B & = \{ U, \boldsymbol{Y}, \boldsymbol{\Lambda} \} \\
V_C & = \{ W, \boldsymbol{Z}, \boldsymbol{\Lambda} \}
\end{align*}

where $\boldsymbol{\Lambda} = \{\tilde{\boldsymbol{\Lambda}}, \boldsymbol{\Lambda}_{\text{twoparty}} \}$, with $\tilde{\boldsymbol{\Lambda}} = \{L_0,L_1,L,F_0,F_1,\boldsymbol{K}_0 \oplus F_0(\boldsymbol{X}|_{L_0}, \boldsymbol{K}_1 \oplus F_1(\boldsymbol{X}|_{L_1})\}$ and $\boldsymbol{\Lambda}_{\text{twoparty}}$ denoting the public messages exchanged during the execution of the two-party OT protocol \cite{ot2007} between Alice and Cathy. 

Let $\Upsilon$ be the event that $\mathcal{P}_n$ aborts in Step~\ref{step:2p_indep_abort0}. Then, due to Chernoff's bound, $P[\Upsilon = 1] \longrightarrow 0$ exponentially fast as $n \longrightarrow \infty$. As in proofs of Lemma~\ref{lem:c2p_ach_wtap} and Lemma~\ref{lem:c1p_ach_wtap}, it suffices to prove that the conditional versions of (\ref{eqn:ach_2p_indep_0}) - (\ref{eqn:ach_2p_indep_7}), conditioned on the event $\Upsilon = 0$, hold for $(\mathcal{P}_n)_{n \in \mathbb{N}}$. The arguments in rest of this proof are all implicitly conditioned on $\Upsilon = 0$.

\begin{enumerate} 

\item (\ref{eqn:ach_2p_indep_0}) holds for $(\mathcal{P}_n)_{n \in \mathbb{N}}$ for the same reasons that (\ref{eqn:ach_2p_wtap_0}) holds for Protocol~\ref{protocol:C2P}.

\item (\ref{eqn:ach_2p_indep_1}) holds for $(\mathcal{P}_n)_{n \in \mathbb{N}}$ due to the correctness of the two-party OT protocol \cite{ot2007}.

\item To show that (\ref{eqn:ach_2p_indep_2}) holds for $(\mathcal{P}_n)_{n \in \mathbb{N}}$, we proceed as follows:
\begin{align*}
I(\boldsymbol{K}_{\overline{U}}, \boldsymbol{J}_{\overline{W}} ; V_B,V_C) & = I(\boldsymbol{K}_{\overline{U}}, \boldsymbol{J}_{\overline{W}} ; U,W, \boldsymbol{Y},  \boldsymbol{Z}, \boldsymbol{\Lambda}) \\
& = I(\boldsymbol{K}_{\overline{U}}, \boldsymbol{J}_{\overline{W}} ; U,W, \boldsymbol{Y},  \boldsymbol{Z}, \tilde{\boldsymbol{\Lambda}}, \boldsymbol{\Lambda}_{\text{twoparty}}) \\
& = I(\boldsymbol{K}_{\overline{U}} ; U,W, \boldsymbol{Y},  \boldsymbol{Z}, \tilde{\boldsymbol{\Lambda}}, \boldsymbol{\Lambda}_{\text{twoparty}}) \\ & \quad + I(\boldsymbol{J}_{\overline{W}} ; U,W, \boldsymbol{Y},  \boldsymbol{Z}, \tilde{\boldsymbol{\Lambda}}, \boldsymbol{\Lambda}_{\text{twoparty}} | \boldsymbol{K}_{\overline{U}}) \\
& \stackrel{\text{(a)}}{=} I(\boldsymbol{K}_{\overline{U}} ; U, \boldsymbol{Y},  \boldsymbol{Z}, \tilde{\boldsymbol{\Lambda}}) + I(\boldsymbol{J}_{\overline{W}} ; U,W, \boldsymbol{Y},  \boldsymbol{Z}, \tilde{\boldsymbol{\Lambda}}, \boldsymbol{\Lambda}_{\text{twoparty}} | \boldsymbol{K}_{\overline{U}}) \\
& = I(\boldsymbol{K}_{\overline{U}} ; U, \boldsymbol{Y},  \boldsymbol{Z}, \tilde{\boldsymbol{\Lambda}}) + I(\boldsymbol{J}_{\overline{W}} ; U,W, \boldsymbol{Y},  \boldsymbol{Z}, \tilde{\boldsymbol{\Lambda}}, \boldsymbol{\Lambda}_{\text{twoparty}}, \boldsymbol{K}_{\overline{U}}) \\
& \stackrel{\text{(b)}}{=} I(\boldsymbol{K}_{\overline{U}} ; U, \boldsymbol{Y},  \boldsymbol{Z}, \tilde{\boldsymbol{\Lambda}}) + I(\boldsymbol{J}_{\overline{W}} ; W, \boldsymbol{Y},  \boldsymbol{Z}, L, \boldsymbol{\Lambda}_{\text{twoparty}}) \\
& \stackrel{\text{(c)}}{=} I(\boldsymbol{K}_{\overline{U}} ; U, \boldsymbol{Y},  \boldsymbol{Z}, \tilde{\boldsymbol{\Lambda}}) + I(\boldsymbol{J}_{\overline{W}} ; W, \boldsymbol{Y}|_L,  \boldsymbol{Z}|_L, \boldsymbol{\Lambda}_{\text{twoparty}}) \\
& = I(\boldsymbol{K}_{\overline{U}} ; U, \boldsymbol{Y},  \boldsymbol{Z}, \tilde{\boldsymbol{\Lambda}}) + I(\boldsymbol{J}_{\overline{W}} ; W, \boldsymbol{Z}|_L,  \boldsymbol{\Lambda}_{\text{twoparty}})
\end{align*}

where (a) follows since $\boldsymbol{K}_{\overline{U}} - U, \boldsymbol{Y},  \boldsymbol{Z}, \tilde{\boldsymbol{\Lambda}} - W,\boldsymbol{\Lambda}_{\text{twoparty}}$ is a Markov chain, (b) follows since $\boldsymbol{J}_{\overline{W}} - W, \boldsymbol{Y},  \boldsymbol{Z}, L, \boldsymbol{\Lambda}_{\text{twoparty}} -  U,\tilde{\boldsymbol{\Lambda}}, \boldsymbol{K}_{\overline{U}}$ is a Markov chain and (c) follows since $\boldsymbol{J}_{\overline{W}} - W, \boldsymbol{Y}|_L,  \boldsymbol{Z}|_L, \boldsymbol{\Lambda}_{\text{twoparty}} - \boldsymbol{Y}, \boldsymbol{Z}, L$ is a Markov chain.

The first term above goes to zero for the same reasons that (\ref{eqn:ach_2p_wtap_1}) holds for Protocol~\ref{protocol:C2P}. The second term goes to zero due to the OT requirements being satisfied by the two-party OT protocol between Alice and Cathy over $\boldsymbol{X}|_L$.

\item To show that (\ref{eqn:ach_2p_indep_3}) holds for $(\mathcal{P}_n)_{n \in \mathbb{N}}$, we proceed as follows:
\begin{align*}
I(U ; V_A,V_C) & = I(U ; \boldsymbol{K}_0, \boldsymbol{K}_1, \boldsymbol{J}_0, \boldsymbol{J}_1, W, \boldsymbol{X}, \boldsymbol{Z}, \boldsymbol{\Lambda} ) \\
& = I(U ; \boldsymbol{K}_0, \boldsymbol{K}_1, \boldsymbol{J}_0, \boldsymbol{J}_1, W, \boldsymbol{X}, \boldsymbol{Z}, \tilde{\boldsymbol{\Lambda}}, \boldsymbol{\Lambda}_{\text{twoparty}} ) \\
& \stackrel{\text{(a)}}{=} I(U ; \boldsymbol{K}_0, \boldsymbol{K}_1, \boldsymbol{X}, \boldsymbol{Z}, \tilde{\boldsymbol{\Lambda}} ) 
\end{align*}

where (a) follows since $U - \boldsymbol{K}_0, \boldsymbol{K}_1, \boldsymbol{X}, \boldsymbol{Z}, \tilde{\boldsymbol{\Lambda}} - W,\boldsymbol{J}_0, \boldsymbol{J}_1, \boldsymbol{\Lambda}_{\text{twoparty}} $ is a Markov chain.

The above term goes to zero for the same reason that (\ref{eqn:ach_2p_wtap_2}) holds for Protocol~\ref{protocol:C2P}.

\item To show that (\ref{eqn:ach_2p_indep_4}) holds for $(\mathcal{P}_n)_{n \in \mathbb{N}}$, we proceed as follows:
\begin{align*}
I(W ; V_A,V_B) & = I(W ; \boldsymbol{K}_0, \boldsymbol{K}_1, \boldsymbol{J}_0, \boldsymbol{J}_1, U, \boldsymbol{X}, \boldsymbol{Y}, \boldsymbol{\Lambda} ) \\
& = I(W ; \boldsymbol{K}_0, \boldsymbol{K}_1, \boldsymbol{J}_0, \boldsymbol{J}_1, U, \boldsymbol{X}, \boldsymbol{Y}, \tilde{\boldsymbol{\Lambda}}, \boldsymbol{\Lambda}_{\text{twoparty}} ) \\
& \stackrel{\text{(a)}}{=} I(W ; \boldsymbol{J}_0, \boldsymbol{J}_1, \boldsymbol{X}, \boldsymbol{Y}, L, \boldsymbol{\Lambda}_{\text{twoparty}} ) \\
& \stackrel{\text{(b)}}{=} I(W ; \boldsymbol{J}_0, \boldsymbol{J}_1, \boldsymbol{X}|_L, \boldsymbol{Y}|_L, \boldsymbol{\Lambda}_{\text{twoparty}} ) \\
& = I(W ; \boldsymbol{J}_0, \boldsymbol{J}_1, \boldsymbol{X}|_L, \boldsymbol{\Lambda}_{\text{twoparty}} ) 
\end{align*}

where (a) follows since $W - \boldsymbol{J}_0, \boldsymbol{J}_1, \boldsymbol{X}, \boldsymbol{Y}, L, \boldsymbol{\Lambda}_{\text{twoparty}} - U, \boldsymbol{K}_0, \boldsymbol{K}_1, \tilde{\boldsymbol{\Lambda}}$ is a Markov chain and (b) follows since $W - \boldsymbol{J}_0, \boldsymbol{J}_1, \boldsymbol{X}|_L, \boldsymbol{Y}|_L, \boldsymbol{\Lambda}_{\text{twoparty}} - \boldsymbol{X}, \boldsymbol{Y}, L$ is a Markov chain.

The above term goes to zero due to the OT requirements being satisfied by the two-party OT protocol between Alice and Cathy over $\boldsymbol{X}|_L$.
 
\item To show that (\ref{eqn:ach_2p_indep_5}) holds for $(\mathcal{P}_n)_{n \in \mathbb{N}}$, we proceed as follows:
\begin{align*}
I(U,W ; V_A) & = I(U ; V_A) + I(W ; V_A | U) \\
& = I(U ; V_A) + I(W ; V_A,U) \\
& \leq I(U ; V_A,V_C) + I(W ; V_A,V_B)
\end{align*}

The two terms above go to zero since (\ref{eqn:ach_2p_indep_3}) and (\ref{eqn:ach_2p_indep_4}) hold.

\item To show that (\ref{eqn:ach_2p_indep_6}) holds for $(\mathcal{P}_n)_{n \in \mathbb{N}}$, we proceed as follows:
\begin{align*}
I(\boldsymbol{K}_0,\boldsymbol{K}_1,U,\boldsymbol{J}_{\overline{W}} ; V_C) & = I(\boldsymbol{K}_0,\boldsymbol{K}_1,U ; V_C) + I(\boldsymbol{J}_{\overline{W}} ; V_C | \boldsymbol{K}_0, \boldsymbol{K}_1, U) \\
& = I(\boldsymbol{K}_0,\boldsymbol{K}_1,U ; W, \boldsymbol{Z}, \boldsymbol{\Lambda} ) + I(\boldsymbol{J}_{\overline{W}} ; V_C, \boldsymbol{K}_0, \boldsymbol{K}_1, U) \\
& \stackrel{\text{(a)}}{=} I(\boldsymbol{K}_0,\boldsymbol{K}_1,U ; W, \boldsymbol{Z}, \tilde{\boldsymbol{\Lambda}} ) + I(\boldsymbol{J}_{\overline{W}} ; V_C, \boldsymbol{K}_0, \boldsymbol{K}_1, U) \\
& \stackrel{\text{(b)}}{=} I(\boldsymbol{K}_0,\boldsymbol{K}_1,U ; W, \boldsymbol{Z}, \tilde{\boldsymbol{\Lambda}} ) + I(\boldsymbol{J}_{\overline{W}} ; V_C, U) \\
& \leq I(\boldsymbol{K}_0,\boldsymbol{K}_1,U ; W, \boldsymbol{Z}, \tilde{\boldsymbol{\Lambda}} ) + I(\boldsymbol{J}_{\overline{W}} ; V_C, V_B)
\end{align*}

where (a) follows since $\boldsymbol{K}_0,\boldsymbol{K}_1,U - W, \boldsymbol{Z}, \tilde{\boldsymbol{\Lambda}} - \boldsymbol{\Lambda}_{\text{twoparty}}$ is a Markov chain and (b) follows since $\boldsymbol{J}_{\overline{W}} - V_C,U - \boldsymbol{K}_0,\boldsymbol{K}_1$ is a Markov chain.

The first term above goes to zero for the same reason that (\ref{eqn:ach_2p_wtap_3}) holds for Protocol~\ref{protocol:C2P}. The second term above goes to zero since (\ref{eqn:ach_2p_indep_2}) holds.

\item The proof for showing that (\ref{eqn:ach_2p_indep_7}) holds is similar to that of showing that (\ref{eqn:ach_2p_indep_6}) holds and is, therefore, omitted.

\end{enumerate}  


\subsection{Proof of Lemma~\ref{lem:small_quant_2p_hbc_indep}}
\label{appndx:proof_small_quant_2p_hbc_indep}

For the proof, we use the following lemma:

\begin{lemma}
\label{lem:cond_indep_hbc_indep}
\[  I(\boldsymbol{K}_0, \boldsymbol{K}_1, \boldsymbol{J}_0, \boldsymbol{J}_1;  U, \boldsymbol{Y}, W, \boldsymbol{Z} | \boldsymbol{X},\boldsymbol{\Lambda}) = 0  \]
\end{lemma}

\begin{proof}
Proof is similar to that for Lemma~$6$ of~\cite{ot2007} or Lemma~$2.2$ of~\cite{sec-key1993} and is, therefore, omitted.
\end{proof}

Now,
\begin{align*}
\frac{1}{n}H(\boldsymbol{K}_0, \boldsymbol{K}_1  , \boldsymbol{J}_0, \boldsymbol{J}_1 | \boldsymbol{X},\boldsymbol{\Lambda})  & \leq \frac{1}{n}H(\boldsymbol{K}_0,\boldsymbol{K}_1 | \boldsymbol{X}, \boldsymbol{\Lambda}) +  \frac{1}{n}H(\boldsymbol{J}_0,\boldsymbol{J}_1 | \boldsymbol{X}, \boldsymbol{\Lambda})
\end{align*}

This lemma will be proved if we show that each of the two terms on the RHS above is small. We begin by showing that $(1/n) \cdot H(\boldsymbol{K}_0,\boldsymbol{K}_1 | \boldsymbol{X}, \boldsymbol{\Lambda})$ is small.

For this, we note that Lemma~\ref{lem:cond_indep_hbc_indep} implies $I(\boldsymbol{K}_0,\boldsymbol{K}_1; U | \boldsymbol{X}, \boldsymbol{\Lambda}) = 0$. This further implies :
\begin{align}
H(\boldsymbol{K}_0,\boldsymbol{K}_1 | \boldsymbol{X},\boldsymbol{\Lambda}) & = H(\boldsymbol{K}_0,\boldsymbol{K}_1 | \boldsymbol{X}, \boldsymbol{\Lambda}, U) \nonumber \\
                     & = H(\boldsymbol{K}_U,\boldsymbol{K}_{\overline{U}} | \boldsymbol{X}, \boldsymbol{\Lambda}, U) \nonumber \\
                     & \leq H(\boldsymbol{K}_U | \boldsymbol{X}, \boldsymbol{\Lambda}, U) + H(\boldsymbol{K}_{\overline{U}} | \boldsymbol{X}, \boldsymbol{\Lambda}, U) \label{eqn:quant_sum}
\end{align}

Lemma~\ref{lem:cond_indep_hbc_indep} also implies $I(\boldsymbol{K}_0,\boldsymbol{K}_1 ; \boldsymbol{Y} | \boldsymbol{X},\boldsymbol{\Lambda}, U) = 0$. This, in turn, implies $I(\boldsymbol{K}_U,\boldsymbol{K}_{\overline{U}} ; \boldsymbol{Y} | \boldsymbol{X},\boldsymbol{\Lambda}, U) = 0$. As a result, we get
\begin{equation*}
I(\boldsymbol{K}_U ; \boldsymbol{Y} | \boldsymbol{X},\boldsymbol{\Lambda}, U) = 0
\end{equation*}

Therefore, we have
\begin{align}
H(\boldsymbol{K}_U | \boldsymbol{X},\boldsymbol{\Lambda}, U) & = H(\boldsymbol{K}_U | \boldsymbol{Y}, \boldsymbol{X},\boldsymbol{\Lambda}, U) \nonumber \\
                    & \stackrel{\text{(a)}}{=} H(\boldsymbol{K}_U | \boldsymbol{Y}, \boldsymbol{X},\boldsymbol{\Lambda}, U, \hat{\boldsymbol{K}}_U) \nonumber \\
                    & \leq H(\boldsymbol{K}_U | \hat{\boldsymbol{K}}_U) \nonumber \\
                    & \stackrel{\text{(b)}}{=} o(n) \label{eqn:quant0_small}
\end{align}

where (a) follows since $\hat{\boldsymbol{K}}_U$ is a function of $V_B = (U,\boldsymbol{Y},\boldsymbol{\Lambda})$ and (b) follows from (\ref{eqn:ach_2p_indep_0}) and Fano's inequality. 

Finally, we note that (\ref{eqn:ach_2p_indep_5}) implies that $ I(U ; V_A) \longrightarrow 0$, where $V_A = $ $(\boldsymbol{K}_0, \boldsymbol{K}_1, \boldsymbol{J}_0, \boldsymbol{J}_1, \boldsymbol{X}, \boldsymbol{\Lambda})$. Together with Lemma~\ref{lem:lemma3_ahl_csis}, this implies that
\begin{align*}
H(\boldsymbol{K}_0|  \boldsymbol{X}, \boldsymbol{\Lambda}, U = 0) - H(\boldsymbol{K}_0|  \boldsymbol{X}, \boldsymbol{\Lambda}, U = 1) & = o(n) \\
H(\boldsymbol{K}_1|  \boldsymbol{X}, \boldsymbol{\Lambda}, U = 0) - H(\boldsymbol{K}_1|  \boldsymbol{X}, \boldsymbol{\Lambda}, U = 1) & = o(n)
\end{align*}

We multiply both equations above by $1/2$ and subtract, to get
\begin{equation}
\label{eqn:diff_small}
H(\boldsymbol{K}_U |  \boldsymbol{X}, \boldsymbol{\Lambda}, U) - H(\boldsymbol{K}_{\overline{U}} |  \boldsymbol{X}, \boldsymbol{\Lambda}, U)  = o(n)
\end{equation}

Hence, (\ref{eqn:quant0_small}) and (\ref{eqn:diff_small}) together give :
\begin{equation}
\label{eqn:quant1_small}
H(\boldsymbol{K}_{\overline{U}} |  \boldsymbol{X}, \boldsymbol{\Lambda}, U) = o(n)
\end{equation}

Using (\ref{eqn:quant0_small}) and (\ref{eqn:quant1_small}) in (\ref{eqn:quant_sum}) gives us:
\begin{equation}
\label{eqn:quant_small}
H(\boldsymbol{K}_0,\boldsymbol{K}_1 | \boldsymbol{X},\boldsymbol{\Lambda}) = o(n)
\end{equation}

An exactly analogous argument shows that $H(\boldsymbol{J}_0,\boldsymbol{J}_1 | \boldsymbol{X}, \boldsymbol{\Lambda}) = o(n)$ and, hence, this lemma is proved.


\section{Oblivious transfer over a degraded wiretapped channel: Proof of Lemma~\ref{lem:c1p_ach_degraded}}
\label{appndx:proof_c1p_ach_degraded}

In order to prove Lemma~\ref{lem:c1p_ach_degraded}, we use a sequence $(\mathcal{P}_n)_{n \in \mathbb{N}}$ of Protocol~\ref{protocol:C1P_degraded} instances of rate $(r - 2 \tilde{\delta})$, where $r < \min \left\{ \frac{1}{3} \cdot \epsilon_2(1 - \epsilon_1), \epsilon_1 \right\}$, and show that (\ref{eqn:ach_1p_wtap_0}) - (\ref{eqn:ach_1p_wtap_3}) hold for $(\mathcal{P}_n)_{n \in \mathbb{N}}$. We note that for $\mathcal{P}_n$, the transcript of the public channel is
\begin{equation}
\boldsymbol{\Lambda} = \{ \tilde{G}, \tilde{B}, F_L, F_L(\boldsymbol{X}|_{\tilde{G}_L}) \oplus \boldsymbol{Q}, F_0, F_1, \boldsymbol{K}_0 \oplus F_0(\boldsymbol{X}|_{L_0 \cup \tilde{G}_S}), \boldsymbol{K}_1 \oplus F_1(\boldsymbol{X}|_{L_1 \cup \tilde{G}_S}) \}.
\end{equation}

Let $\Upsilon$ be the indicator random variable for the event that Bob aborts the protocol $\mathcal{P}_n$. Using Chernoff bound, we see that $P[\Upsilon=1] \longrightarrow 0$ exponentially fast as $n \longrightarrow \infty$.

\begin{enumerate}

\item \label{part:1} In order to show (\ref{eqn:ach_1p_wtap_0}) holds for $\{\mathcal{P}_n\}_{n \in \mathbb{N}}$, given that $P[\Upsilon = 1] \longrightarrow 0$, it suffices to show that $P[\hat{\boldsymbol{K}}_U \neq \boldsymbol{K}_U | \Upsilon = 0] \longrightarrow 0$.

When $\Upsilon = 0$, Bob knows $L_U, \boldsymbol{X}|_{L_U}, \tilde{G}_S, \boldsymbol{X}|_{\tilde{G}_S}$. Hence, Bob knows $\boldsymbol{X}|_{L_U \cup \tilde{G}_S}$. As a result, Bob knows the key $F_U(\boldsymbol{X}|_{L_U \cup \tilde{G}_S})$. Hence, Bob can get $\boldsymbol{K}_U$ using $\boldsymbol{K}_U \oplus F_U(\boldsymbol{X}|_{L_U \cup \tilde{G}_S})$ sent by Alice. Thus, $P[\hat{K}_U \neq \boldsymbol{K}_U | \Upsilon = 0] = 0$.

\item \label{part:2} In order to show (\ref{eqn:ach_1p_wtap_1}) holds for $\{\mathcal{P}_n\}_{n \in \mathbb{N}}$, it suffices to show that $I(\boldsymbol{K}_{\overline{U}} ; V_B | \Upsilon = 0) \longrightarrow 0$. All terms and assertions below are conditioned on the event $\Upsilon = 0$, but we suppress this conditioning for ease of writing.
\begin{align*}
I(\boldsymbol{K}_{\overline{U}} ; V_B) &= I(\boldsymbol{K}_{\overline{U}} ; U, \boldsymbol{Y}, \boldsymbol{\Lambda}) \\
&= I(\boldsymbol{K}_{\overline{U}} ; U,\boldsymbol{Y}, \tilde{G}, \tilde{B}, F_L, F_L(\boldsymbol{X}|_{\tilde{G}_L}) \oplus \boldsymbol{Q}, F_0, F_1, \boldsymbol{K}_0 \oplus F_0(\boldsymbol{X}|_{L_0 \cup \tilde{G}_S}), \boldsymbol{K}_1 \oplus F_1(\boldsymbol{X}|_{L_1 \cup \tilde{G}_S})) \\
&= I(\boldsymbol{K}_{\overline{U}} ; U,\boldsymbol{Y}, \tilde{G}, \tilde{B}, F_L, F_L(\boldsymbol{X}|_{\tilde{G}_L}) \oplus \boldsymbol{Q}, F_U, F_{\overline{U}}, \boldsymbol{K}_U \oplus F_U(\boldsymbol{X}|_{L_U \cup \tilde{G}_S}), \boldsymbol{K}_{\overline{U}} \oplus F_{\overline{U}}(\boldsymbol{X}|_{L_{\overline{U}} \cup \tilde{G}_S})) \\
& \stackrel{\text{(a)}}{=} I(\boldsymbol{K}_{\overline{U}} ; U,\boldsymbol{Y}, \tilde{G}, \tilde{B}, F_L, \boldsymbol{Q}, F_U, F_{\overline{U}}, \boldsymbol{K}_U \oplus F_U(\boldsymbol{X}|_{L_U \cup \tilde{G}_S}), \boldsymbol{K}_{\overline{U}} \oplus F_{\overline{U}}(\boldsymbol{X}|_{L_{\overline{U}} \cup \tilde{G}_S})) \\
& \stackrel{\text{(b)}}{=} I(\boldsymbol{K}_{\overline{U}} ; U,\boldsymbol{Y}, \tilde{G}, \tilde{B}, F_L, L_U,L_{\overline{U}}, F_U, F_{\overline{U}}, \boldsymbol{K}_U \oplus F_U(\boldsymbol{X}|_{L_U \cup \tilde{G}_S}), \boldsymbol{K}_{\overline{U}} \oplus F_{\overline{U}}(\boldsymbol{X}|_{L_{\overline{U}} \cup \tilde{G}_S})) \\
& \stackrel{\text{(c)}}{=} I(\boldsymbol{K}_{\overline{U}} ; U,\boldsymbol{Y}, \tilde{G}, \tilde{B}, F_L, L_U,L_{\overline{U}}, F_U, F_{\overline{U}}, \boldsymbol{K}_U, \boldsymbol{K}_{\overline{U}} \oplus F_{\overline{U}}(\boldsymbol{X}|_{L_{\overline{U}} \cup \tilde{G}_S})) \\
& \stackrel{\text{(d)}}{=}  I(\boldsymbol{K}_{\overline{U}} ; U,\boldsymbol{Y}, \tilde{G}, \tilde{B}, F_L, L_U,L_{\overline{U}}, F_U, F_{\overline{U}},  \boldsymbol{K}_{\overline{U}} \oplus F_{\overline{U}}(\boldsymbol{X}|_{L_{\overline{U}} \cup \tilde{G}_S})) \\
&= I(\boldsymbol{K}_{\overline{U}} ;  \boldsymbol{K}_{\overline{U}} \oplus F_{\overline{U}}(\boldsymbol{X}|_{L_{\overline{U}} \cup \tilde{G}_S}) | U,\boldsymbol{Y},\tilde{G},\tilde{B}, F_L, L_U, L_{\overline{U}}, F_U, F_{\overline{U}} ) \\
& = H(\boldsymbol{K}_{\overline{U}} \oplus F_{\overline{U}}(\boldsymbol{X}|_{L_{\overline{U}} \cup \tilde{G}_S}) | U,\boldsymbol{Y},\tilde{G},\tilde{B}, F_L, L_U, L_{\overline{U}}, F_U, F_{\overline{U}}) \\ & \quad - H(F_{\overline{U}}(\boldsymbol{X}|_{L_{\overline{U}} \cup \tilde{G}_S}) | \boldsymbol{K}_{\overline{U}}, U,\boldsymbol{Y},\tilde{G},\tilde{B}, F_L, L_U, L_{\overline{U}}, F_U, F_{\overline{U}}) \\
& \leq |F_{\overline{U}}(\boldsymbol{X}|_{L_{\overline{U}} \cup \tilde{G}_S})| -  H(F_{\overline{U}}(\boldsymbol{X}|_{L_{\overline{U}} \cup \tilde{G}_S}) | \boldsymbol{K}_{\overline{U}}, U,\boldsymbol{Y},\tilde{G},\tilde{B}, F_L, L_U, L_{\overline{U}}, F_U, F_{\overline{U}}) \\
& = n(r - 2 \tilde{\delta}) - H(F_{\overline{U}}(\boldsymbol{X}|_{L_{\overline{U}} \cup \tilde{G}_S}) | \boldsymbol{K}_{\overline{U}}, U,\boldsymbol{Y}, \boldsymbol{Y}|_{\tilde{G}_S}, \tilde{G},\tilde{B}, F_L, L_U, L_{\overline{U}}, F_U, F_{\overline{U}}) \\
& \stackrel{\text{(e)}}{=}  n(r - 2 \tilde{\delta}) - H(F_{\overline{U}}(\boldsymbol{X}|_{L_{\overline{U}} \cup \tilde{G}_S}) |  F_{\overline{U}}, \boldsymbol{Y}|_{\tilde{G}_S}, \tilde{G}_S, L_{\overline{U}}) \\
& = n(r - 2 \tilde{\delta}) - H(F_{\overline{U}}(\boldsymbol{X}|_{L_{\overline{U}} \cup \tilde{G}_S}) |  F_{\overline{U}}, \boldsymbol{X}|_{\tilde{G}_S}, \tilde{G}_S, L_{\overline{U}}) \\
& \stackrel{\text{(f)}}{\leq} n(r - 2 \tilde{\delta}) - \left( n(r - 2 \tilde{\delta}) - \frac{2^{ n(r - 2 \tilde{\delta})  -  n(r - \tilde{\delta})  }}{\ln 2} \right) \\
& = \frac{2^{ -n \tilde{\delta} }}{\ln 2}
\end{align*}

where (a) hold since $F_L(\boldsymbol{X}|_{\tilde{G}_L})$ is a function of $(F_L, \boldsymbol{Y},\tilde{G} )$, (b) holds since ($L_U,L_{\overline{U}}$) is a function of ($U,\boldsymbol{Q},\tilde{G},\tilde{B}$) and $\boldsymbol{Q}$ is a function of ($U,L_U,L_{\overline{U}}$), (c) holds since $F_U(\boldsymbol{X}|_{L_U \cup \tilde{G}_S})$ is a function of ($F_U, \boldsymbol{Y},L_U,\tilde{G}$), (d) holds since $\boldsymbol{K}_U$ is independent of all other variables, (e) holds since $F_{\overline{U}}(\boldsymbol{X}|_{L_{\overline{U}} \cup \tilde{G}_S}) - F_{\overline{U}}, \boldsymbol{Y}|_{\tilde{G}_S}, \tilde{G}_S, L_{\overline{U}} - \boldsymbol{K}_{\overline{U}}, U,\boldsymbol{Y}, \tilde{G},\tilde{B}, F_L, L_U, F_U$ is a Markov chain and (f) holds due to $R(\boldsymbol{X}|_{L_{\overline{U}} \cup \tilde{G}_S} \;\; \mathlarger{\mid} \;\; \boldsymbol{X}|_{\tilde{G}_S} = \boldsymbol{x}|_{\tilde{g}_s}, \tilde{G}_S = \tilde{g}_S, L_{\overline{U}} = l_{\overline{u}}) = |L_{\overline{U}}| = n(r - \tilde{\delta})$ and Lemma~\ref{lem:privacy_amplification}.

\item \label{part:3} In order to show (\ref{eqn:ach_1p_wtap_2}) holds for $\{\mathcal{P}_n\}_{n \in \mathbb{N}}$, it suffices to show that $I(U ; V_A | \Upsilon = 0) \longrightarrow 0$. All terms and assertions below are conditioned on the event $\Upsilon = 0$, but we suppress this conditioning for ease of writing.
\begin{align*}
I(U ; V_A ) &= I(U ;\boldsymbol{K}_0,\boldsymbol{K}_1,\boldsymbol{X},\boldsymbol{\Lambda} ) \\
&= I(U ; \boldsymbol{K}_0,\boldsymbol{K}_1,\boldsymbol{X}, \tilde{G}, \tilde{B}, F_L, F_L(\boldsymbol{X}|_{\tilde{G}_L}) \oplus \boldsymbol{Q}, F_0, F_1, \boldsymbol{K}_0 \oplus F_0(\boldsymbol{X}|_{L_0 \cup \tilde{G}_S}), \boldsymbol{K}_1 \oplus F_1(\boldsymbol{X}|_{L_1 \cup \tilde{G}_S})) \\
& = I(U ; \boldsymbol{K}_0,\boldsymbol{K}_1,\boldsymbol{X}, \tilde{G}, \tilde{B}, F_L, F_L(\boldsymbol{X}|_{\tilde{G}_L}) \oplus \boldsymbol{Q}, F_0, F_1, F_0(\boldsymbol{X}|_{L_0 \cup \tilde{G}_S}), F_1(\boldsymbol{X}|_{L_1 \cup \tilde{G}_S})) \\
& \stackrel{\text{(a)}}{=}  I(U ; \boldsymbol{K}_0,\boldsymbol{K}_1,\boldsymbol{X}, \tilde{G}, \tilde{B}, F_L, \boldsymbol{Q}, F_0, F_1, F_0(\boldsymbol{X}|_{L_0 \cup \tilde{G}_S}), F_1(\boldsymbol{X}|_{L_1 \cup \tilde{G}_S})) \\
& \stackrel{\text{(b)}}{=} I(U ; \boldsymbol{K}_0,\boldsymbol{K}_1,\boldsymbol{X}, \tilde{G}, \tilde{B}, F_L, L_0, L_1, F_0, F_1, F_0(\boldsymbol{X}|_{L_0 \cup \tilde{G}_S}), F_1(\boldsymbol{X}|_{L_1 \cup \tilde{G}_S})) \\
& \stackrel{\text{(c)}}{=} I(U ; \boldsymbol{K}_0,\boldsymbol{K}_1,\boldsymbol{X}, \tilde{G}, \tilde{B}, F_L, L_0, L_1, F_0, F_1) \\
& \stackrel{\text{(d)}}{=} I(U ; L_0,L_1) \\
& \stackrel{\text{(e)}}{=} 0
\end{align*}

where (a) hold since $F_L(\boldsymbol{X}|_{\tilde{G}_L})$ is a function of $(F_L, \boldsymbol{X},\tilde{G})$, (b) holds since $(L_0,L_1)$ is a function of $(\tilde{G},\tilde{B},\boldsymbol{Q})$ and $\boldsymbol{Q}$ is a function of $(L_0,L_1)$, (c) holds since $F_0(\boldsymbol{X}|_{L_0 \cup \tilde{G}_S}),F_1(\boldsymbol{X}|_{L_1 \cup \tilde{G}_S})$ is a function of ($F_0, F_1, \boldsymbol{X},L_0,L_1,\tilde{G})$, (d) holds since $U - L_0,L_1 - \boldsymbol{K}_0,\boldsymbol{K}_1,\boldsymbol{X}, \tilde{G}, \tilde{B}, F_L, F_0, F_1$ is a Markov chain and (e) holds since the channel acts independently on each input bit and since $|L_0| = |L_1|$.

\item \label{part:4} In order to show (\ref{eqn:ach_1p_wtap_3}) holds for $\{\mathcal{P}_n\}_{n \in \mathbb{N}}$, it suffices to show that $I(\boldsymbol{K}_0,\boldsymbol{K}_1,U ; V_E | \Upsilon = 0) \longrightarrow 0$ as $n \longrightarrow \infty$. All terms and assertions below are conditioned on the event $\Upsilon = 0$, but we suppress this conditioning for ease of writing.
\begin{align*}
I(\boldsymbol{K}_0,\boldsymbol{K}_1,U ; V_E) &= I(\boldsymbol{K}_U, \boldsymbol{K}_{\overline{U}}, U ; V_E ) \\
&= I(U ; V_E) + I(\boldsymbol{K}_{\overline{U}} ; V_E |U) + I(\boldsymbol{K}_U ; V_E | U, \boldsymbol{K}_{\overline{U}}) \\
&= I(U ; V_E ) + I(\boldsymbol{K}_{\overline{U}} ; U, V_E ) + I(\boldsymbol{K}_U ; U, \boldsymbol{K}_{\overline{U}}, V_E)
\end{align*}

We look at each of the above three terms separately.

\begin{align*}
I(U ; & V_E ) \\ & = I(U ; \boldsymbol{Z}, \boldsymbol{\Lambda}) \\
&= I(U ; \boldsymbol{Z}, \tilde{G}, \tilde{B}, F_L, F_L(\boldsymbol{X}|_{\tilde{G}_L}) \oplus \boldsymbol{Q}, F_0, F_1, \boldsymbol{K}_0 \oplus F_0(\boldsymbol{X}|_{L_0 \cup \tilde{G}_S}), \boldsymbol{K}_1 \oplus F_1(\boldsymbol{X}|_{L_1 \cup \tilde{G}_S})) \\
&\leq I(U ; \boldsymbol{Z}, \tilde{G}, \tilde{B}, F_L, F_L(\boldsymbol{X}|_{\tilde{G}_L}) \oplus \boldsymbol{Q}, F_0, F_1, \boldsymbol{K}_0, F_0(\boldsymbol{X}|_{L_0 \cup \tilde{G}_S}), \boldsymbol{K}_1, F_1(\boldsymbol{X}|_{L_1 \cup \tilde{G}_S})) \\
& \stackrel{\text{(a)}}{=} I(U ; \boldsymbol{Z}, \tilde{G}, \tilde{B}, F_L, F_L(\boldsymbol{X}|_{\tilde{G}_L}) \oplus \boldsymbol{Q}, F_0, F_1, F_0(\boldsymbol{X}|_{L_0 \cup \tilde{G}_S}), F_1(\boldsymbol{X}|_{L_1 \cup \tilde{G}_S})) \\
& = I(U ; F_L(\boldsymbol{X}|_{\tilde{G}_L}) \oplus \boldsymbol{Q}, F_0(\boldsymbol{X}|_{L_0 \cup \tilde{G}_S}), F_1(\boldsymbol{X}|_{L_1 \cup \tilde{G}_S}) |  \boldsymbol{Z}, \tilde{G}, \tilde{B}, F_L, F_0, F_1) \\
& = H(F_L(\boldsymbol{X}|_{\tilde{G}_L}) \oplus \boldsymbol{Q}, F_0(\boldsymbol{X}|_{L_0 \cup \tilde{G}_S}), F_1(\boldsymbol{X}|_{L_1 \cup \tilde{G}_S}) |  \boldsymbol{Z}, \tilde{G}, \tilde{B}, F_L, F_0, F_1) \\ & \quad -  H(F_L(\boldsymbol{X}|_{\tilde{G}_L}) \oplus \boldsymbol{Q}, F_0(\boldsymbol{X}|_{L_0 \cup \tilde{G}_S}), F_1(\boldsymbol{X}|_{L_1 \cup \tilde{G}_S}) |  U, \boldsymbol{Z}, \tilde{G}, \tilde{B}, F_L, F_0, F_1) \\
& \leq |F_L(\boldsymbol{X}|_{\tilde{G}_L})| + |F_0(\boldsymbol{X}|_{L_0 \cup \tilde{G}_S})| + |F_1(\boldsymbol{X}|_{L_1 \cup \tilde{G}_S})| \\ & \quad -  H(F_L(\boldsymbol{X}|_{\tilde{G}_L}), F_0(\boldsymbol{X}|_{L_0 \cup \tilde{G}_S}), F_1(\boldsymbol{X}|_{L_1 \cup \tilde{G}_S}) | \boldsymbol{Q}, U, \boldsymbol{Z}, \tilde{G}, \tilde{B}, F_L, F_0, F_1) \\
& \stackrel{\text{(b)}}{=} |F_L(\boldsymbol{X}|_{\tilde{G}_L})| + |F_0(\boldsymbol{X}|_{L_0 \cup \tilde{G}_S})| + |F_1(\boldsymbol{X}|_{L_1 \cup \tilde{G}_S})| \\ & \quad -  H( F_L(\boldsymbol{X}|_{\tilde{G}_L}), F_0(\boldsymbol{X}|_{L_0 \cup \tilde{G}_S}), F_1(\boldsymbol{X}|_{L_1 \cup \tilde{G}_S}) | L_U, L_{\overline{U}}, U, \boldsymbol{Z}, \tilde{G}, \tilde{B}, F_L, F_0, F_1) \\
& = |F_L(\boldsymbol{X}|_{\tilde{G}_L})| + |F_0(\boldsymbol{X}|_{L_0 \cup \tilde{G}_S})| + |F_1(\boldsymbol{X}|_{L_1 \cup \tilde{G}_S})| \\ & \quad -  H( F_L(\boldsymbol{X}|_{\tilde{G}_L}), F_U(\boldsymbol{X}|_{L_U \cup \tilde{G}_S}), F_{\overline{U}}(\boldsymbol{X}|_{L_{\overline{U}} \cup \tilde{G}_S}) | L_U, L_{\overline{U}}, U, \boldsymbol{Z}, \tilde{G}, \tilde{B}, F_L, F_U, F_{\overline{U}}) \\
& \stackrel{\text{(c)}}{=}  |F_L(\boldsymbol{X}|_{\tilde{G}_L})| + |F_0(\boldsymbol{X}|_{L_0 \cup \tilde{G}_S})| + |F_1(\boldsymbol{X}|_{L_1 \cup \tilde{G}_S})| \\ & \quad -  H(F_L(\boldsymbol{X}|_{\tilde{G}_L}), F_U(\boldsymbol{X}|_{L_U \cup \tilde{G}_S}), F_{\overline{U}}(\boldsymbol{X}|_{L_{\overline{U}} \cup \tilde{G}_S})  \;\; \mathlarger{\mid} \; \;  L_U, L_{\overline{U}}, \boldsymbol{Z}|_{L_U}, \boldsymbol{Z}|_{\tilde{G}_S}, \boldsymbol{Z}|_{\tilde{G}_L}, \tilde{G}_S, F_L, F_U, F_{\overline{U}}) \\
&= |F_L(\boldsymbol{X}|_{\tilde{G}_L})| + |F_0(\boldsymbol{X}|_{L_0 \cup \tilde{G}_S})| + |F_1(\boldsymbol{X}|_{L_1 \cup \tilde{G}_S})|-  H(F_L(\boldsymbol{X}|_{\tilde{G}_L}) \;\; \mathlarger{\mid} \; \; L_U, L_{\overline{U}}, \boldsymbol{Z}|_{L_U}, \boldsymbol{Z}|_{\tilde{G}_S}, \boldsymbol{Z}|_{\tilde{G}_L}, \tilde{G}_S, F_L, F_U, F_{\overline{U}} ) \\ & \qquad - H(F_U(\boldsymbol{X}|_{L_U \cup \tilde{G}_S}), F_{\overline{U}}(\boldsymbol{X}|_{L_{\overline{U}} \cup \tilde{G}_S}) \;\; \mathlarger{\mid} \; \; F_L(\boldsymbol{X}|_{\tilde{G}_L}), L_U, L_{\overline{U}},  \boldsymbol{Z}|_{L_U}, \boldsymbol{Z}|_{\tilde{G}_S}, \boldsymbol{Z}|_{\tilde{G}_L}, \tilde{G}_S, F_L, F_U, F_{\overline{U}} ) \\
& \stackrel{\text{(d)}}{=}  |F_L(\boldsymbol{X}|_{\tilde{G}_L})|  -  H(F_L(\boldsymbol{X}|_{\tilde{G}_L}) \;\; \mathlarger{\mid} \; \; F_L, \boldsymbol{Z}|_{\tilde{G}_L}) \\
   & \quad + |F_0(\boldsymbol{X}|_{L_0 \cup \tilde{G}_S})| + |F_1(\boldsymbol{X}|_{L_1 \cup \tilde{G}_S})| - H(F_U(\boldsymbol{X}|_{L_U \cup \tilde{G}_S}), F_{\overline{U}}(\boldsymbol{X}|_{L_{\overline{U}} \cup \tilde{G}_S}) \;\; \mathlarger{\mid} \; \;  L_U, L_{\overline{U}}, \boldsymbol{Z}|_{L_U}, \boldsymbol{Z}|_{\tilde{G}_S}, \tilde{G}_S, F_U, F_{\overline{U}}) \\
& = |F_L(\boldsymbol{X}|_{\tilde{G}_L})|  -  H(F_L(\boldsymbol{X}|_{\tilde{G}_L}) \;\; \mathlarger{\mid} \; \; F_L, \boldsymbol{Z}|_{\tilde{G}_L}) + |F_0(\boldsymbol{X}|_{L_0 \cup \tilde{G}_S})| + |F_1(\boldsymbol{X}|_{L_1 \cup \tilde{G}_S})| \\
& \quad - H(F_U(\boldsymbol{X}|_{L_U \cup \tilde{G}_S}) \;\; \mathlarger{\mid} \; \;  L_U, L_{\overline{U}}, \boldsymbol{Z}|_{L_U}, \boldsymbol{Z}|_{\tilde{G}_S}, \tilde{G}_S, F_U, F_{\overline{U}})  \\ & \quad -  H(F_{\overline{U}}(\boldsymbol{X}|_{L_{\overline{U}} \cup \tilde{G}_S}) \;\; \mathlarger{\mid} \; \;  F_U(\boldsymbol{X}|_{L_U \cup \tilde{G}_S}), L_U, L_{\overline{U}}, \boldsymbol{Z}|_{L_U}, \boldsymbol{Z}|_{\tilde{G}_S}, \tilde{G}_S, F_U, F_{\overline{U}}) \\
& \stackrel{\text{(e)}}{=} |F_L(\boldsymbol{X}|_{\tilde{G}_L})|  -  H(F_L(\boldsymbol{X}|_{\tilde{G}_L}) \;\; \mathlarger{\mid} \; \; F_L, \boldsymbol{Z}|_{\tilde{G}_L}) + |F_0(\boldsymbol{X}|_{L_0 \cup \tilde{G}_S})| + |F_1(\boldsymbol{X}|_{L_1 \cup \tilde{G}_S})| \\
& \quad - H(F_U(\boldsymbol{X}|_{L_U \cup \tilde{G}_S}) \;\; \mathlarger{\mid} \; \;  F_U, \boldsymbol{Z}|_{L_U \cup \tilde{G}_S})  -  H(F_{\overline{U}}(\boldsymbol{X}|_{L_{\overline{U}} \cup \tilde{G}_S}) \;\; \mathlarger{\mid} \; \;  F_U(\boldsymbol{X}|_{L_U \cup \tilde{G}_S}), L_U, L_{\overline{U}}, \boldsymbol{Z}|_{L_U}, \boldsymbol{Z}|_{\tilde{G}_S}, \tilde{G}_S, F_U, F_{\overline{U}}) \\
& \leq |F_L(\boldsymbol{X}|_{\tilde{G}_L})|  -  H(F_L(\boldsymbol{X}|_{\tilde{G}_L}) \;\; \mathlarger{\mid} \; \; F_L, \boldsymbol{Z}|_{\tilde{G}_L}) + |F_0(\boldsymbol{X}|_{L_0 \cup \tilde{G}_S})| + |F_1(\boldsymbol{X}|_{L_1 \cup \tilde{G}_S})| \\
& \quad - H(F_U(\boldsymbol{X}|_{L_U \cup \tilde{G}_S}) \;\; \mathlarger{\mid} \; \;  F_U, \boldsymbol{Z}|_{L_U \cup \tilde{G}_S})  -  H(F_{\overline{U}}(\boldsymbol{X}|_{L_{\overline{U}} \cup \tilde{G}_S}) \;\; \mathlarger{\mid} \; \;  \boldsymbol{X}|_{L_U}, \boldsymbol{X}|_{\tilde{G}_S}, L_U, L_{\overline{U}}, \boldsymbol{Z}|_{L_U}, \boldsymbol{Z}|_{\tilde{G}_S}, \tilde{G}_S, F_U, F_{\overline{U}}) \\
& \stackrel{\text{(f)}}{=} |F_L(\boldsymbol{X}|_{\tilde{G}_L})|  -  H(F_L(\boldsymbol{X}|_{\tilde{G}_L}) \;\; \mathlarger{\mid} \; \; F_L, \boldsymbol{Z}|_{\tilde{G}_L}) + |F_0(\boldsymbol{X}|_{L_0 \cup \tilde{G}_S})| + |F_1(\boldsymbol{X}|_{L_1 \cup \tilde{G}_S})| \\
& \quad - H(F_U(\boldsymbol{X}|_{L_U \cup \tilde{G}_S}) \;\; \mathlarger{\mid} \; \;  F_U, \boldsymbol{Z}|_{L_U \cup \tilde{G}_S})  -  H(F_{\overline{U}}(\boldsymbol{X}|_{L_{\overline{U}} \cup \tilde{G}_S}) \;\; \mathlarger{\mid} \; \;  F_{\overline{U}}, \boldsymbol{X}|_{\tilde{G}_S}, \tilde{G}_S, L_{\overline{U}}) \\
& \stackrel{\text{(g)}}{\leq} 2(r - \tilde{\delta})n -  (1 - \xi)\left( 2(r - \tilde{\delta})n - \frac{2^{ 2(r - \tilde{\delta})n  -  2nr }}{\ln 2} \right)+  n(r - 2\tilde{\delta}) + n(r - 2 \tilde{\delta}) \\ & \quad - (1 - \xi)\left(   n(r - 2 \tilde{\delta}) - \frac{ 2^{  n(r - 2 \tilde{\delta}) - n(r - \tilde{\delta})  }  }{\ln 2}  \right)  -  \left(  n(r - 2 \tilde{\delta}) - \frac{ 2^{  n(r - 2 \tilde{\delta}) - n(r - \tilde{\delta})  }  }{\ln 2} \right) \\
& =  2 \xi n (r - \tilde{\delta})  +   (1 - \xi) \cdot \frac{2^{-2\tilde{\delta}n}}{\ln 2} + \xi n (r - 2 \tilde{\delta}) + (2 - \xi) \cdot \frac{2^{- \tilde{\delta}n }}{\ln 2} 
\end{align*}

where (a) hold since $\boldsymbol{K}_0,\boldsymbol{K}_1$ are independent of all the other variables, (b) holds since $(L_U,L_{\overline{U}})$ is a function of $(U,\boldsymbol{Q},\tilde{G},\tilde{B})$  and $\boldsymbol{Q}$ is a function of $(U,L_U,L_{\overline{U}})$, (c) holds since $F_L(\boldsymbol{X}|_{\tilde{G}_L}), F_U(\boldsymbol{X}|_{L_U \cup \tilde{G}_S}), F_{\overline{U}}(\boldsymbol{X}|_{L_{\overline{U}} \cup \tilde{G}_S}) - L_U, L_{\overline{U}}, \boldsymbol{Z}|_{L_U}, \boldsymbol{Z}|_{\tilde{G}_S}, \boldsymbol{Z}|_{\tilde{G}_L},$ $\tilde{G}_S, F_L, F_U, F_{\overline{U}} - U, \boldsymbol{Z}, \tilde{G}, \tilde{B}$ is a Markov chain, (d) holds since $F_L(\boldsymbol{X}|_{\tilde{G}_L}) - F_L, \boldsymbol{Z}|_{\tilde{G}_L} -  L_U, L_{\overline{U}}, \boldsymbol{Z}|_{L_U}, \boldsymbol{Z}|_{\tilde{G}_S}, \tilde{G}_S, F_U, F_{\overline{U}}$ and $F_U(\boldsymbol{X}|_{L_U \cup \tilde{G}_S}), F_{\overline{U}}(\boldsymbol{X}|_{L_{\overline{U}} \cup \tilde{G}_S}) - L_U, L_{\overline{U}}, \boldsymbol{Z}|_{L_U}, \boldsymbol{Z}|_{\tilde{G}_S}, \tilde{G}_S, F_U, F_{\overline{U}} - F_L(\boldsymbol{X}|_{\tilde{G}_L}), \boldsymbol{Z}|_{\tilde{G}_L}, F_L$ are Markov chains, (e) holds since $F_U(\boldsymbol{X}|_{L_U \cup \tilde{G}_S}) - F_U, \boldsymbol{Z}|_{L_U \cup \tilde{G}_S} - L_U, L_{\overline{U}},  \boldsymbol{Z}|_{L_U}, \boldsymbol{Z}|_{\tilde{G}_S}, \tilde{G}_S, F_{\overline{U}}$ is a Markov chain, (f) hold since $F_{\overline{U}}(\boldsymbol{X}|_{L_{\overline{U}} \cup \tilde{G}_S}) - F_{\overline{U}}, \boldsymbol{X}|_{\tilde{G}_S}, \tilde{G}_S, L_{\overline{U}} - \boldsymbol{X}|_{L_U}, L_U, \boldsymbol{Z}|_{L_U}, \boldsymbol{Z}|_{\tilde{G}_S}, F_U$ is a Markov chain and (g) holds for the following reasons:

\begin{itemize}
\item  $R(\boldsymbol{X}|_{\tilde{G}_L} \;\; \mathlarger{\mid} \;\; \boldsymbol{Z}|_{\tilde{G}_L} = \boldsymbol{z}|_{\tilde{g}_L}  ) = \#_e(\boldsymbol{z}|_{\tilde{g}_L})$. Whenever $\#_e(\boldsymbol{z}|_{\tilde{g}_L}) \geq (\epsilon_2 - \delta)|\tilde{G}_L| = 2nr$, by applying Lemma~\ref{lem:privacy_amplification} we get $H(F_L(\boldsymbol{X}|_{\tilde{G}_L}) \;\; \mathlarger{\mid} \; \; F_L, \boldsymbol{Z}|_{\tilde{G}_L} = \boldsymbol{z}|_{\tilde{g}_L}) \geq \left( 2(r - \tilde{\delta})n - \frac{2^{ 2(r - \tilde{\delta})n  -  2nr }}{\ln 2} \right)$. Since by Chernoff's bound, $P[\#_e(\boldsymbol{Z}|_{\tilde{G}_L}) \geq (\epsilon_2 - \delta)|\tilde{G}_L|] \geq 1 - \xi$, where $\xi \longrightarrow 0$ exponentially fast as $n \longrightarrow \infty$, we have $H(F_L(\boldsymbol{X}|_{\tilde{G}_L}) \;\; \mathlarger{\mid} \; \; F_L, \boldsymbol{Z}|_{\tilde{G}_L}) \geq (1 - \xi) \cdot \left( 2(r - \tilde{\delta})n - \frac{2^{ 2(r - \tilde{\delta})n  -  2nr }}{\ln 2} \right)$.

\item Note that $R(\boldsymbol{X}|_{L_U \cup \tilde{G}_S} \;\; \mathlarger{\mid} \;\; \boldsymbol{Z}|_{L_U \cup \tilde{G}_S} = \boldsymbol{z}|_{l_u \cup \tilde{g}_S}) = \#_e(\boldsymbol{z}|_{l_u \cup \tilde{g}_S})$, by Chernoff's bound $P[\#_e(\boldsymbol{Z}|_{L_U \cup \tilde{G}_S} \geq (\epsilon_2 - \delta)(|L_U| + |\tilde{G}_S|)] \geq 1 - \xi$ and  $(\epsilon_2 - \delta)(|L_U| + |\tilde{G}_S|) = n(r - \tilde{\delta})$. By a similar argument as above, we get $H(F_U(\boldsymbol{X}|_{L_U \cup \tilde{G}_S}) \;\; \mathlarger{\mid} \; \;  F_U, \boldsymbol{Z}|_{L_U \cup \tilde{G}_S}) \geq (1 - \xi) \cdot \left(   n(r - 2 \tilde{\delta}) - \frac{ 2^{  n(r - 2 \tilde{\delta}) - n(r - \tilde{\delta})  }  }{\ln 2}  \right) $.

\item $R(\boldsymbol{X}|_{L_{\overline{U}} \cup \tilde{G}_S} \;\; \mathlarger{\mid} \;\; \boldsymbol{X}|_{\tilde{G}_S} = \boldsymbol{x}|_{\tilde{g}_S}, \tilde{G}_S = \tilde{g}_S, L_{\overline{U}} = l_{\overline{u}}) = |L_{\overline{U}}| = n(r - \tilde{\delta})$. Applying Lemma~\ref{lem:privacy_amplification}, we get  \\ $H(F_{\overline{U}}(\boldsymbol{X}|_{L_{\overline{U}} \cup \tilde{G}_S}) \;\; \mathlarger{\mid} \; \;  F_{\overline{U}}, \boldsymbol{X}|_{\tilde{G}_S} = \boldsymbol{x}|_{\tilde{g}_S}, \tilde{G}_S = \tilde{g}_S, L_{\overline{U}} = l_{\overline{u}}) \geq \left(  n(r - 2 \tilde{\delta}) - \frac{ 2^{  n(r - 2 \tilde{\delta}) - n(r - \tilde{\delta})  }  }{\ln 2} \right)$. As a result, $H(F_{\overline{U}}(\boldsymbol{X}|_{L_{\overline{U}} \cup \tilde{G}_S}) \;\; \mathlarger{\mid} \; \;  F_{\overline{U}}, \boldsymbol{X}|_{\tilde{G}_S}, \tilde{G}_S, L_{\overline{U}}) \geq \left(  n(r - 2 \tilde{\delta}) - \frac{ 2^{  n(r - 2 \tilde{\delta}) - n(r - \tilde{\delta})  }  }{\ln 2} \right)$.
\end{itemize}

\begin{align*}
I(\boldsymbol{K}_{\overline{U}} ; U, V_E) &= I(\boldsymbol{K}_{\overline{U}} ; U, \boldsymbol{Z}, \boldsymbol{\Lambda}) \\
&= I(\boldsymbol{K}_{\overline{U}} ; U,  \boldsymbol{Z}, \tilde{G}, \tilde{B}, F_L, F_L(\boldsymbol{X}|_{\tilde{G}_L}) \oplus \boldsymbol{Q}, F_0, F_1, \boldsymbol{K}_0 \oplus F_0(\boldsymbol{X}|_{L_0 \cup \tilde{G}_S}), \boldsymbol{K}_1 \oplus F_1(\boldsymbol{X}|_{L_1 \cup \tilde{G}_S}) ) \\
& = I(\boldsymbol{K}_{\overline{U}} ; U,  \boldsymbol{Z}, \tilde{G}, \tilde{B}, F_L, F_L(\boldsymbol{X}|_{\tilde{G}_L}) \oplus \boldsymbol{Q}, F_U, F_{\overline{U}}, \boldsymbol{K}_U \oplus F_U(\boldsymbol{X}|_{L_U \cup \tilde{G}_S}), \boldsymbol{K}_{\overline{U}} \oplus F_{\overline{U}}(\boldsymbol{X}|_{L_{\overline{U}} \cup \tilde{G}_S}) ) \\
&\leq I( \boldsymbol{K}_{\overline{U}} ; U,  \boldsymbol{Z}, \tilde{G}, \tilde{B}, F_L, F_L(\boldsymbol{X}|_{\tilde{G}_L}) \oplus \boldsymbol{Q}, F_U, F_{\overline{U}}, \boldsymbol{K}_U, F_U(\boldsymbol{X}|_{L_U \cup \tilde{G}_S}), \boldsymbol{K}_{\overline{U}} \oplus F_{\overline{U}}(\boldsymbol{X}|_{L_{\overline{U}} \cup \tilde{G}_S}) ) \\
& \stackrel{\text{(a)}}{=} I( \boldsymbol{K}_{\overline{U}} ; U,  \boldsymbol{Z}, \tilde{G}, \tilde{B}, F_L, F_L(\boldsymbol{X}|_{\tilde{G}_L}) \oplus \boldsymbol{Q}, F_U, F_{\overline{U}}, F_U(\boldsymbol{X}|_{L_U \cup \tilde{G}_S}), \boldsymbol{K}_{\overline{U}} \oplus F_{\overline{U}}(\boldsymbol{X}|_{L_{\overline{U}} \cup \tilde{G}_S}) ) \\
& \leq I( \boldsymbol{K}_{\overline{U}} ; U,  \boldsymbol{Z}, \tilde{G}, \tilde{B}, F_L, F_L(\boldsymbol{X}|_{\tilde{G}_L}), \boldsymbol{Q}, F_U, F_{\overline{U}}, F_U(\boldsymbol{X}|_{L_U \cup \tilde{G}_S}), \boldsymbol{K}_{\overline{U}} \oplus F_{\overline{U}}(\boldsymbol{X}|_{L_{\overline{U}} \cup \tilde{G}_S}) ) \\
& \stackrel{\text{(b)}}{=} I( \boldsymbol{K}_{\overline{U}} ; U,  \boldsymbol{Z}, \tilde{G}, \tilde{B}, F_L, F_L(\boldsymbol{X}|_{\tilde{G}_L}), L_U, L_{\overline{U}}, F_U, F_{\overline{U}}, F_U(\boldsymbol{X}|_{L_U \cup \tilde{G}_S}), \boldsymbol{K}_{\overline{U}} \oplus F_{\overline{U}}(\boldsymbol{X}|_{L_{\overline{U}} \cup \tilde{G}_S}) ) \\
& =  I( \boldsymbol{K}_{\overline{U}} ; \boldsymbol{K}_{\overline{U}} \oplus F_{\overline{U}}(\boldsymbol{X}|_{L_{\overline{U}} \cup \tilde{G}_S})  \;\; \mathlarger{\mid} \;\;     U,  \boldsymbol{Z}, \tilde{G}, \tilde{B}, F_L, F_L(\boldsymbol{X}|_{\tilde{G}_L}), L_U, L_{\overline{U}}, F_U, F_{\overline{U}}, F_U(\boldsymbol{X}|_{L_U \cup \tilde{G}_S})) \\
& = H(\boldsymbol{K}_{\overline{U}} \oplus F_{\overline{U}}(\boldsymbol{X}|_{L_{\overline{U}} \cup \tilde{G}_S})  \;\; \mathlarger{\mid} \;\;     U,  \boldsymbol{Z}, \tilde{G}, \tilde{B}, F_L, F_L(\boldsymbol{X}|_{\tilde{G}_L}), L_U, L_{\overline{U}}, F_U, F_{\overline{U}}, F_U(\boldsymbol{X}|_{L_U \cup \tilde{G}_S})) \\ & \quad - H( F_{\overline{U}}(\boldsymbol{X}|_{L_{\overline{U}} \cup \tilde{G}_S})  \;\; \mathlarger{\mid} \;\;   \boldsymbol{K}_{\overline{U}},  U,  \boldsymbol{Z}, \tilde{G}, \tilde{B}, F_L, F_L(\boldsymbol{X}|_{\tilde{G}_L}), L_U, L_{\overline{U}}, F_U, F_{\overline{U}}, F_U(\boldsymbol{X}|_{L_U \cup \tilde{G}_S})) \\
& \leq |F_{\overline{U}}(\boldsymbol{X}|_{L_{\overline{U}} \cup \tilde{G}_S})| - H( F_{\overline{U}}(\boldsymbol{X}|_{L_{\overline{U}} \cup \tilde{G}_S})  \;\; \mathlarger{\mid} \;\;   \boldsymbol{K}_{\overline{U}},  U,  \boldsymbol{Z}, \tilde{G}, \tilde{B}, F_L, F_L(\boldsymbol{X}|_{\tilde{G}_L}), L_U, L_{\overline{U}}, F_U, F_{\overline{U}}, F_U(\boldsymbol{X}|_{L_U \cup \tilde{G}_S})) \\
& \leq |F_{\overline{U}}(\boldsymbol{X}|_{L_{\overline{U}} \cup \tilde{G}_S})| - H( F_{\overline{U}}(\boldsymbol{X}|_{L_{\overline{U}} \cup \tilde{G}_S})  \;\; \mathlarger{\mid} \;\;   \boldsymbol{K}_{\overline{U}},  U,  \boldsymbol{Z}, \tilde{G}, \tilde{B}, F_L, F_L(\boldsymbol{X}|_{\tilde{G}_L}), L_U, L_{\overline{U}}, F_U, F_{\overline{U}}, \boldsymbol{X}|_{L_U},  \boldsymbol{X}|_{\tilde{G}_S}) \\
& \stackrel{\text{(c)}}{=} |F_{\overline{U}}(\boldsymbol{X}|_{L_{\overline{U}} \cup \tilde{G}_S})| - H( F_{\overline{U}}(\boldsymbol{X}|_{L_{\overline{U}} \cup \tilde{G}_S})  \;\; \mathlarger{\mid} \;\;   F_{\overline{U}}, \boldsymbol{X}|_{\tilde{G}_S}, \tilde{G}_S, L_{\overline{U}}) \\
& = n(r - 2 \tilde{\delta}) -  H( F_{\overline{U}}(\boldsymbol{X}|_{L_{\overline{U}} \cup \tilde{G}_S})  \;\; \mathlarger{\mid} \;\;   F_{\overline{U}}, \boldsymbol{X}|_{\tilde{G}_S}, \tilde{G}_S, L_{\overline{U}}) \\
& \stackrel{\text{(d)}}{\leq}  n(r - 2 \tilde{\delta}) -  \left(  n(r - 2 \tilde{\delta}) -  \frac{2^{ n(r - 2 \tilde{\delta}) -  n(r - \tilde{\delta})   }}{\ln 2} \right) \\
& = \frac{2^{- \tilde{\delta}n}}{\ln 2}
\end{align*}

where (a) holds since $\boldsymbol{K}_U$ is independent of all other variables, (b) holds since $(L_U,L_{\overline{U}})$ is a function of $(U,\boldsymbol{Q},\tilde{G},\tilde{B})$  and $\boldsymbol{Q}$ is a function of $(U,L_U,L_{\overline{U}})$, (c) holds since $F_{\overline{U}}(\boldsymbol{X}|_{L_{\overline{U}} \cup \tilde{G}_S}) -     F_{\overline{U}}, \boldsymbol{X}|_{\tilde{G}_S}, \tilde{G}_S, L_{\overline{U}}  -  \boldsymbol{K}_{\overline{U}},  U,  \boldsymbol{Z}, \tilde{G}, \tilde{B}, F_L, F_L(\boldsymbol{X}|_{\tilde{G}_L}), L_U, F_U,  \boldsymbol{X}|_{L_U}$ is a Markov chain and (d) follows by applying Lemma~\ref{lem:privacy_amplification} knowing that $R(\boldsymbol{X}|_{L_{\overline{U}} \cup \tilde{G}_S} \;\; \mathlarger{\mid} \;\;  \boldsymbol{X}|_{\tilde{G}_S} = \boldsymbol{x}|_{\tilde{x}_S}, \tilde{G}_S = \tilde{g}_S, \tilde{L}_U = \tilde{l}_u) = |L_{\overline{U}}| = n(r - \tilde{\delta})$.

\begin{align*}
I(\boldsymbol{K}_U ; U, \boldsymbol{K}_{\overline{U}}, & V_E) \\ & = I(\boldsymbol{K}_U ; U, \boldsymbol{K}_{\overline{U}}, \boldsymbol{Z}, \boldsymbol{\Lambda}) \\
&= I(\boldsymbol{K}_U ; U, \boldsymbol{K}_{\overline{U}}, \boldsymbol{Z}, \tilde{G}, \tilde{B}, F_L, F_L(\boldsymbol{X}|_{\tilde{G}_L}) \oplus \boldsymbol{Q}, F_0, F_1, \boldsymbol{K}_0 \oplus F_0(\boldsymbol{X}|_{L_0 \cup \tilde{G}_S}), \boldsymbol{K}_1 \oplus F_1(\boldsymbol{X}|_{L_1 \cup \tilde{G}_S})) \\
& = I(\boldsymbol{K}_U ; U, \boldsymbol{K}_{\overline{U}}, \boldsymbol{Z}, \tilde{G}, \tilde{B}, F_L, F_L(\boldsymbol{X}|_{\tilde{G}_L}) \oplus \boldsymbol{Q}, F_U, F_{\overline{U}}, \boldsymbol{K}_U \oplus F_U(\boldsymbol{X}|_{L_U \cup \tilde{G}_S}), \boldsymbol{K}_{\overline{U}} \oplus F_{\overline{U}}(\boldsymbol{X}|_{L_{\overline{U}} \cup \tilde{G}_S})) \\
& =  I(\boldsymbol{K}_U ; U, \boldsymbol{K}_{\overline{U}}, \boldsymbol{Z}, \tilde{G}, \tilde{B}, F_L, F_L(\boldsymbol{X}|_{\tilde{G}_L}) \oplus \boldsymbol{Q}, F_U, F_{\overline{U}}, \boldsymbol{K}_U \oplus F_U(\boldsymbol{X}|_{L_U \cup \tilde{G}_S}), F_{\overline{U}}(\boldsymbol{X}|_{L_{\overline{U}} \cup \tilde{G}_S})) \\
& \stackrel{\text{(a)}}{=}  I(\boldsymbol{K}_U ; U, \boldsymbol{Z}, \tilde{G}, \tilde{B}, F_L, F_L(\boldsymbol{X}|_{\tilde{G}_L}) \oplus \boldsymbol{Q}, F_U, F_{\overline{U}}, \boldsymbol{K}_U \oplus F_U(\boldsymbol{X}|_{L_U \cup \tilde{G}_S}), F_{\overline{U}}(\boldsymbol{X}|_{L_{\overline{U}} \cup \tilde{G}_S})) \\
& \leq I(\boldsymbol{K}_U ; U, \boldsymbol{Z}, \tilde{G}, \tilde{B}, F_L, F_L(\boldsymbol{X}|_{\tilde{G}_L}), \boldsymbol{Q}, F_U, F_{\overline{U}}, \boldsymbol{K}_U \oplus F_U(\boldsymbol{X}|_{L_U \cup \tilde{G}_S}), F_{\overline{U}}(\boldsymbol{X}|_{L_{\overline{U}} \cup \tilde{G}_S})) \\
& \stackrel{\text{(b)}}{=} I(\boldsymbol{K}_U ; U, \boldsymbol{Z}, \tilde{G}, \tilde{B}, F_L, F_L(\boldsymbol{X}|_{\tilde{G}_L}), L_U, L_{\overline{U}}, F_U, F_{\overline{U}}, \boldsymbol{K}_U \oplus F_U(\boldsymbol{X}|_{L_U \cup \tilde{G}_S}), F_{\overline{U}}(\boldsymbol{X}|_{L_{\overline{U}} \cup \tilde{G}_S})) \\
& = I(\boldsymbol{K}_U ; \boldsymbol{K}_U \oplus F_U(\boldsymbol{X}|_{L_U \cup \tilde{G}_S}), F_{\overline{U}}(\boldsymbol{X}|_{L_{\overline{U}} \cup \tilde{G}_S})  |    U, \boldsymbol{Z}, \tilde{G}, \tilde{B}, F_L, F_L(\boldsymbol{X}|_{\tilde{G}_L}), L_U, L_{\overline{U}}, F_U, F_{\overline{U}}   ) \\
& = H(\boldsymbol{K}_U \oplus F_U(\boldsymbol{X}|_{L_U \cup \tilde{G}_S}), F_{\overline{U}}(\boldsymbol{X}|_{L_{\overline{U}} \cup \tilde{G}_S})  |    U, \boldsymbol{Z}, \tilde{G}, \tilde{B}, F_L, F_L(\boldsymbol{X}|_{\tilde{G}_L}), L_U, L_{\overline{U}}, F_U, F_{\overline{U}}) \\ & \quad - H(F_U(\boldsymbol{X}|_{L_U \cup \tilde{G}_S}), F_{\overline{U}}(\boldsymbol{X}|_{L_{\overline{U}} \cup \tilde{G}_S})  |    \boldsymbol{K}_U, U, \boldsymbol{Z}, \tilde{G}, \tilde{B}, F_L, F_L(\boldsymbol{X}|_{\tilde{G}_L}), L_U, L_{\overline{U}}, F_U, F_{\overline{U}}) \\
& \leq |F_U(\boldsymbol{X}|_{L_U \cup \tilde{G}_S})| + |F_{\overline{U}}(\boldsymbol{X}|_{L_{\overline{U}} \cup \tilde{G}_S})| \\ & \quad  -  H(F_U(\boldsymbol{X}|_{L_U \cup \tilde{G}_S}), F_{\overline{U}}(\boldsymbol{X}|_{L_{\overline{U}} \cup \tilde{G}_S})  |    \boldsymbol{K}_U, U, \boldsymbol{Z}, \tilde{G}, \tilde{B}, F_L, F_L(\boldsymbol{X}|_{\tilde{G}_L}), L_U, L_{\overline{U}}, F_U, F_{\overline{U}}) \\
& \stackrel{\text{(c)}}{=} |F_U(\boldsymbol{X}|_{L_U \cup \tilde{G}_S})| + |F_{\overline{U}}(\boldsymbol{X}|_{L_{\overline{U}} \cup \tilde{G}_S})| \\ & \quad  -  H(F_U(\boldsymbol{X}|_{L_U \cup \tilde{G}_S}), F_{\overline{U}}(\boldsymbol{X}|_{L_{\overline{U}} \cup \tilde{G}_S})  \;\;  \mathlarger{\mid}  \;\; \boldsymbol{Z}|_{L_U}, \boldsymbol{Z}|_{L_{\overline{U}}}, \boldsymbol{Z}|_{\tilde{G}_S}, \tilde{G}_S, L_U, L_{\overline{U}}, F_U, F_{\overline{U}}) \\
& = |F_U(\boldsymbol{X}|_{L_U \cup \tilde{G}_S})| + |F_{\overline{U}}(\boldsymbol{X}|_{L_{\overline{U}} \cup \tilde{G}_S})|  \\ & \quad -  H(F_U(\boldsymbol{X}|_{L_U \cup \tilde{G}_S})  \;\;  \mathlarger{\mid}  \;\; \boldsymbol{Z}|_{L_U}, \boldsymbol{Z}|_{L_{\overline{U}}}, \boldsymbol{Z}|_{\tilde{G}_S}, \tilde{G}_S, L_U, L_{\overline{U}}, F_U, F_{\overline{U}})  \\ & \quad -  H(F_{\overline{U}}(\boldsymbol{X}|_{L_{\overline{U}} \cup \tilde{G}_S})  \;\;  \mathlarger{\mid}  \;\; F_U(\boldsymbol{X}|_{L_U \cup \tilde{G}_S}), \boldsymbol{Z}|_{L_U}, \boldsymbol{Z}|_{L_{\overline{U}}}, \boldsymbol{Z}|_{\tilde{G}_S}, \tilde{G}_S, L_U, L_{\overline{U}}, F_U, F_{\overline{U}})  \\
& \stackrel{\text{(d)}}{=} |F_U(\boldsymbol{X}|_{L_U \cup \tilde{G}_S})| + |F_{\overline{U}}(\boldsymbol{X}|_{L_{\overline{U}} \cup \tilde{G}_S})|  -  H(F_U(\boldsymbol{X}|_{L_U \cup \tilde{G}_S})  \;\;  \mathlarger{\mid}  \;\; F_U, \boldsymbol{Z}|_{L_U \cup \tilde{G}_S})  \\ & \quad -  H(F_{\overline{U}}(\boldsymbol{X}|_{L_{\overline{U}} \cup \tilde{G}_S})  \;\;  \mathlarger{\mid}  \;\; F_U(\boldsymbol{X}|_{L_U \cup \tilde{G}_S}), \boldsymbol{Z}|_{L_U}, \boldsymbol{Z}|_{L_{\overline{U}}}, \boldsymbol{Z}|_{\tilde{G}_S}, \tilde{G}_S, L_U, L_{\overline{U}}, F_U, F_{\overline{U}})  \\
& \leq |F_U(\boldsymbol{X}|_{L_U \cup \tilde{G}_S})| + |F_{\overline{U}}(\boldsymbol{X}|_{L_{\overline{U}} \cup \tilde{G}_S})|  -  H(F_U(\boldsymbol{X}|_{L_U \cup \tilde{G}_S})  \;\;  \mathlarger{\mid}  \;\; F_U, \boldsymbol{Z}|_{L_U \cup \tilde{G}_S})  \\ & \quad -  H(F_{\overline{U}}(\boldsymbol{X}|_{L_{\overline{U}} \cup \tilde{G}_S})  \;\;  \mathlarger{\mid}  \;\; \boldsymbol{X}|_{L_U},  \boldsymbol{X}|_{\tilde{G}_S}, \boldsymbol{Z}|_{L_U}, \boldsymbol{Z}|_{L_{\overline{U}}}, \boldsymbol{Z}|_{\tilde{G}_S}, \tilde{G}_S, L_U, L_{\overline{U}}, F_U, F_{\overline{U}})  \\
& \stackrel{\text{(e)}}{=} |F_U(\boldsymbol{X}|_{L_U \cup \tilde{G}_S})| + |F_{\overline{U}}(\boldsymbol{X}|_{L_{\overline{U}} \cup \tilde{G}_S})|  -  H(F_U(\boldsymbol{X}|_{L_U \cup \tilde{G}_S})  \;\;  \mathlarger{\mid}  \;\; F_U, \boldsymbol{Z}|_{L_U \cup \tilde{G}_S})  \\ & \quad -  H(F_{\overline{U}}(\boldsymbol{X}|_{L_{\overline{U}} \cup \tilde{G}_S})  \;\;  \mathlarger{\mid}  \;\; F_{\overline{U}}, \boldsymbol{X}|_{\tilde{G}_S}, \tilde{G}_S, L_{\overline{U}})  \\
& = n(r - 2\tilde{\delta}) + n(r - 2\tilde{\delta}) - H(F_U(\boldsymbol{X}|_{L_U \cup \tilde{G}_S})  \;\;  \mathlarger{\mid}  \;\; F_U, \boldsymbol{Z}|_{L_U \cup \tilde{G}_S})  -  H(F_{\overline{U}}(\boldsymbol{X}|_{L_{\overline{U}} \cup \tilde{G}_S})  \;\;  \mathlarger{\mid}  \;\; F_{\overline{U}}, \boldsymbol{X}|_{\tilde{G}_S}, \tilde{G}_S, L_{\overline{U}})  \\
& \stackrel{\text{(f)}}{\leq} n(r - 2\tilde{\delta}) + n(r - 2\tilde{\delta}) - (1 - \xi)\left(  n(r - 2\tilde{\delta}) - \frac{2^{ n(r - 2\tilde{\delta}) - n(r - \tilde{\delta})   }}{\ln 2}  \right)   - \left(  n(r - 2\tilde{\delta}) - \frac{2^{ n(r - 2\tilde{\delta}) - n(r - \tilde{\delta})   }}{\ln 2}  \right) \\
& = \xi n (r - 2 \tilde{\delta}) + (2 - \tilde{\delta}) \cdot \frac{2^{- \tilde{\delta}n}}{\ln 2} 
\end{align*} 

where (a) hold since $\boldsymbol{K}_{\overline{U}}$ is independent of all other variables, (b) holds since $(L_U,L_{\overline{U}})$ is a function of $(U,\boldsymbol{Q},\tilde{G},\tilde{B})$  and $\boldsymbol{Q}$ is a function of $(U,L_U,L_{\overline{U}})$, (c) holds since $F_U(\boldsymbol{X}|_{L_U \cup \tilde{G}_S}), F_{\overline{U}}(\boldsymbol{X}|_{L_{\overline{U}} \cup \tilde{G}_S}) - \boldsymbol{Z}|_{L_U}, \boldsymbol{Z}|_{L_{\overline{U}}}, \boldsymbol{Z}|_{\tilde{G}_S}, \tilde{G}_S, L_U, L_{\overline{U}}, F_U, F_{\overline{U}} - \boldsymbol{K}_U, U, \boldsymbol{Z}, \tilde{G}, \tilde{B}, F_L, F_L(\boldsymbol{X}|_{\tilde{G}_L})$ is a Markov chain, (d) holds since $F_U(\boldsymbol{X}|_{L_U \cup \tilde{G}_S}) - F_U, \boldsymbol{Z}|_{L_U \cup \tilde{G}_S} - \boldsymbol{Z}|_{L_U}, \boldsymbol{Z}|_{L_{\overline{U}}}, \boldsymbol{Z}|_{\tilde{G}_S}, \tilde{G}_S, L_U,$ $L_{\overline{U}}, F_{\overline{U}}$ is a Markov chain, (e) holds since $F_{\overline{U}}(\boldsymbol{X}|_{L_{\overline{U}} \cup \tilde{G}_S}) - F_{\overline{U}}, \boldsymbol{X}|_{\tilde{G}_S}, \tilde{G}_S, L_{\overline{U}} - \boldsymbol{X}|_{L_U}, \boldsymbol{Z}|_{L_U}, \boldsymbol{Z}|_{L_{\overline{U}}}, \boldsymbol{Z}|_{\tilde{G}_S}, L_U, F_U$ is a Markov chain and (f) holds for the following reasons:

\begin{itemize}
 \item $R(\boldsymbol{X}|_{L_U \cup \tilde{G}_S} \;\; \mathlarger{\mid} \;\; \boldsymbol{Z}|_{L_U \cup \tilde{G}_S} = \boldsymbol{z}|_{l_u \cup \tilde{g}_S}) = \#_e(\boldsymbol{z}|_{l_u \cup \tilde{g}_S})$. Whenever $\#_e(\boldsymbol{z}|_{l_u \cup \tilde{g}_S}) \geq (\epsilon_2 - \delta)(|L_U| + |\tilde{G}_S|) = n(r - \tilde{\delta})$, then by applying Lemma~\ref{lem:privacy_amplification} we get $H(F_U(\boldsymbol{X}|_{L_U \cup \tilde{G}_S})  \;\;  \mathlarger{\mid}  \;\; F_U, \boldsymbol{Z}|_{L_U \cup \tilde{G}_S} = \boldsymbol{z}|_{l_u \cup \tilde{g}_S}) \geq \left(  n(r - 2\tilde{\delta}) - \frac{2^{ n(r - 2\tilde{\delta}) - n(r - \tilde{\delta})   }}{\ln 2}  \right)$. By Chernoff's bound, $P[ \#_e(\boldsymbol{Z}|_{L_U \cup \tilde{G}_S} ) \geq  (\epsilon_2 - \delta)(|L_U| + |\tilde{G}_S|)] \geq 1 - \xi$, where $\xi \longrightarrow 0$ exponentially fast as $n \longrightarrow \infty$. As a result, $H(F_U(\boldsymbol{X}|_{L_U \cup \tilde{G}_S})  \;\;  \mathlarger{\mid}  \;\; F_U, \boldsymbol{Z}|_{L_U \cup \tilde{G}_S})  \geq (1 - \xi) \left(  n(r - 2\tilde{\delta}) - \frac{2^{ n(r - 2\tilde{\delta}) - n(r - \tilde{\delta})   }}{\ln 2}  \right)$.

\item $R(\boldsymbol{X}|_{L_{\overline{U}} \cup \tilde{G}_S} \;\; \mathlarger{\mid} \;\; \boldsymbol{X}|_{\tilde{G}_S} = \boldsymbol{x}|_{\tilde{g}_S}, \tilde{G}_S = \tilde{g}_S, L_{\overline{U}} = l_{\overline{u}}) = |L_{\overline{U}}| = n(r - \tilde{\delta})$. Applying Lemma~\ref{lem:privacy_amplification} we get:
\[ H(F_{\overline{U}}(\boldsymbol{X}|_{L_{\overline{U}} \cup \tilde{G}_S})  \;\;  \mathlarger{\mid}  \;\; F_{\overline{U}}, \boldsymbol{X}|_{\tilde{G}_S} = \boldsymbol{x}|_{\tilde{g}_S}, \tilde{G}_S = \tilde{g}_S, L_{\overline{U}} = l_{\overline{u}}) \geq \left(  n(r - 2\tilde{\delta}) - \frac{2^{ n(r - 2\tilde{\delta}) - n(r - \tilde{\delta})   }}{\ln 2}  \right). \]
 As a result, $H(F_{\overline{U}}(\boldsymbol{X}|_{L_{\overline{U}} \cup \tilde{G}_S})  \;\;  \mathlarger{\mid}  \;\; F_{\overline{U}}, \boldsymbol{X}|_{\tilde{G}_S}, \tilde{G}_S, L_{\overline{U}}) \geq \left(  n(r - 2\tilde{\delta}) - \frac{2^{ n(r - 2\tilde{\delta}) - n(r - \tilde{\delta})   }}{\ln 2}  \right)$.
\end{itemize}

\end{enumerate}


\end{document}